\documentclass[ejs]{imsart}
\RequirePackage{amsthm,amsmath,amsfonts,amssymb}
\RequirePackage[numbers]{natbib}
\RequirePackage[colorlinks,citecolor=blue,urlcolor=blue]{hyperref}
\RequirePackage{graphicx}

\startlocaldefs

\RequirePackage[OT1]{fontenc}
\RequirePackage{amsthm,amsmath}
\RequirePackage[numbers]{natbib}
\RequirePackage[colorlinks,citecolor=blue,urlcolor=blue]{hyperref}
\usepackage{graphicx}

\pubyear{2020}
\volume{0}
\issue{0}
\firstpage{1}
\lastpage{8}
\arxiv{2010.00000}

\startlocaldefs
\numberwithin{equation}{section}
\theoremstyle{plain}
\newtheorem{thm}{Theorem}[section]
\newtheorem{assumption}{Assumption}[section]
\newtheorem{definition}{Definition}[section]
\newtheorem{remark}{Remark}[section]
\newtheorem{lemma}{Lemma}[section]

\endlocaldefs

\begin{document}

\begin{frontmatter}
\title{Unadjusted Langevin algorithm for non-convex weakly smooth potentials}
\runtitle{Non-convex weakly smooth sampling}

\begin{aug}
\author[A]{\fnms{Dao} \snm{Nguyen}\ead[label=e1]{dxnguyen@go.olemiss.edu}},
\author[B]{\fnms{Xin} \snm{Dang}\ead[label=e2]{xdang@olemiss.edu}}
\and
\author[C]{\fnms{Yixin} \snm{Chen}\ead[label=e3]{ychen@cs.olemiss.edu}}
\address[A]{Department of Mathematics,
University of Mississippi.
\printead{e1}}

\address[B]{Department of Mathematics,
University of Mississippi.
\printead{e2}}

\address[C]{Department of Computer Science,
University of Mississippi.
\printead{e3}}
\runauthor{D. Nguyen et al.}
\affiliation{University of Mississippi}

\end{aug}

\begin{abstract}
Discretization of continuous-time diffusion processes is a widely
recognized method for sampling. However, the canonical Euler Maruyama
discretization of the Langevin diffusion process, referred as Unadjusted
Langevin Algorithm (ULA), studied mostly in the context of smooth
(gradient Lipschitz) and strongly log-concave densities, is a considerable
hindrance for its deployment in many sciences, including statistics
and machine learning. In this paper, we establish several theoretical
contributions to the literature on such sampling methods for non-convex
distributions. Particularly, we introduce a new mixture weakly smooth
condition, under which we prove that ULA will converge with additional
log-Sobolev inequality. We also show that ULA for smoothing potential
will converge in $L_{2}$-Wasserstein distance. Moreover, using convexification
of nonconvex domain \citep{ma2019sampling} in combination with regularization,
we establish the convergence in Kullback-Leibler (KL) divergence with
the number of iterations to reach $\epsilon$-neighborhood of a target
distribution in only polynomial dependence on the dimension. We relax
the conditions of \citep{vempala2019rapid} and prove convergence
guarantees under isoperimetry, and non-strongly convex at infinity. 
\end{abstract}

\begin{keyword}[class=MSC]
\kwd{Langevin Monte Carlo}
\kwd{non-convex sampling}
\kwd{Kullback-Leibler divergence}
\kwd{regularization}
\kwd{weakly smooth}
\end{keyword}
\tableofcontents
\end{frontmatter}

\section{Introduction}

\label{intro} Over the last decade, Bayesian inference has become
one of the most prevalent inferring instruments for a variety of disciplines,
including the computational statistics and statistical learning \citep{robert2013monte}.
In general, Bayesian inference seeks to generate samples of the posterior
distribution of the form: 
\begin{equation}
\rho(\mathrm{x})=\mathrm{e}^{-U(x)}/\int_{R^{d}}\mathrm{e}^{-U(y)}\mathrm{d}y,
\end{equation}
where the function $U(\mathrm{x})$, also known as the potential function,
is often convex. The most conventional approaches, random walks Metropolis
Hasting \citep{robert2013monte}, often struggle to select a proper
proposal distribution for sampling. As a result, it has been proposed
to consider continuous dynamics which inherently leaves the objective
distribution $\rho$ invariant. Probably one of the most well-known
example of these continuous dynamic applications is the over-damped
Langevin diffusion \citep{dalalyan2019user} associated with $U$,
\begin{align}
dX_{t}=-\nabla U(X_{t})dt+\sqrt{2}dB_{t},\label{cont}
\end{align}
where $B_{t}$ is a $d$-dimensional Brownian motion and its Euler-Maruyama
discretization hinges on the following update rule: 
\begin{equation}
\mathrm{x}_{k+1}=\mathrm{x}_{k}-\eta_{k}\nabla U(\mathrm{x}_{k})+\sqrt{2\eta}\xi_{k},
\end{equation}
where $(\eta_{k})_{k\geq1}$ is a sequence of step sizes, which can
be kept constant or decreasing to $0$, and $\xi_{k}\sim{N}(0,\ I_{d})$
($I_{d}$ denotes identity matrix dimension $d$), are independent
Gaussian random vectors. It can be seen that Euler-Maruyama discretization,
also known as Langevin Monte Carlo (LMC) algorithm, does not involve
knowledge of $U$, but gradient of $U$ instead, which makes it ideally
applicable where we typically only know $U$ up to a normalizing constant.
Owing to its simplicity, efficiency, and well understood properties,
there are various applications using LMC \citep{welling2011bayesian,dalalyan2017further,raginsky2017non, xu2018global,li2019stochastic}.
Much of the theory of convergence of sampling used to focus on asymptotic
convergence, failing to provide a detailed study of dimension dependence.
Recently, there is a surge of interests in non-asymptotic rates of
convergence, which include dimension dependence, especially polynomial
dependence on the dimension of target distribution; see, e.g., \citep{cheng2018convergence,dalalyan2019user,durmus2019analysis,dwivedi2019log,bernton2018langevin,durmus2017convergence,chen2019optimal,lee2018convergence,mangoubi2017rapid,mangoubi2018dimensionally}.
However, there is a critical gap in the theory of discretization of
an underlying stochastic differential equation (SDE) to the broad
spectrum of applications in statistical inference. In particular,
the application of techniques from SDEs traditionally requires $U(\mathrm{x})$
to have Lipschitz-continuous gradients. This requirement prohibits
many typical utilizations \citep{durmus2018efficient}. \citep{chatterji2019langevin}
has recently established an original approach to deal with weakly
smooth (possibly non-smooth) potential problems through smoothing.
Their technique rests on results obtained from the optimization community,
perturbing a gradient evaluating point by a Gaussian. However, \citep{chatterji2019langevin}
analyzes over-damped Langevin diffusion in the contexts of convex
potential functions while many applications involve sampling in high
dimensional spaces have non convex settings. In another research,
\citep{erdogdu2020convergence} proposes a very beautiful result using
tail growth for weakly smooth and weakly dissipative potentials. By
using degenerated convex and modified log Sobolev inequality, they
prove that LMC gets $\epsilon$-neighborhood of a target distribution
in KL divergence with the convergence rate of $\tilde{O}(d^{\frac{1}{\alpha}+\frac{1+\alpha}{\alpha}(\frac{2}{\beta}-\mathrm{1}_{\{\beta\neq1\}})}\epsilon^{\frac{-1}{\alpha}})$
where $\alpha$ and $\beta$ are degrees of weakly smooth and dissipative
defined in the next section.

Unfortunately, (weakly) smooth conditions typically can not cover
a mixture of distributions with different tail growth behaviors, which
prohibit a large range of applications. Therefore, we first introduce
an $\alpha$-mixture weakly smooth condition to overcome the limitation
of the current weakly smooth condition. Under our novel condition
and log-Sobolev inequality, we will recover the ULA convergence results
\citep{vempala2019rapid}. In addition, we also show that ULA for
smoothing potential will converge in $L_{2}$-Wasserstein distance.
Since log-Sobolev inequality is preserved under bounded perturbation,
we will extend the results based on convexification of a non-convex
domain \citep{ma2019sampling}. Our contributions can be outlined
as follows.

First, for a potential function $U$, which satisfies an $\alpha$-mixture
weakly smooth, and $\gamma$-log Sobolev inequality, we prove that
ULA achieves the convergence rate of 
\begin{equation}
\tilde{O}\left(d^{\frac{1}{\alpha}}\gamma^{\frac{-1}{\alpha}}\epsilon^{\frac{-1}{\alpha}}\right)
\end{equation}
in KL-divergence.

Second, our convergence results also cover sampling from non-convex
potentials, satisfying $\alpha$-mixture weakly smooth, $2$-dissipative
and $\gamma$-Poincar\'{e} or $\gamma$-Talagrand with convergence rate
in KL divergence of 
\begin{equation}
\tilde{O}\left(d^{\frac{2}{\alpha}+1}\gamma^{\left(\frac{-1}{\alpha}-1\right)}\epsilon^{\frac{-1}{\alpha}}\right).
\end{equation}

Third, we further investigate the case of $\alpha$-mixture weakly
smooth, $2$-dissipative and non strongly convex outside the ball
of radius $R$ and obtain the convergence rate of 
\begin{equation}
\tilde{O}\left(d^{1+\frac{2}{\alpha}}\gamma^{\left(\frac{-1}{\alpha}-1\right)}e^{5\left(2\sum_{i}L_{i}R^{1+\alpha_{i}}+4L_{N}R^{2}+4\max\left\{ L_{i}\right\} R^{1+\alpha}\right)\left(\frac{1}{\alpha}+1\right)}\epsilon^{\frac{-1}{\alpha}}\right).
\end{equation}

Finally, our convergence results remain valid under finite perturbations,
indicating that it is applicable to an even larger class of potentials.
Last but not least, convergence in KL divergence implies convergence
in total variation and in $L_{2}$-Wasserstein metrics based on Pinsker
inequalities, which in turn gives convergence rates of $O(\cdot\epsilon^{\frac{-2}{\alpha})}$
and $O(\cdot\epsilon^{\frac{-2\beta}{\alpha}}d^{\frac{2}{\alpha}})$
in place of $O(\cdot\epsilon^{\frac{-1}{\alpha})}$ in each case above,
respectively for total variation and $L_{2}$-Wasserstein metrics.

The rest of the paper is organized as follows. Section 2 sets out
the notation, definition and smoothing properties necessary to give
our main results in section 3. Section 4 extends the convexification
of non convex domain of \citep{ma2019sampling} for strongly convex
outside the ball to non-strongly convex outside the ball, and employs
this outcome in combination with regularization to obtain convergence
in KL divergence for potentials which satisfy log Sobolev, Talagrand,
Poincar\'{e} inequalities, or non-strongly convex at infinity while Section
5 presents our conclusions. 

\section{Preliminaries}

\label{sec:problem}\label{sec:papproach}

This section provides the notational conventions, assumptions, and
some auxiliary results used in this paper. We let $\left|s\right|$,
for a real number $s\in\mathbb{R}$, denote its absolute value and
use $\left\langle \ ,\ \right\rangle $ to specify inner products.
We use $\Vert x\Vert_{p}$ to denote the $p$-norm of a vector $x\in\mathbb{R}^{d}$
and throughout the paper, we drop the subscript and just write $\Vert x\Vert\stackrel{\triangle}{=}\Vert x\Vert_{2}$
whenever $p=2$. For a function $U$ :$\mathbb{R}^{d}\rightarrow\mathbb{R}$,
which is twice differentiable, we use $\nabla U(x)$ and $\nabla^{2}U(x)$
to denote correspondingly the gradient and the Hessian of $U$ with
respect to $x$. We use $A\succeq B$ if $A-B$ is a positive semi-definite
matrix. We use big-oh notation $O$ in the following sense that if
$f(x)={\displaystyle O(g(x))}$ implies $\lim_{x\rightarrow\infty}\sup\frac{f(x)}{g(x)}<\infty$
and $\tilde{O}$ suppresses the logarithmic factors.

\subsection{Assumptions on the potential $U$}

The objective of this paper is to sample from a distribution $\pi\propto\exp(-U(x))$,
where $x\in\mathbb{R}^{d}$. While sampling from the exact distribution
$\pi(x)$ is generally computationally demanding, it is largely adequate
to sample from an approximated distribution $\tilde{\pi}(x)$ which
is in the vicinity of $\pi(x)$ by some distances. In this paper,
we suppose some of the following conditions hold:

\begin{assumption} \label{A0} ($\alpha$-mixture weakly smooth)
There exist $0<\alpha=\alpha_{1}<...<\alpha_{N}\leq1$, $i=1,..,N$
$0<L_{i}<\infty$ so that $\forall x,\ y\in\mathbb{R}^{d}$, we obtain
$\left\Vert \nabla U(x)-\nabla U(y)\right\Vert \leq\sum_{i=1}^{N}L_{i}\left\Vert x-y\right\Vert ^{\alpha_{i}}$
where $\nabla U(x)$ represents an arbitrary sub-gradient of $U$
at $x$. \end{assumption}

\begin{assumption} \label{A1} ($\left(\alpha,\ell\right)-$weakly
smooth) There exist $0\leq\ell$, $0<L<\infty$ and $\alpha\in[0,1]$
so that $\forall x,\ y\in\mathbb{R}^{d}$, we obtain 
\[
\left\Vert \nabla U(x)-\nabla U(y)\right\Vert \leq L\left(1+\left\Vert x-y\right\Vert ^{\ell}\right)\left\Vert x-y\right\Vert ^{\alpha},
\]
where $\nabla U(x)$ represents an arbitrary sub-gradient of $U$
at $x$. \end{assumption}

\begin{assumption} \label{A2} ($\left(\mu,\theta\right)$-degenerated
convex outside the ball) There exists some $\mu>0,$ $1\geq\theta\geq0$
so that for every $\left\Vert x\right\Vert \geq R,$ the potential
function $U(x)$ satisfies $\nabla^{2}U(x)\succeq m\left(\left\Vert x\right\Vert \right)I_{d}$
where $m\left(r\right)=\mu\left(1+r^{2}\right)^{-\frac{\theta}{2}}.$
\end{assumption}

\begin{assumption}

\label{A3} ($\beta-$dissipativity). There exists $\beta\geq1$,
$a$, $b>0$ such that $\forall x\in\mathbb{R}^{d}$, $\left\langle \nabla U(x),x\right\rangle \geq a\left\Vert x\right\Vert ^{\beta}-b.$

\end{assumption} \begin{assumption} \label{A4} ($LSI\left(\gamma\right)$)
There exists some $\gamma>0,$ so that for all probability distribution
$p\left(x\right)$ absolutely continuous $w.r.t.\ \pi\left(x\right)$,
$H({\displaystyle p|\pi)\leq\frac{1}{2\gamma}I(p|\pi)}.$

\end{assumption} \begin{assumption} \label{A5} ($PI\left(\gamma\right)$)
There exists some $\gamma>0,$ so that for all smooth function $g\colon\mathbb{R}^{d}\to\mathbb{R}$,
$Var_{\pi}(g)\le\frac{1}{\gamma}E_{\pi}\left[\left\Vert \nabla g\right\Vert ^{2}\right]$
where $Var_{\pi}(g)=E_{\pi}[g^{2}]-E_{\pi}[g]^{2}$ is the variance
of $g$ under $\pi$. \end{assumption} \begin{assumption} \label{A6}
(non-strongly convex outside the ball) For every $\left\Vert x\right\Vert \geq R$,
the potential function $U(x)$ is positive semi-definite, that is
for every $y\in\mathbb{R}^{d}$, ${\displaystyle \left\langle y,\nabla^{2}U(x)\ y\right\rangle \geq0}.$

\end{assumption}\begin{assumption} \label{A7} The function $U(x)$
has stationary point at zero $\nabla U(0)=0.$

\end{assumption} \begin{remark} Assumption \ref{A7} is imposed
without loss of generality. Condition \ref{A0} often holds for a
mixture of distribution with different tail growth behaviors. It is
straightforward to generalize condition \ref{A0} from the mixture
of two distribution with the same constant $L$, so we will consider
condition \ref{A1} to simplify the proof while optimize the convergence
rate. Condition \ref{A1} is an extension of $\alpha$-weakly smooth
or $\alpha-$Holder continuity of the (sub)gradients of $U$ (that
is when $\ell=0$, we recover normal $\alpha$-weakly smooth).

\end{remark}

\subsection{Smoothing using $p$-generalized Gaussian smoothing}

A feature that follows straightforwardly from Assumption \ref{A0}
is that for $\forall x,\ y\in\mathrm{\mathbb{R}}^{d}$:

\begin{lemma} If potential $U:\mathbb{R}^{d}\rightarrow\mathbb{R}$
satisfies an $\alpha$-mixture weakly smooth for some $0<\alpha=\alpha_{1}<...<\alpha_{N}\leq1$,
$i=1,..,N$ $0<L_{i}<\infty$, then:

\begin{equation}
U(y)\leq U(x)+\left\langle \nabla U(x),\ y-x\right\rangle +\sum_{i}\frac{L_{i}}{1+\alpha_{i}}\Vert y-x\Vert^{1+\alpha_{i}}.\label{eq:4-1}
\end{equation}
In particular, if the potential $U:\mathbb{R}^{d}\rightarrow\mathbb{R}$
satisfies $\left(\alpha,\ell\right)-$weakly smooth for some $\alpha+\ell\leq1$
and $\alpha\in(0,1]$, then:

\begin{equation}
U(y)\leq U(x)+\left\langle \nabla U(x),\ y-x\right\rangle +\frac{L}{1+\alpha}\Vert y-x\Vert^{1+\alpha}+\frac{L}{1+\ell+\alpha}\Vert y-x\Vert^{1+\ell+\alpha}.\label{eq:4}
\end{equation}

\end{lemma} \begin{proof} See Appendix \ref{Asmooth} \end{proof}

Here, to deal with the heavy tail behavior of some distributions in
the mixture, we use $p$-generalized Gaussian smoothing. Particularly,
for $\mu\geq0$, $p$-generalized Gaussian smoothing $U_{\mu}$ of
$U$ is defined as: 
\[
U_{\mu}(y):=\mathrm{E}_{\xi}[U(y+\mu\xi)]=\frac{1}{\kappa}\int_{\mathbb{R}^{d}}U(y+\mu\xi)e^{-\left\Vert \xi\right\Vert _{p}^{p}/p}d\xi,
\]
where $\kappa\stackrel{_{def}}{=}\int_{\mathbb{R}^{d}}e^{-\left\Vert \xi\right\Vert _{p}^{p}/p}d\xi=\frac{2^{d}\Gamma^{d}(\frac{1}{p})}{p^{d-\frac{d}{p}}}$
and $\xi\sim N_{p}(0,I_{d\times d})$ (the $p$-generalized Gaussian
distribution). The rationale for taking into account the $p$-generalized
Gaussian smoothing $U_{\mu}$ rather than $U$ is that it typically
benefits from superior smoothness properties. In particular, $U_{\mu}$
is smooth albeit $U$ is not. In addition, $p$-generalized Gaussian
smoothing is more generalized than Gaussian smoothing in the sense
that it contains all normal distribution when $p=2$ and all Laplace
distribution when $p=1$. This family of distributions allows for
tails that are either heavier or lighter than normal and in the limit
as well as containing all the continuous uniform distribution. More
importantly, we prove that a smoothing potential $U_{\mu}(x)$ is
actually smooth (gradient Lipschitz). This property is novel and potentially
useful in the optimization or sampling process, especially when the
potential exhibits some sort of weakly smooth behaviors. Here to simplify
the proof, we consider $p\in\mathbb{R}^{+},$ $2\geq p\geq1$ and
some primary features of $U_{\mu}$ based on adapting those results
of \citep{nesterov2017random}.

\begin{lemma} \label{2.1} If potential $U:\mathbb{R}^{d}\rightarrow\mathbb{R}$
satisfies an $\alpha$-mixture weakly smooth for some $0<\alpha=\alpha_{1}<...<\alpha_{N}\leq1$,
$i=1,..,N$ $0<L_{i}<\infty$, then:

(i) $\forall x\in\mathbb{R}^{d}$ : $\left|U_{\mu}(x)-U(x)\right|{\displaystyle \leq\sum_{i}L_{i}\mu^{1+\alpha_{i}}d^{\frac{1+\alpha_{i}}{p}},}$

(ii) $\forall x\in\mathbb{R}^{d}$: ${\displaystyle \left\Vert \nabla U_{\mu}(x)-\nabla U(x)\right\Vert \leq\sum_{i}L_{i}\mu^{\alpha_{i}}d^{\frac{3}{p}}},$

(iii) $\forall x,\ y\in\mathbb{R}^{d}$: ${\displaystyle \left\Vert \nabla U_{\mu}(y)-\nabla U_{\mu}(x)\right\Vert \leq\sum_{i}\frac{L_{i}}{\mu^{1-\alpha_{i}}}d^{\frac{2}{p}}\left\Vert y-x\right\Vert .}$

In particular, if the potential $U:\mathbb{R}^{d}\rightarrow\mathbb{R}$
satisfies $\left(\alpha,\ell\right)-$weakly smooth for some $\alpha+\ell\leq1$
and $\alpha\in[0,1]$, then:

(i) $\forall x\in\mathbb{R}^{d}$ : $\left|U_{\mu}(x)-U(x)\right|{\displaystyle \leq2L\mu^{1+\ell+\alpha}d^{\frac{1+\ell+\alpha}{p}},}$

(ii) $\forall x\in\mathbb{R}^{d}$: ${\displaystyle \left\Vert \nabla U_{\mu}(x)-\nabla U(x)\right\Vert \leq L\mu^{\alpha}d^{1+\frac{1}{p}}},$

(iii) $\forall x,\ y\in\mathbb{R}^{d}$: ${\displaystyle \left\Vert \nabla U_{\mu}(y)-\nabla U_{\mu}(x)\right\Vert \leq\frac{L}{\mu^{1-\alpha}}d^{\frac{2}{p}}\left\Vert y-x\right\Vert .}$

\end{lemma} \begin{proof} See Appendix \ref{Apgeneralized} \end{proof}

\section{Convergence in KL divergence along ULA under LSI \label{sec:Experiment}}

In this section we review the definition of KL divergence and the
convergence of KL divergence along the Langevin dynamics in continuous
time under log-Sobolev inequality. We then derive our main result
for ULA algorithm under LSI.

\subsection{Recall KL divergence along Langevin dynamics}

\label{Sec:Review}

Let $p,\pi$ be probability density functions with respect to the
Lebesgue measure on $\mathbb{R}^{d}$. KL divergence of $p$ with
respect to $\pi$ is defined as

\begin{equation}
{\displaystyle H(p|\pi)\stackrel{\triangle}{=}\int_{\mathbb{R}^{d}}\log\frac{p(x)}{\pi(x)}\pi(x)dx.}
\end{equation}
By definition, KL divergence can be considered as a measure of asymmetric
``distance'' of a probability distribution $p$ from a base distribution
$\pi$. $H(p|\pi)$ is always nonnegative and equals zero only when
$p$ equals $\pi$ in distribution. KL divergence is a rather strong
measure of distance, which upper bounds a variety of distance measures.
We provide the definition of other measures in Appendix \ref{App0}.
In general, convergence in KL divergence implies convergence in total
variation by Csiszar-Kullback-Pinsker inequality. In addition, under
log-Sobolev inequality with constant $\gamma$, KL divergence also
controls the quadratic Wasserstein $W_{2}$ distance by $\mathcal{W}_{2}(p,\ \pi)^{2}\leq\frac{2}{\gamma}H(p|\pi).$

The Langevin dynamics for target distribution $\pi=e^{-U}$ is a continuous-time
stochastic process $(X_{t})_{t\ge0}$ in $\mathbb{R}^{d}$ that progresses
following the stochastic differential equation: 
\begin{align}
dX_{t}=-\nabla U(X_{t})\,dt+\sqrt{2}\,dW_{t}\label{Eq:LD}
\end{align}
where $(W_{t})_{t\ge0}$ is the standard Brownian motion in $\mathbb{R}^{d}$.

If $(X_{t})_{t\ge0}$ updates following the Langevin dynamics~\eqref{Eq:LD},
then their probability density function $(p_{t})_{t\ge0}$ will satisfy
the following the Fokker-Planck equation: 
\begin{align}
\frac{\partial p_{t}}{t}\,=\,\nabla\cdot(p_{t}\nabla U)+\Delta p_{t}\,=\,\nabla\cdot\left(p_{t}\nabla\log\frac{p_{t}}{\pi}\right).\label{Eq:FP}
\end{align}
Here $\nabla\cdot$ is the divergence and $\Delta$ is the Laplacian
operator. In general, by evolving along the Langevin dynamics, a distribution
will get closer to its target distribution $\pi$. From \citep{vempala2019rapid}
Lemma 1, we have

\begin{align}
\frac{d}{dt}(H(p_{t}|\pi))=-\mathbb{E}_{\pi}\left\Vert \nabla\log\frac{p_{t}}{\pi}\right\Vert ^{2}.\label{Eq:HdotLD}
\end{align}
Since $\mathbb{E}_{\pi}\left\Vert \nabla\log\frac{p_{t}}{\pi}\right\Vert ^{2}\ge0$,
the identity~\eqref{Eq:HdotLD} exhibits that KL divergence with
respect to $\pi$ is diminishing along the Langevin dynamics, thus
the distribution $p_{t}$ actually converges to $\pi$. When $\pi$
fulfills log-Sobolev inequality (LSI), \citep{vempala2019rapid} Lemma
2 shows that

\begin{align}
H(p_{t}|\pi)\le e^{-2\gamma t}H(p_{0}|\pi).\label{Eq:HRateLD}
\end{align}
Hence, KL divergence converges exponentially fast along the Langevin
dynamics. Log-Sobolev inequality can be thought as a relaxation of
logconcavity in continuous time. LSI was originally initiated by \citep{gross1975logarithmic}
for the scenario of Gaussian measure, characterizes concentration
of measure and sub-Gaussian tail property, to name a few. \citep{bakry1985diffusions}
broadened it to strongly log-concave measure, where $\pi$ enjoys
LSI with constant $\gamma$ whenever $-\log\pi$ is $\gamma$-strongly
convex. However, LSI is more general than strongly log-concave condition
since it is preserved under bounded perturbation \citep{holley1986logarithmic},
Lipschitz mapping, tensorization, among others. Therefore, we will
study the KL convergence under log-Sobolev inequality.

\subsection{Main result: KL divergence along Unadjusted Langevin Algorithm}

\label{Sec:Review-1}

In general, a practical algorithm often needs to be discretized \citep{kloeden1992stochastic}
but the discretization algorithms are often more complicated and require
more assumptions. In this section, we investigate the behavior of
KL divergence along the Unadjusted Langevin Algorithm (ULA) in discrete
time. In order to sample from a target distribution $\pi=e^{-U}$
in $\mathbb{R}^{d}$, the updating rule for the discretized ULA algorithm
with step size $\eta>0$ is defined as 
\begin{align}
x_{k+1}=x_{k}-\eta\nabla U(x_{k})+\sqrt{2\eta}\,z_{k}\label{Eq:ULA}
\end{align}
where $z_{k}\sim N(0,I)$ is an independent standard Gaussian random
variable in $\mathbb{R}^{d}$. As $x_{k}$ is updated following ULA,
let $p_{k}$ represent its probability distribution. It is known that
ULA converges to a biased limiting distribution $\pi_{\eta}\neq\pi$
for any fixed $\eta>0$. Thus, $H(p_{k}|\pi)$ does not converge to
$0$ along ULA, as it has an asymptotic bias $H(\pi_{\eta}|\pi)>0$.
When the true target distribution $\pi$ complies with an $\alpha$-mixture
weakly smoothness and LSI, we can prove convergence in KL divergence
along ULA. A key observation is that ULA algorithm will converge uniformly
in time if the discretization error between the ULA output in one
iteration and the Langevin dynamics is bounded. This technique has
been used in many papers, \citep{vempala2019rapid,cheng2018convergence}.
Our proof structure is similar to that of \citep{vempala2019rapid},
whose analysis needs stronger assumptions.

Let $x_{k+1}\sim p_{k+1}$ be the output of one step of ULA~\eqref{Eq:ULA}
from $x_{k}\sim p_{k}$, we have

\begin{lemma}\label{Lem:OneStep} Suppose $\pi$ is $\gamma-$log-Sobolev,
$\alpha$-mixture weakly smooth, $\max\left\{ L_{i}\right\} =L\geq1$.
If $0<\eta\le\left(\frac{\gamma}{9N^{\frac{3}{2}}L^{3}}\right)^{\frac{1}{\alpha}}$
and then along each step of ULA~\eqref{Eq:ULA},

\begin{align}
H(p_{k+1}|\pi)\le e^{-\gamma\eta}H(p_{k}|\pi)+2\eta^{\alpha+1}D_{3},\label{Eq:Main1-2-1-1-1-1-1}
\end{align}
where 
\begin{equation}
D_{3}=\sum_{i}10N^{3}L^{6}+16NL^{4}+8N^{2}L^{4}d^{\frac{3}{p}}+4NL^{2}d.\label{eq:D3}
\end{equation}
In particular, if$\pi$ is $\gamma$-log-Sobolev, $\left(\alpha,\ell\right)$-weakly
smooth with $0<\alpha+\ell\leq1$. If $0<\eta\le\left(\frac{\gamma}{2L^{1+\alpha}}\right)^{\frac{1}{\alpha}}$,
then along each step of ULA~\eqref{Eq:ULA},

\begin{align}
H(p_{k+1}|\pi)\le e^{-\gamma\eta}H(p_{k}|\pi)+2\eta^{\alpha+1}D_{3}^{\prime},\label{Eq:Main1-2-1-1-1-2}
\end{align}
where $D_{3}^{\prime}=16L^{2+2\alpha+2\ell}+4L^{2+2\alpha}d^{\frac{3-\alpha}{1+\alpha}\left(\alpha+\ell\right)}+4L^{2}d^{\alpha+\ell}$.

\end{lemma}\begin{proof} See Appendix \ref{Proof-of-LemmaOneStep}.
\end{proof}By using this component, we obtain the following theorem.

\begin{thm} \label{T1} Suppose $\pi$ is $\gamma-$log-Sobolev,
$\alpha$-mixture weakly smooth,$\max\left\{ L_{i}\right\} =L\geq1$,
and for any $x_{0}\sim p_{0}$ with $H(p_{0}|\pi)=C_{0}<\infty$,
the iterates $x_{k}\sim p_{k}$ of LMC~ with step size $\eta\leq1\wedge\frac{1}{4\gamma}\wedge\left(\frac{\gamma}{9N^{\frac{3}{2}}L^{3}}\right)^{\frac{1}{\alpha}}$satisfies
\begin{align}
H(p_{k}|\pi)\le e^{-\gamma\eta k}H(p_{0}|\pi)+\frac{8\eta^{\alpha}D_{3}}{3\gamma},\label{Eq:Main1-2-1-1}
\end{align}

Then, for any $\epsilon>0$, to achieve $H(p_{k}|\pi)<\epsilon$,
it suffices to run ULA with step size $\eta\le1\wedge\frac{1}{4\gamma}\wedge\left(\frac{\gamma}{9N^{\frac{3}{2}}L^{3}}\right)^{\frac{1}{\alpha}}\wedge\left(\frac{3\epsilon\gamma}{16D_{3}}\right)^{\frac{1}{\alpha}}$for
$k\ge\frac{1}{\gamma\eta}\log\frac{2H\left(p_{0}|\pi\right)}{\epsilon}$
iterations.\end{thm}

\begin{proof} See Appendix \ref{Proof-of-Theorem1}. \end{proof}If
we initialize with a Gaussian distribution $p_{0}=N(0,\frac{1}{L}I)$,
we have the following lemma.

\begin{lemma}\label{Lem:Initial} Suppose $\pi=e^{-U}$ is $\alpha$-mixture
weakly smooth. Let $p_{0}=N(0,\frac{1}{L}I)$. Then $H(p_{0}|\pi)\le U(0)-\frac{d}{2}\log\frac{2\Pi e}{L}+\sum_{i}\frac{L}{1+\alpha_{i}}\left(\frac{d}{L}\right)^{\frac{1+\alpha_{i}}{2}}=O(d).$

\end{lemma}

\begin{proof} See Appendix \ref{BInitial}. \end{proof} Therefore,
Theorem \ref{T1} states that to achieve $H(p_{k}|\pi)\le\epsilon$,
ULA has iteration complexity $\tilde{O}\left(\frac{d^{\frac{3-\alpha}{1+\alpha}}}{\epsilon^{\frac{1}{\alpha}}\gamma^{\frac{1}{\alpha}+1}}\right).$
By Pinsker's inequality, we have $TV\left(p_{k}|\pi\right)\leq\sqrt{\frac{H(p_{k}|\pi)}{2}}$
which implies that to get $TV\left(p_{k}|\pi\right)\leq\epsilon$,
it is enough to obtain $H(p_{k}|\pi)\le2\epsilon^{2}$. This bound
indicates that the number of iteration to reach $\epsilon$ accuracy
for total variation is $\tilde{O}\left(d^{\frac{3-\alpha}{1+\alpha}}\gamma^{\frac{-1}{\alpha}-1}\epsilon^{\frac{-2}{\alpha}}\right)$.
On the other hand, from Talagrand inequality, which comes from log-Sobolev
inequality, we know that $W_{2}^{2}(p_{k},\ \pi)\leq H\left(p_{k}|\pi\right)$,
by replacing this in the bound above, we obtain the number of iteration
for $L_{2}$-Wasserstein distance is $\tilde{O}\left(d^{\frac{3-\alpha}{1+\alpha}}\gamma^{\frac{-1}{\alpha}-1}\epsilon^{\frac{-2}{\alpha}}\right)$.

\subsection{Sampling via smoothing potential}

Inspired by the approach of \citep{chatterji2019langevin}, we study
the convergence of the discrete-time process for the smoothing potential
that have the following form: 
\begin{equation}
U_{\mu}(x):=\mathrm{\mathbb{E}}_{\xi}[U(y+\mu\xi)].\label{eq:2.5b}
\end{equation}
Observe from Lemma \ref{2.1} that $U(\cdot)$ is $\alpha$-mixture
weakly smooth but $U_{\mu}(x)$ is smooth. Recall that ULA in terms
of the smoothing potential $U_{\mu}$ can be specified as: 
\begin{equation}
x_{k+1}=x_{k}-\eta\nabla U_{\mu}(x_{k})+\sqrt{2\eta}\varsigma_{k},\label{eq:LMC}
\end{equation}
where $\varsigma_{k}\sim N(0,\ I_{d\times d})$ are independent Gaussian
random vectors. In general, we do not have access to an oracle of
$\nabla U_{\mu}(x)$, so rather than working with $\nabla U_{\mu}(x)$
as specified by Eq. \ref{eq:LMC}, we need to use an estimate of the
gradient: 
\begin{align}
g_{\mu}(x)=\nabla U(x+\mu\xi)\label{eq:gradest}
\end{align}
where $\xi\sim N_{p}(0,I_{d})$. Based on the above estimate of the
gradient, we obtain the following result.\begin{lemma} \label{3.3.1}
For any $x_{k}\in\mathbb{R}^{d}$, $g_{\mu}(x_{k},\zeta_{k})=\nabla U_{\mu}(x_{k})+\zeta_{k}$
is an unbiased estimator of $\nabla U_{\mu}$ such that 
\begin{align*}
\mathrm{Var}\left[g_{\mu}(x_{k},\zeta_{k})\right]\leq4N^{2}L^{2}\mu^{2\alpha}d^{\frac{2\alpha}{p}}.
\end{align*}
\end{lemma}\begin{proof} See Appendix \ref{Proof-of-Lemma3.3.1}.
\end{proof}

Let the distribution of the $k^{th}$ iterate $x_{k}$ be represented
by $\pi_{\mu,k}$, and let $\pi_{\mu}\propto\exp(-U_{\mu})$ be the
distribution with $U_{\mu}$ as the potential. First, we prove that
the $p$-generalized Gaussian smoothing does not alter the objective
distribution substantially in term of the Wasserstein distance, by
bounding $W_{2}(\pi,\pi_{\mu})$. \begin{lemma} \label{3.3.2}Assume
that $\pi\propto\exp(-\pi)$ and $\pi_{\mu}\propto\exp(-U_{\mu})$
and $\pi$ has a bounded second moment, that is $\int\left\Vert x\right\Vert ^{2}\pi(x)dx=E_{2}<\infty$.
We deduce the following bounds 
\[
W_{2}^{2}(\pi,\ \pi_{\mu})\leq8.24NL\mu^{1+\alpha}d^{\frac{2}{p}}E_{2}.
\]
for any $\mu\leq0.05$. \end{lemma}\begin{proof} See Appendix \ref{Proof-of-Lemma3.3.2}.
\end{proof}We then derive a result on mixing times of Langevin diffusion
with stochastic estimated gradients under log-Sobolev inequality condition,
which enables us to bound $W_{2}(\pi_{\mu,k},\pi_{\mu})$. Our main
outcome is stated in the subsequent theorem. \begin{thm} \label{Theorem3.3.1}Suppose
$\pi_{\mu}$ is $\gamma_{1}-$log-Sobolev, $\alpha$-mixture weakly
smooth, with $\max\left\{ L_{i}\right\} =L\geq1$ and $\int\left\Vert x\right\Vert ^{2}\pi(x)dx=E_{2}<\infty$
and for any $x_{0}\sim p_{0}$ with $H(p_{0}|\pi)=C_{0}<\infty$,
the iterates $x_{k}\sim p_{k}$ of ULA~ with step size 
\begin{equation}
\eta\le\min\left\{ 1,\frac{1}{4\gamma},\left(\frac{\gamma_{1}}{13N^{\frac{3}{2}}L^{3}}\right)^{\frac{1}{\alpha}}\right\} 
\end{equation}
satisfies 
\begin{align*}
 & W_{2}(\pi_{\mu,K},\pi)\leq e^{-\frac{\gamma_{1}}{2}\eta k}\sqrt{H(p_{0}|\pi_{\mu})}+\sqrt{\frac{8\eta^{\alpha}D_{4}}{3\gamma_{1}}}+3\sqrt{NLE_{2}}d^{\frac{1}{p}}\eta^{\frac{\alpha}{2}},
\end{align*}
where $D_{4}=\sum_{i}10N^{3}L^{6}+16NL^{4}+8N^{2}L^{4}d^{\frac{3}{p}}+4NL^{2}d+8N^{2}L^{2}d^{\frac{2\alpha}{p}}$.

Then, for any $\epsilon>0$, to achieve $W_{2}(\pi_{\mu,K},\pi)<\epsilon$,
it suffices to run ULA with step size $\eta\le1\wedge\frac{1}{4\gamma_{1}}\wedge\left(\frac{\gamma}{13N^{\frac{3}{2}}L^{3}}\right)^{\frac{1}{\alpha}}\wedge\left(\frac{\epsilon\gamma_{1}}{6\sqrt{D_{4}}}\right)^{\frac{2}{\alpha}}\wedge\left(\frac{\epsilon}{9\sqrt{NLE_{2}}d^{\frac{1}{p}}}\right)^{\frac{2}{\alpha}}$for
$k\ge\frac{2}{\gamma_{1}\eta}\log\frac{3\sqrt{H\left(p_{0}|\pi\right)\gamma_{1}}}{\epsilon}$
iterations.\end{thm}

\begin{proof} See Appendix \ref{Proof-of-Theorem3.3.1}. \end{proof}

\section{Extended result \label{sec:Experiment-1}}


Since log-Sobolev inequalities are preserved under bounded perturbations
by \citep{holley1986logarithmic}'s theorem, we provide our extended
results through convexification of non-convex domain \citep{ma2019sampling,yan2012extension}.
Convexification of non-convex domain is an original approach proposed
by \citep{ma2019sampling,yan2012extension}, developed and apply to
strongly convex outside a compact set by \citep{yan2012extension}.
We would like to emphasize that it is non trivial to apply their results
in our case since the requirement of strong convexity. Before starting
our extension, we need an additional lemma, taken from \citep{ma2019sampling,yan2012extension},
for our proof.

\begin{lemma}\label{Lem4.0.1} {[}\citep{ma2019sampling} Lemma 2{]}.
Let us define $\Omega=\mathbb{R}^{d}\backslash\mathbb{B}(0,R)$ where
$\mathbb{B}(0,R)$ is the open ball of radius $R$ centered at $0$,
and define $V\left(x\right)=\inf\left\{ \sum_{i}\lambda_{i}U(\ x_{i})\right\} $
where the infimum is running over all possible convex combination
of points $x_{i}$ (that is $\lambda_{i}\geq0$, $\sum_{i}\lambda_{i}=1$
and $\sum_{i}\lambda_{i}x_{i}=x$). Then for $\forall\ x\in\mathbb{B}(0,R)$,
$V(\ x)$ can be represented as a convex combination of $U\left(x_{j}\right)$
such that $\left\Vert x_{j}\right\Vert =R,$ that is $V\left(x\right)=\inf\left\{ \sum_{j}\lambda_{j}U(\ x_{j})\right\} $
where $\lambda_{j}\geq0$, $\sum_{j}\lambda_{j}=1$ and $\sum_{j}\lambda_{j}x_{j}=x$
and $\left\Vert x_{j}\right\Vert =R.$ Then, $\inf_{\left\Vert \bar{x}\right\Vert =R}U(\bar{x})\leq V(\ x)\leq\sup_{\left\Vert \bar{x}\right\Vert =R}U(\bar{x}).$\label{L1-1}
\end{lemma}

\begin{proof} See Appendix \ref{Proof-of-Lemma4.0.1}. \end{proof}Adapted
techniques from \citep{ma2019sampling} for non-strongly convex and
$\alpha$-mixture weakly smooth potentials, we derive a tighter bound
for the difference between constructed convex potential and the original
one in the following lemma.

\begin{lemma} \label{Lem4.0.2}For $U$ satisfying $\alpha$-mixture
weakly smooth and $\left(\mu,\theta\right)$-degenerated convex outside
the ball radius $R$, there exists $\hat{U}\in C^{1}(\mathbb{R}^{d})$
with a Hessian that exists everywhere on $\mathbb{R}^{d}$, and $\hat{U}$
is $\left(\left(1-\theta\right)\frac{\mu}{2},\theta\right)$-degenerated
convex on $\mathbb{R}^{d}$ (that is $\nabla^{2}\hat{U}(x)\succeq\left(1-\theta\right)\frac{\mu}{2}\left(1+\left\Vert x\right\Vert ^{2}\right)^{-\frac{\theta}{2}}I_{d}$),
such that 
\begin{align}
\sup\left(\hat{U}(\ x)-U(\ x)\right) & -\inf\left(\hat{U}(\ x)-U(\ x)\right)\leq\sum_{i}L_{i}R^{1+\alpha_{i}}+\frac{4\mu}{\left(2-\theta\right)}\ R^{2-\theta}.
\end{align}
\label{L2-3} \end{lemma} \begin{proof} See Appendix \ref{Proof-of-Lemma4.0.2}.
\end{proof}\begin{remark} This result can be applied to potential
with degenerated convex outside the ball. Setting $\mu=0$ implies
a result for potential with non-strongly convex outside the ball,
while setting $\theta=0$ implies a result for potential with strongly
convex outside the ball. The constant could be improved by a factor
of $2$ if we take $\epsilon$ defined in the proof to be arbitrarily
small. \end{remark}

\subsection{ULA convergence under $\gamma-$Poincar\'{e} inequality, $\alpha$-mixture
weakly smooth and $2-$dissipativity}

In general, $PI$ is weaker than $LSI$. In order to apply the previous
results of log Sobolev inequalities, we will also need $2-$dissipativity
assumption. First, using convexification of non-convex domain result
above, we have the following lemma for bounded perturbation.\begin{lemma}
\label{Lem4.1.1}For $U$ satisfying $\gamma-$Poincar\'{e}, $\alpha$-mixture
weakly smooth, there exists $\breve{U}\in C^{1}(\mathbb{R}^{d})$
with a Hessian that exists everywhere on $\mathbb{R}^{d}$, and $\breve{U}$
is log-Sobolev on $\mathbb{R}^{d}$ such that 
\begin{equation}
\sup\left(\breve{U}(\ x)-U(\ x)\right)-\inf\left(\breve{U}(\ x)-U(\ x)\right)\leq2\sum_{i}L_{i}R^{1+\alpha_{i}}+4L_{N}R^{2}+4LR^{1+\alpha}.
\end{equation}
\label{lemma:hat_U-2-2-1} \end{lemma} \begin{proof} See Appendix
\ref{Proof-of-lemma4.1.1}. \end{proof}Using bounded perturbation
theorem, this result implies $\pi$ satisfies a log-Sobolev inequality,
which in turn give us the following result.\begin{thm} \label{Prop4.1.1}
Suppose $\pi$ is $\gamma-$Poincar\'{e}, $\alpha$-mixture weakly smooth
with $\alpha_{N}=1$ and $2-$dissipativity (i.e.$\left\langle \nabla U(x),x\right\rangle \geq a\left\Vert x\right\Vert ^{2}-b$)
for some $a,b>0$, and for any $x_{0}\sim p_{0}$ with $H(p_{0}|\pi)=C_{0}<\infty$,
the iterates $x_{k}\sim p_{k}$ of ULA~ with step size $\eta\leq1\wedge\frac{1}{4\gamma_{3}}\wedge\left(\frac{\gamma_{3}}{16L^{1+\alpha}}\right)^{\frac{1}{\alpha}}$satisfies
\begin{align}
H(p_{k}|\pi)\le e^{-\gamma_{3}\eta k}H(p_{0}|\pi)+\frac{8\eta^{\alpha}D_{3}}{3\gamma_{3}},\label{Eq:Main1-2-1-4-1-2-1}
\end{align}
where $D_{3}$ is defined as in equation (\ref{eq:D3}) and 
\begin{align}
M_{2} & =\int\left\Vert x\right\Vert ^{2}e^{-\breve{U}(x)}dx=O(d)\\
\zeta & =\sqrt{2\left[\frac{2\left(b+\left(L+\frac{\lambda_{0}}{2}\right)R^{2}+aR^{2}+d\right)}{a}+M_{2}\right]\frac{e^{4\left(2\sum_{i}L_{i}R^{1+\alpha_{i}}+4L_{N}R^{2}+4LR^{1+\alpha}\right)}}{\gamma}}\\
A & =(1-\frac{L}{2})\frac{8}{a^{2}}+\zeta,\\
B & =2\left[\frac{2\left(\left(b+4\left(L+\frac{\lambda_{0}}{4}\right)R^{2}+aR^{2}\right)+d\right)}{a}+M_{2}\right](1-\frac{L}{2}+\frac{1}{\zeta}),\\
\gamma_{3} & =\frac{2\gamma e^{-\left(2\sum_{i}L_{i}R^{1+\alpha_{i}}+4L_{N}R^{2}+4LR^{1+\alpha}\right)}}{[A\gamma+(B+2)e^{4\left(2\sum_{i}L_{i}R^{1+\alpha_{i}}+4L_{N}R^{2}+4LR^{1+\alpha}\right)})]}.\nonumber 
\end{align}
Then, for any $\epsilon>0$, to achieve $H(p_{k}|\pi)<\epsilon$,
it suffices to run ULA with step size $\eta\le1\wedge\frac{1}{4\gamma_{3}}\wedge\left(\frac{\gamma_{3}}{16L^{1+\alpha}}\right)^{\frac{1}{\alpha}}\wedge\left(\frac{3\epsilon\gamma_{3}}{16D_{3}}\right)^{\frac{1}{\alpha}}$for
$k\ge\frac{1}{\gamma_{3}\eta}\log\frac{2H\left(p_{0}|\pi\right)}{\epsilon}$
iterations.\end{thm}\begin{proof} See Appendix \ref{Proof-of-Prop4.1.1-1}.
\end{proof}From Theorem \ref{Prop4.1.1}, LMC can achieve $H(p_{k}|\pi)\le\epsilon$,
with iteration complexity of $\tilde{O}\left(\frac{d^{\frac{1}{\alpha}}}{\epsilon^{\frac{1}{\alpha}}\gamma_{3}^{\frac{1}{\alpha}+1}}\right)$
where 
\begin{align*}
\gamma_{3} & =O\left(\frac{1}{d\gamma e^{5\left(2\sum_{i}L_{i}R^{1+\alpha_{i}}+4L_{N}R^{2}+4LR^{1+\alpha}\right)}}\right)
\end{align*}
so the number of iteration needed is 
\[
\tilde{O}\left(\frac{d^{\frac{2}{\alpha}+1}e^{5\left(2\sum_{i}L_{i}R^{1+\alpha_{i}}+4L_{N}R^{2}+4LR^{1+\alpha}\right)\left(\frac{1}{\alpha}+1\right)}}{\gamma_{3}^{\left(\frac{1}{\alpha}+1\right)}\epsilon^{\frac{1}{\alpha}}}\right).
\]
Similar as before, from Pinsker's inequality, the number of iteration
to reach $\epsilon$ accuracy for total variation is 
\begin{equation}
\tilde{O}\left(\frac{d^{\frac{2}{\alpha}+1}e^{5\left(2\sum_{i}L_{i}R^{1+\alpha_{i}}+4L_{N}R^{2}+4LR^{1+\alpha}\right)\left(\frac{1}{\alpha}+1\right)}}{\gamma_{3}^{\left(\frac{1}{\alpha}+1\right)}\epsilon^{\frac{2}{\alpha}}}\right).
\end{equation}
To have ${W}_{\alpha}(p_{k},\ \pi)\leq\epsilon$, it is sufficient
to choose $\mathrm{H}(p_{k}|\pi)=\tilde{O}\left(\epsilon^{4}d^{-2}\right)$,
which in turn implies the number of iteration for ${W}_{\alpha}(p_{k},\ \pi)\leq\epsilon$
is 
\begin{equation}
\tilde{O}\left(\frac{d^{\frac{4}{\alpha}+1}e^{5\left(2\sum_{i}L_{i}R^{1+\alpha_{i}}+4L_{N}R^{2}+4LR^{1+\alpha}\right)\left(\frac{1}{\alpha}+1\right)}}{\gamma_{3}^{\left(\frac{1}{\alpha}+1\right)}\epsilon^{\frac{4}{\alpha}}}\right).
\end{equation}

\subsection{ULA convergence under non-strongly convex outside the ball, $\alpha$-mixture
weakly smooth and $2-$dissipativity}

Using convexification of non-convex domain result above, we obtain
the following lemma.\begin{lemma} \label{Lem4.1.1-1}Suppose $\pi$
is non-strongly convex outside the ball of radius $R$, $\alpha$-mixture
weakly smooth with $\alpha_{N}=1$ and $2-$dissipativity (i.e.$\left\langle \nabla U(x),x\right\rangle \geq a\left\Vert x\right\Vert ^{2}-b$)
for some $a,b>0$, there exists $\breve{U}\in C^{1}(\mathbb{R}^{d})$
with a Hessian that exists everywhere on $\mathbb{R}^{d}$, and $\breve{U}$
is convex on $\mathbb{R}^{d}$ such that 
\begin{equation}
\sup\left(\breve{U}(\ x)-U(\ x)\right)-\inf\left(\breve{U}(\ x)-U(\ x)\right)\leq2\sum_{i}L_{i}R^{1+\alpha_{i}}.
\end{equation}
\label{lemma:hat_U-2-2-1-1} \end{lemma} \begin{proof} It comes
directly from Lemma \ref{Lem4.0.2}. \end{proof}Based on result in
previous section, we get the following result.

\begin{thm} \label{Prop4.1.1-1} Suppose $\pi$ is non-strongly convex
outside the ball $\mathbb{B}(0,R)$, $\alpha$-mixture weakly smooth
with $\alpha_{N}=1$ and $2-$dissipativity (i.e.$\left\langle \nabla U(x),x\right\rangle \geq a\left\Vert x\right\Vert ^{2}-b$)
for some $a,b>0$, and for any $x_{0}\sim p_{0}$ with $H(p_{0}|\pi)=C_{0}<\infty$,
the iterates $x_{k}\sim p_{k}$ of LMC~ with step size $\eta\leq1\wedge\frac{1}{4\gamma_{3}}\wedge\left(\frac{\gamma_{3}}{16L^{1+\alpha}}\right)^{\frac{1}{\alpha}}$satisfies
\begin{align}
H(p_{k}|\pi)\le e^{-\gamma_{3}\eta k}H(p_{0}|\pi)+\frac{8\eta^{\alpha}D_{3}}{3\gamma_{3}},\label{Eq:Main1-2-1-4-1-3-1-1}
\end{align}
where $D_{3}$ is defined as in equation (\ref{eq:D3}) and for some
universal constant $K$, 
\begin{align}
M_{2} & =\int\left\Vert x\right\Vert ^{2}e^{-\breve{U}(x)}dx=O(d)\\
\zeta & =K\sqrt{64d\left[\frac{2\left(b+\left(L+\frac{\lambda_{0}}{2}\right)R^{2}+aR^{2}+d\right)}{a}+M_{2}\right]\left(\frac{a+b+2aR^{2}+3}{ae^{-4\left(4L_{N}R^{2}+4LR^{1+\alpha}\right)}}\right)}\\
A & =(1-\frac{L}{2})\frac{8}{a^{2}}+\zeta,\\
B & =2\left[\frac{2\left(\left(b+4\left(L+\frac{\lambda_{0}}{4}\right)R^{2}+aR^{2}\right)+d\right)}{a}+M_{2}\right](1-\frac{L}{2}+\frac{1}{\zeta}),\\
\gamma_{3} & =\frac{2e^{-\left(2\sum_{i}L_{i}R^{1+\alpha_{i}}+4L_{N}R^{2}+4LR^{1+\alpha}\right)}}{A+(B+2)32K^{2}d\left(\frac{a+b+2aR^{2}+3}{a}\right)e^{4\left(4L_{N}R^{2}+4LR^{1+\alpha}\right)}}=\frac{1}{O(d)}.\nonumber 
\end{align}
Then, for any $\epsilon>0$, to achieve $H(p_{k}|\pi)<\epsilon$,
it suffices to run ULA with step size $\eta\le1\wedge\frac{1}{4\gamma_{3}}\wedge\left(\frac{\gamma_{3}}{16L^{1+\alpha}}\right)^{\frac{1}{\alpha}}\wedge\left(\frac{3\epsilon\gamma_{3}}{16D_{3}}\right)^{\frac{1}{\alpha}}$for
$k\ge\frac{1}{\gamma_{3}\eta}\log\frac{2H\left(p_{0}|\pi\right)}{\epsilon}$
iterations.\end{thm}

\begin{proof} See Appendix \ref{Proof-of-Prop4.1.1-1} . \end{proof}




\section{Conclusion\label{sec:conclusion}}

In this article, we derive polynomial-dimension theoretical assurances
of unadjusted LMC algorithm for a family of potentials that are $\alpha$-mixture
weakly smooth and isoperimetric (i.e. log Sobolev, Poincar\'{e}, and Talagrand).
In addition, we also investigate the family of potential which is
non-strongly convex outside the ball and $2$-dissipative. The analysis
we proposed is an extension of the recently published paper \citep{vempala2019rapid}
in combination with the convexification of non-convex domain \citep{ma2019sampling}.
There are a number of valuable potential directions which one can
explore, among them we speculate some here. It is potential to broaden
our results to apply underdamped LMC or higher order LMC to these
class of potential while the computational complexity remains polynomial
dependence on $d$. Another fascinating question is whether it is
feasible to sampling from distributions with non-smooth and totally
non-convex structure and integrate into derivative-free LMC algorithm.

\appendix

\section{Measure definitions and isoperimetry \label{App0}}

Let $p,\pi$ be probability distributions on $\mathbb{R}^{d}$ with
full support and smooth densities, define the Kullback-Leibler (KL)
divergence of $p$ with respect to $\pi$ as 
\begin{equation}
H(p|\pi)\stackrel{\triangle}{=}\int_{{R}^{d}}p(x)\log\frac{p(x)}{\pi(x)}\,dx.
\end{equation}
\label{Eq:Hnu-1} Likewise, we denote the entropy of $p$ with 
\begin{equation}
{\displaystyle \mathrm{H}(p)\stackrel{\triangle}{=}-\int p(x)\log p(x)dx}
\end{equation}
and for $\mathcal{B}(\mathbb{R}^{d})$ denotes the Borel $\sigma$-field
of $\mathbb{R}^{d}$, define the relative Fisher information and total
variation metrics correspondingly as 
\begin{equation}
{\displaystyle \mathrm{I}(p|\pi)\stackrel{\triangle}{=}\int_{\mathbb{R}^{d}}p(x)\Vert\nabla\log\frac{p(x)}{\pi(x)}\Vert^{2}dx},
\end{equation}
\begin{equation}
{\displaystyle TV(p,{\displaystyle \ \pi)\stackrel{\triangle}{=}\sup_{A\in\mathcal{B}(\mathbb{R}^{d})}|\int_{A}p(x)dx-\int_{A}\pi(x)dx|}.}
\end{equation}
Furthermore, we define a transference plan $\zeta$, a distribution
on $(\mathbb{R}^{d}\times\mathbb{R}^{d},\ \mathcal{B}(\mathbb{R}^{d}\times\mathbb{R}^{d}))$
(where $\mathcal{B}(\mathbb{R}^{d}\times\mathbb{R}^{d})$ is the Borel
$\sigma$-field of ($\mathbb{R}^{d}\times\mathbb{R}^{d}$)) so that
$\zeta(A\times\mathbb{R}^{d})=p(A)$ and $\zeta(\mathbb{R}^{d}\times A)=\pi(A)$
for any $A\in\mathcal{B}(\mathbb{R}^{d})$. Let $\Gamma(P,\ Q)$ designate
the set of all such transference plans. Then for $\beta>0$, the $L_{\beta}$-Wasserstein
distance is formulated as: 
\begin{equation}
W_{\beta}(p,\pi)\stackrel{\triangle}{=}\left(\inf_{\zeta\in\Gamma(P,Q)}\int_{x,y\in\mathbb{R}^{d}}\Vert x-y\Vert^{\beta}\mathrm{d}\zeta(x,\ y)\right)^{1/\beta}.
\end{equation}
Note that although KL divergence is an asymmetric measure of distance
between probability distributions, it is the preferred measure of
distance here since it also implies total variation distance via Pinsker's
inequality. In addition, KL divergence also governs the quadratic
Wasserstein $W_{2}$ distance under log-Sobolev, Talagrand, and Poincar\'{e}
inequalities defined below. \begin{definition} \label{D1} The probability
distribution $p$ satisfies a logarithmic Sobolev inequality with
constant $\gamma>0$ (in short: $LSI(\gamma)$) if for all probability
distribution $p$ absolutely continuous $w.r.t.\ \pi$, 
\begin{equation}
H({\displaystyle p|\pi)\leq\frac{1}{2\gamma}I(p|\pi)}.
\end{equation}
\end{definition} \begin{definition} \label{D2} The probability
distribution $p$ satisfies a Talagrand inequality with constant $\gamma>0$
(in short: $T(\gamma)$) if for all probability distribution $p$,
absolutely continuous $w.r.t.\ \pi$, with finite moments of order
2, 
\begin{equation}
W_{2}(p,\ \pi)\leq\sqrt{\frac{2H(p|\pi)}{\gamma}}.
\end{equation}
\end{definition} \begin{definition} \label{D3} The probability
distribution $p$ satisfies a Poincar\'{e} inequality with constant $\gamma>0$
(in short: $PI(\gamma)$) if for all smooth function $g\colon\mathbb{R}^{d}\to\mathbb{R}$,
\begin{equation}
Var_{p}(g)\le\frac{1}{\gamma}E_{p}[\|\nabla g\|^{2}],
\end{equation}
where $Var_{p}(g)=E_{p}[g^{2}]-E_{p}[g]^{2}$ is the variance of $g$
under $p$. \end{definition}

\section{Proofs of $p$-generalized Gaussian smoothing \label{AppA}}

\subsection{Proof of $\alpha$-mixture weakly smooth property\label{Asmooth}}

\begin{lemma} If potential $U:\mathbb{R}^{d}\rightarrow\mathbb{R}$
satisfies $\alpha$-mixture weakly smooth then:

\[
U(y)\leq U(x)+\left\langle \nabla U(x),\ y-x\right\rangle +\sum_{i}\frac{L_{i}}{1+\alpha_{i}}\Vert y-x\Vert^{1+\alpha_{i}}.
\]
In particular, if potential $U:\mathbb{R}^{d}\rightarrow\mathbb{R}$
satisfies $\left(\alpha,\ell\right)-$weakly smooth for some $\alpha+\ell\leq1$
and $\alpha\in[0,1]$, then:

\[
U(y)\leq U(x)+\left\langle \nabla U(x),\ y-x\right\rangle +\frac{L}{1+\alpha}\Vert y-x\Vert^{1+\alpha}+\frac{L}{1+\ell+\alpha}\Vert y-x\Vert^{1+\ell+\alpha}.
\]

\end{lemma}

\begin{proof} We have 
\begin{align*}
 & \left|U(x)-U(y)-\langle\nabla U(y),x-y\rangle\right|\\
= & \Big\vert\int_{0}^{1}\langle\nabla U(y+t(x-y)),x-y\rangle\text{d}t-\langle\nabla U(y),x-y\rangle\Big\vert\\
= & \Big\vert\int_{0}^{1}\langle\nabla U(y+t(x-y))-\nabla U(y),x-y\rangle\text{d}t\Big\vert.\\
\leq & \int_{0}^{1}\Vert\nabla U(y+t(x-y))-\nabla U(y)\Vert\Vert x-y\Vert\text{d}t\\
\leq & \int_{0}^{1}\sum_{i}L_{i}t^{\alpha_{i}}\Vert x-y\Vert^{\alpha_{i}}\Vert x-y\Vert\text{d}t\\
= & \sum_{i}\frac{L_{i}}{1+\alpha_{i}}\Vert x-y\Vert^{1+\alpha_{i}},
\end{align*}
where the first line comes from Taylor expansion, the third line follows
from Cauchy-Schwarz inequality and the fourth line is due to Assumption~\ref{A0}.
This gives us the desired result. By replacing Assumption \ref{A0}
with Assumption \ref{A1}, we immediately get

\[
U(y)\leq U(x)+\left\langle \nabla U(x),\ y-x\right\rangle +\frac{L}{1+\alpha}\Vert y-x\Vert^{1+\alpha}+\frac{L}{1+\ell+\alpha}\Vert y-x\Vert^{1+\ell+\alpha}.
\]

\end{proof}

\subsection{Proof of $p$-generalized Gaussian smoothing properties\label{Apgeneralized}}

\begin{lemma} If potential $U:\mathbb{R}^{d}\rightarrow\mathbb{R}$
satisfies $\alpha$-mixture weakly smooth then:

(i) $\forall x\in\mathbb{R}^{d}$ : $\left|U_{\mu}(x)-U(x)\right|{\displaystyle \leq\sum_{i}L_{i}\mu^{1+\alpha_{i}}d^{\frac{1+\alpha_{i}}{p}},}$

(ii) $\forall x\in\mathbb{R}^{d}$: ${\displaystyle \left\Vert \nabla U_{\mu}(x)-\nabla U(x)\right\Vert \leq\sum_{i}L_{i}\mu^{\alpha_{i}}d^{\frac{3}{p}}},$

(iii) $\forall x,\ y\in\mathbb{R}^{d}$: ${\displaystyle \left\Vert \nabla U_{\mu}(y)-\nabla U_{\mu}(x)\right\Vert \leq\sum_{i}\frac{L_{i}}{\mu^{1-\alpha_{i}}}d^{\frac{2}{p}}\left\Vert y-x\right\Vert .}$

In particular, if the potential $U:\mathbb{R}^{d}\rightarrow\mathbb{R}$
satisfies $\left(\alpha,\ell\right)-$weakly smooth for some $\alpha+\ell\leq1$
and $\alpha\in[0,1]$, then:

(i) $\forall x\in\mathbb{R}^{d}$ : $\left|U_{\mu}(x)-U(x)\right|{\displaystyle \leq2L\mu^{1+\ell+\alpha}d^{\frac{1+\ell+\alpha}{p}},}$

(ii) $\forall x\in\mathbb{R}^{d}$: ${\displaystyle \left\Vert \nabla U_{\mu}(x)-\nabla U(x)\right\Vert \leq2L\mu^{\alpha}d^{\frac{3}{p}}},$

(iii) $\forall x,\ y\in\mathbb{R}^{d}$: ${\displaystyle \left\Vert \nabla U_{\mu}(y)-\nabla U_{\mu}(x)\right\Vert \leq\frac{L}{\mu^{1-\alpha}}d^{\frac{2}{p}}\left\Vert y-x\right\Vert .}$

\end{lemma}

\begin{proof} (i). Since $U_{\mu}(x)=\mathrm{\mathbb{E}}_{\xi}[U(x+\mu\xi)]$,
$U(x)=\mathrm{\mathbb{E}}_{\xi}[U(x)]$ and $\mathbb{E}_{\xi}\mu\left\langle \nabla U(x),\ \xi\right\rangle =0$,
we have 
\[
U_{\mu}(x)-U(x)=\mathbb{E}_{\xi}\left[U(x+\mu\xi)-U(x)-\mu\left\langle \nabla U(x),\ \xi\right\rangle \right].
\]
By the definition of the density of $p$-generalized Gaussian distribution
\citep{arellano2012skewed}, we also have: 
\[
U_{\mu}(x)-U(x)=\frac{1}{\kappa}\int_{\mathbb{R}^{d}}[U(x+\mu\xi)-U(x)-\mu\left\langle \nabla U(x),\ \xi\right\rangle ]e^{-\left\Vert \xi\right\Vert _{p}^{p}/p}d\xi.
\]
Applying Eq. \ref{eq:4} and previous inequality:

\begin{align*}
|U_{\mu}(x)-U(x)| & =\left|\frac{1}{\kappa}\int_{\mathbb{R}^{d}}\left[U(x+\mu\xi)-U(x)-\mu\left\langle \nabla U(x),\ \xi\right\rangle \right]e^{-\left\Vert \xi\right\Vert _{p}^{p}/p}d\xi\right|\\
 & \leq\sum_{i}\frac{L_{i}}{\kappa(1+\alpha_{i})}\mu^{1+\alpha_{i}}\int_{\mathbb{R}^{d}}\left\Vert \xi\right\Vert ^{(1+\alpha_{i})}e^{-\left\Vert \xi\right\Vert _{p}^{p}/p}d\xi\\
 & =\sum_{i}\frac{L_{i}\mu^{1+\alpha_{i}}}{(1+\alpha_{i})}E\left[\left\Vert \xi\right\Vert ^{(1+\alpha_{i})}\right].
\end{align*}
If $p\leq2$ then $\left\Vert \xi\right\Vert \leq\left\Vert \xi\right\Vert _{p}$
and we get

\begin{align*}
|U_{\mu}(x)-U(x)| & \leq\sum_{i}\frac{L_{i}\mu^{1+\alpha_{i}}}{(1+\alpha_{i})}E\left[\left\Vert \xi\right\Vert ^{(1+\alpha_{i})}\right]\\
 & \stackrel{_{1}}{\leq}\sum_{i}\frac{L_{i}\mu^{1+\alpha_{i}}}{(1+\alpha_{i})}\mathbb{E}\left[\left\Vert \xi\right\Vert _{p}^{2}\right]^{\frac{1+\alpha_{i}}{2}}\\
 & \stackrel{_{2}}{\leq}\sum_{i}\frac{L_{i}\mu^{1+\alpha_{i}}}{(1+\alpha_{i})}\left(\left(d+1\right)^{\frac{2}{p}}\right)^{\frac{1+\alpha_{i}}{2}}\\
 & \leq\sum_{i}\frac{L_{i}\mu^{1+\alpha_{i}}}{(1+\alpha_{i})}d^{\frac{1+\alpha_{i}}{p}}\\
 & \leq\sum_{i}\frac{L_{i}\mu^{1+\alpha_{i}}}{(1+\alpha_{i})}d^{\frac{2}{p}}
\end{align*}
where step $1$ follows from Jensen inequality and $0\leq\alpha\leq1$,
step $2$ is from Lemma \ref{L4} below in which if $\xi\sim N_{p}\left(0,I_{d}\right)$
then $d^{\left\lfloor \frac{n}{p}\right\rfloor }\leq E(\left\Vert \xi\right\Vert _{p}^{n})\leq\left[d+\frac{n}{2}\right]^{\frac{n}{p}}$where$\left\lfloor x\right\rfloor $
denotes the largest integer less than or equal to $x$, and the last
step is by simplification when $d$ is large enough and $\mu$ is
small enough. By replacing Assumption \ref{A0} with Assumption \ref{A1},
for $\mu$ is small enough, we immediately get 
\[
\left|U_{\mu}(x)-U(x)\right|{\displaystyle \leq2L\mu^{1+\ell+\alpha}d^{\frac{1+\ell+\alpha}{p}}.}
\]

(ii). We adapt the technique of \citep{nesterov2017random} to $p$-generalized
Gaussian smoothing. Let $y=x+\mu\xi$, then $U_{\mu}(x)$ is rewritten
in another form as 
\begin{align*}
U_{\mu}(x) & =\mathrm{\mathbb{E}}_{\xi}[U(x+\mu\xi)]\\
 & =\frac{1}{\kappa\mu}\int_{\mathbb{R}^{d}}U(y)e^{-\frac{1}{p\mu^{p}}\left\Vert y-x\right\Vert _{p}^{p}}dy.
\end{align*}
Now taking the gradient with respect to $x$ of $U_{\mu}(x)$ gives
\[
\nabla_{x}U_{\mu}(x)=\frac{1}{\kappa\mu}\nabla_{x}\int_{\mathbb{R}^{d}}U(y)e^{-\frac{1}{p\mu^{p}}\left\Vert y-x\right\Vert _{p}^{p}}dy.
\]
By Fubini Theorem with some regularity (i.e. $\mathbb{E}|U(y)|<\infty$),
we can exchange the gradient and integral and get 
\begin{align*}
\nabla_{x}U_{\mu}(x) & =\frac{1}{\kappa\mu}\int_{\mathbb{R}^{d}}\nabla_{x}\left(U(y)e^{-\frac{1}{p\mu^{p}}\left\Vert y-x\right\Vert _{p}^{p}}\right)dy\\
 & =\frac{1}{\kappa\mu}\int_{\mathbb{R}^{d}}U(y)\nabla_{x}\left(e^{-\frac{1}{p\mu^{p}}\left\Vert y-x\right\Vert _{p}^{p}}\right)dy\\
 & =\frac{1}{\kappa\mu}\int_{\mathbb{R}^{d}}U(y)e^{-\frac{1}{p\mu^{p}}\left\Vert y-x\right\Vert _{p}^{p}}\frac{-1}{\mu^{p}}\left\Vert y-x\right\Vert _{p}^{p}\nabla_{x}(\left\Vert y-x\right\Vert _{p})dy\\
 & =\frac{1}{\kappa\mu}\int_{\mathbb{R}^{d}}U(y)e^{-\frac{1}{p\mu^{p}}\left\Vert y-x\right\Vert _{p}^{p}}\frac{1}{\mu^{p}}(y-x)\circ\left|y-x\right|^{p-2}dy.
\end{align*}
where $\circ$ stands for the Hadamard product and $\left|\cdot\right|$
is used for absolute value of each component of the vector $y-x$.
Therefore, by changing variable back to $\xi$, we deduce 
\begin{align*}
\nabla_{x}U_{\mu}(x) & =\frac{1}{\kappa}\int_{\mathbb{R}^{d}}U(x+\mu\xi)e^{-\frac{1}{p}\left\Vert \xi\right\Vert _{p}^{p}}\frac{1}{\mu}\xi\circ\left|\xi\right|^{p-2}d\xi\\
 & =\mathbb{E}_{\xi}\left[\frac{U(x+\mu\xi)\xi\circ\left|\xi\right|^{p-2}}{\mu}\right].
\end{align*}
In addition, if $\xi\sim N_{p}(0,I_{d})$, $\mathbb{E}\left(\xi\right)=\frac{1}{\kappa}\int\xi e^{-\frac{1}{p}\left\Vert \xi\right\Vert _{p}^{p}}d\xi=0$
and then $\nabla_{\xi}\mathbb{E}\left(\xi\right)=0$. Since $\xi e^{-\frac{1}{p}\left\Vert \xi\right\Vert _{p}^{p}}$
is bounded, we can exchange the gradient and the integral and get
\begin{align*}
\nabla_{\xi}\frac{1}{\kappa}\int\xi e^{-\frac{1}{p}\left\Vert \xi\right\Vert _{p}^{p}}d\xi & =\frac{1}{\kappa}\int\nabla_{\xi}\left(\xi e^{-\frac{1}{p}\left\Vert \xi\right\Vert _{p}^{p}}\right)d\xi\\
0 & =\frac{1}{\kappa}\int e^{-\frac{1}{p}\left\Vert \xi\right\Vert _{p}^{p}}d\xi+\frac{1}{\kappa}\int\xi\nabla_{\xi}\left(e^{-\frac{1}{p}\left\Vert \xi\right\Vert _{p}^{p}}\right)d\xi\\
0 & =1-\frac{1}{\kappa}\int\xi\mathrm{e}^{-\frac{1}{p}\left\Vert \xi\right\Vert _{p}^{p}}\left\Vert \xi\right\Vert _{p}^{p-1}\nabla_{\xi}\left(\left\Vert \xi\right\Vert _{p}\right)d\xi\\
0 & =1-\frac{1}{\kappa}\int\xi\cdot\xi\circ\left|\xi\right|^{p-2}e^{-\frac{1}{p}\left\Vert \xi\right\Vert _{p}^{p}}d\xi,
\end{align*}
which implies 
\begin{equation}
\frac{1}{\kappa}\int\xi\cdot\xi\circ\left|\xi\right|^{p-2}e^{-\frac{1}{p}\left\Vert \xi\right\Vert _{p}^{p}}d\xi=1.\label{eq:3}
\end{equation}
On the other hand, we also have $\frac{1}{\kappa}\int e^{-\frac{1}{p}\left\Vert \xi\right\Vert _{p}^{p}}d\xi=1$
so $\nabla_{\xi}\int e^{-\frac{1}{p}\left\Vert \xi\right\Vert _{p}^{p}}d\xi=0.$
By exchange the gradient and the integral and we also get

\begin{align*}
0 & =\nabla_{\xi}\int e^{-\frac{1}{p}\left\Vert \xi\right\Vert _{p}^{p}}d\xi\\
 & =\int\nabla_{\xi}e^{-\frac{1}{p}\left\Vert \xi\right\Vert _{p}^{p}}d\xi\\
 & =\int\nabla_{\xi}\left(e^{-\frac{1}{p}\left\Vert \xi\right\Vert _{p}^{p}}\right)d\xi\\
 & =-\int e^{-\frac{1}{p}\left\Vert \xi\right\Vert _{p}^{p}}\xi\circ\left|\xi\right|^{p-2}d\xi
\end{align*}
which implies that 
\begin{equation}
\mathbb{E}_{\xi}\left[\xi\circ\left|\xi\right|^{p-2}\right]=0.\label{eq:4-1-1}
\end{equation}
From \ref{eq:3} and \ref{eq:4-1-1}, we obtain 
\begin{align*}
\left\Vert \nabla U_{\mu}(x)-\nabla U(x)\right\Vert  & =\left\Vert \frac{1}{\kappa}\int_{\mathbb{R}^{d}}\left[\frac{U(x+\mu\xi)-U(x)}{\mu}-\left\langle \nabla U(x),\xi\right\rangle \right]\xi\circ\left|\xi\right|^{p-2}e^{-\frac{1}{p}\left\Vert \xi\right\Vert _{p}^{p}}d\xi\right\Vert \\
 & \stackrel{_{1}}{\leq}\frac{1}{\kappa\mu}\int_{\mathbb{R}^{d}}\left|U(x+\mu\xi)-U(x)-\mu\left\langle \nabla U(x),\xi\right\rangle \right|e^{-\frac{1}{p}\left\Vert \xi\right\Vert _{p}^{p}}\left\Vert \xi\circ\left|\xi\right|^{p-2}\right\Vert d\xi\\
 & \stackrel{_{2}}{\leq}\sum_{i}\frac{L_{i}\mu^{\alpha_{i}}}{\kappa\left(1+\alpha_{i}\right)}\int_{\mathbb{R}^{d}}\left\Vert \xi\right\Vert ^{\alpha_{i}+1}e^{-\frac{1}{p}\left\Vert \xi\right\Vert _{p}^{p}}\left\Vert \xi\circ\left|\xi\right|^{p-2}\right\Vert d\xi\\
 & =\sum_{i}\frac{L_{i}\mu^{\alpha_{i}}}{\kappa\left(1+\alpha_{i}\right)}\int_{\mathbb{R}^{d}}\left\Vert \xi\right\Vert ^{\alpha_{i}+1}e^{-\frac{1}{p}\left\Vert \xi\right\Vert _{p}^{p}}\left\Vert \xi^{p-1}\right\Vert d\xi,
\end{align*}
where step $1$ follows from Jensen inequality, step $2$ is due to
\ref{eq:4} and the last step follows from component-wise operation
of norm. If $p\leq2$, by using generalized Holder inequality, $\left\Vert \xi^{p-1}\right\Vert $
can be bounded as follow:

\begin{align}
\left\Vert \xi^{p-1}\right\Vert  & \leq\left\Vert \xi^{p-1}\right\Vert _{p}\nonumber \\
 & =\left\Vert \xi^{p-1}\cdot1_{d}\right\Vert _{p}\nonumber \\
 & \stackrel{}{\leq}\left\Vert \xi\right\Vert _{p}^{p-1}\left\Vert 1_{d}\right\Vert _{p}^{2-p}\nonumber \\
 & =\left\Vert \xi\right\Vert _{p}^{p-1}d^{\frac{2-p}{p}}.\label{eq:norm0}
\end{align}
As a result, if $1\leq p\leq2$ we have 
\begin{align*}
\left\Vert \nabla U_{\mu}(x)-\nabla U(x)\right\Vert  & \leq\sum_{i}\frac{L_{i}\mu^{\alpha_{i}}}{\kappa\left(1+\alpha_{i}\right)}\int_{\mathbb{R}^{d}}\left\Vert \xi\right\Vert ^{\alpha_{i}+1}\left\Vert \xi\right\Vert _{p}^{p-1}e^{-\frac{1}{p}\left\Vert \xi\right\Vert _{p}^{p}}d\xi\\
 & \stackrel{_{1}}{\leq}\sum_{i}\frac{L_{i}\mu^{\alpha_{i}}}{\left(1+\alpha_{i}\right)}d^{\frac{2-p}{p}}\mathbb{E}\left[\left\Vert \xi\right\Vert _{p}^{p+\alpha_{i}}\right]\\
 & \stackrel{_{2}}{\leq}\sum_{i}\frac{L_{i}\mu^{\alpha_{i}}}{\left(1+\alpha_{i}\right)}d^{\frac{2-p}{p}}\mathbb{E}\left[\left\Vert \xi\right\Vert _{p}^{2p}\right]^{\frac{p+\alpha}{2p}}\\
 & \stackrel{_{3}}{\leq}\sum_{i}\frac{L_{i}\mu^{\alpha_{i}}}{\left(1+\alpha_{i}\right)}d^{\frac{2-p}{p}}\left(d+p\right)^{\frac{p+\alpha}{p}}\\
 & \stackrel{}{\leq}\sum_{i}L_{i}\mu^{\alpha_{i}}d^{\frac{3}{p}}
\end{align*}
where step $1$ is from $\left\Vert \xi\right\Vert \leq\left\Vert \xi\right\Vert _{p}$,
step $2$ follows from Jensen inequality and $\alpha\leq p$, step
$3$ is due to \ref{eq:4} and in the last two steps we have used
simplification for large enough $d$ and small enough $\mu$. By replacing
Assumption \ref{A0} with Assumption \ref{A1}, for $\mu$ is small
enough, we immediately get 
\[
\left\Vert \nabla U_{\mu}(x)-\nabla U(x)\right\Vert \leq2L\mu^{\alpha}d^{\frac{3}{p}}.
\]

iii) In this case, using Eqs. \ref{eq:4} and \ref{eq:4-1-1}, we
get: 
\[
\nabla U_{\mu}(x)=\frac{1}{\kappa}\int_{\mathbb{R}^{d}}\left[\frac{U(x+\mu\xi)-U(x)}{\mu}\right]\xi\circ\left|\xi\right|^{p-2}e^{-\frac{1}{p}\left\Vert \xi\right\Vert _{p}^{p}}d\xi.
\]
Let $V(x)=U(x+\mu\xi)-U(x)$, from above equation, we obtain

\begin{align*}
 & \left\Vert \nabla U_{\mu}(y)-\nabla U_{\mu}(x)\right\Vert \\
 & =\left\Vert \frac{1}{\mu\kappa}\int_{\mathbb{R}^{d}}\left(V(y)-V(x)\right)e^{-\frac{1}{p}\left\Vert \xi\right\Vert _{p}^{p}}\xi\circ\left|\xi\right|^{p-2}d\xi\right\Vert \\
 & =\left\Vert \frac{1}{\mu\kappa}\int_{\mathbb{R}^{d}}\int_{0}^{1}\left\langle \nabla V\left(ty+\left(1-t\right)x\right),y-x\right\rangle dt\,e^{-\frac{1}{p}\left\Vert \xi\right\Vert _{p}^{p}}\xi\circ\left|\xi\right|^{p-2}d\xi\right\Vert \\
 & =\left\Vert \frac{1}{\mu\kappa}\int_{\mathbb{R}^{d}}\int_{0}^{1}\left\langle \nabla U\left(ty+\left(1-t\right)x+\mu\xi\right)-\nabla U\left(ty+\left(1-t\right)x\right),y-x\right\rangle dt\,e^{-\frac{1}{p}\left\Vert \xi\right\Vert _{p}^{p}}\xi\circ\left|\xi\right|^{p-2}d\xi\right\Vert \\
 & \leq\frac{1}{\mu\kappa}\int_{\mathbb{R}^{d}}\int_{0}^{1}\left\Vert \nabla U\left(ty+\left(1-t\right)x+\mu\xi\right)-\nabla U\left(ty+\left(1-t\right)x\right)\right\Vert \left\Vert y-x\right\Vert dt\,e^{-\frac{1}{p}\left\Vert \xi\right\Vert _{p}^{p}}\left\Vert \xi\circ\left|\xi\right|^{p-2}\right\Vert d\xi\\
 & \leq\sum_{i}\frac{L_{i}}{\mu^{1-\alpha_{i}}\kappa}\int_{\mathbb{R}^{d}}\left\Vert \xi\right\Vert ^{\alpha_{i}}\left\Vert y-x\right\Vert \,e^{-\frac{1}{p}\left\Vert \xi\right\Vert _{p}^{p}}\left\Vert \xi^{p-1}\right\Vert d\xi.
\end{align*}
Since $p\leq2$ we have 
\begin{align*}
 & \left\Vert \nabla U_{\mu}(y)-\nabla U_{\mu}(x)\right\Vert \\
 & \leq\sum_{i}\frac{L_{i}}{\mu^{1-\alpha_{i}}}d^{\frac{2-p}{p}}\mathbb{E}\left(\left\Vert \xi\right\Vert _{p}^{p-1+\alpha}\right)\left\Vert y-x\right\Vert \\
 & \stackrel{_{1}}{\leq}\sum_{i}\frac{L_{i}}{\mu^{1-\alpha_{i}}}d^{\frac{2-p}{p}}\mathbb{E}\left(\left\Vert \xi\right\Vert _{p}^{p}\right)^{\frac{p-1+\alpha}{p}}\left\Vert y-x\right\Vert \\
 & \stackrel{_{2}}{\leq}\sum_{i}\frac{L_{i}}{\mu^{1-\alpha_{i}}}d^{\frac{2-p}{p}}\left(d+\frac{p}{2}\right)^{\frac{p-1+\alpha}{p}}\left\Vert y-x\right\Vert \\
 & \stackrel{}{\leq}\sum_{i}\frac{L_{i}}{\mu^{1-\alpha_{i}}}d^{\frac{2}{p}}\left\Vert y-x\right\Vert ,
\end{align*}
where step $1$ follows from Jensen inequality and $\alpha_{i}\leq1$,
step $2$ is due to \ref{eq:4} and in the last two step is because
of simplification for large enough $d$ and small enough $\mu$. By
replacing Assumption \ref{A0} with Assumption \ref{A1}, for $\mu$
is small enough, we immediately get

\[
\left\Vert \nabla U_{\mu}(y)-\nabla U_{\mu}(x)\right\Vert \leq\frac{L}{\mu^{1-\alpha}}d^{\frac{2}{p}}\left\Vert y-x\right\Vert .
\]

\end{proof}

\section{Proofs under LSI \label{AppB}}

\subsection{Proof of Lemma~\ref{Lem:Initial}}

\begin{lemma}Suppose $\pi=e^{-U}$ satisfies $\alpha$-mixture weakly
smooth. Let $p_{0}=N(0,\frac{1}{L}I)$. Then $H(p_{0}|\pi)\le U(0)-\frac{d}{2}\log\frac{2\Pi e}{L}+\sum_{i}\frac{L_{i}}{1+\alpha_{i}}\left(\frac{d}{L}\right)^{\frac{1+\alpha_{i}}{2}}=O(d).$

\end{lemma}

\begin{proof} \label{BInitial}Since $U$ is mixture weakly smooth,
for all $x\in\mathbb{R}^{d}$ we have 
\begin{align*}
U(x) & \le U(0)+\langle\nabla U(0),x\rangle+\sum_{i}\frac{L_{i}}{1+\alpha_{i}}\Vert x\Vert^{1+\alpha_{i}}\\
 & =U(0)+\sum_{i}\frac{L_{i}}{1+\alpha_{i}}\Vert x\Vert^{1+\alpha_{i}}.
\end{align*}
Let $X\sim\rho=N(0,\frac{1}{L}I)$. Then 
\begin{align*}
\mathbb{E}_{\rho}[U(X)] & \le U(0)+\sum_{i}\frac{L_{i}}{1+\alpha_{i}}\mathbb{E}_{\rho}\left(\Vert x\Vert^{1+\alpha_{i}}\right)\\
 & \leq U(0)+\sum_{i}\frac{L_{i}}{1+\alpha_{i}}\mathbb{E}_{\rho}\left(\Vert x\Vert^{2}\right)^{\frac{1+\alpha_{i}}{2}}\\
 & \leq U(0)+\sum_{i}\frac{L_{i}}{1+\alpha_{i}}\left(\frac{d}{L}\right)^{\frac{1+\alpha_{i}}{2}}.
\end{align*}
Recall the entropy of $\rho$ is $H(\rho)=-\mathbb{E}_{\rho}[\log\rho(X)]=\frac{d}{2}\log\frac{2\Pi e}{L}$.
Therefore, the KL divergence is 
\begin{align*}
\mathbb{E}(\rho|\pi) & =\int\rho\left(\log\rho+U\right)dx\\
 & =-H(\rho)+\mathbb{E}_{\rho}[U]\\
 & \le U(0)-\frac{d}{2}\log\frac{2\Pi e}{L}+\sum_{i}\frac{L_{i}}{1+\alpha_{i}}\left(\frac{d}{L}\right)^{\frac{1+\alpha_{i}}{2}}\\
 & =O(d).
\end{align*}
This is the desired result. \end{proof}

\subsection{Proof of Lemma~\ref{Lem:Initial}}

\begin{lemma}\label{Lem:GradStat-1} Assume $\pi=e^{-U(x)}$ is $\alpha$-mixture
weakly smooth. Then

\[
\mathbb{E}_{\pi}\left[\left\Vert \nabla U(x)\right\Vert ^{2}\right]\le2\left(\sum_{i}L_{i}\right)^{2}d^{\frac{3}{p}},
\]

In particular, if $\pi=e^{-U(x)}$ is $\left(\alpha,\ell\right)$-weakly
smooth. Then

\[
\mathbb{E}_{\pi}\left[\left\Vert \nabla U(x)\right\Vert ^{2\alpha}\right]\le L^{2\alpha}d^{\frac{3-\alpha}{1+\alpha}\alpha},
\]

\[
\mathbb{E}_{\pi}\left[\left\Vert \nabla U(x)\right\Vert ^{2\ell+2\alpha}\right]\le L^{2\left(\ell+\alpha\right)}d^{\frac{3-\alpha}{1+\alpha}\left(\ell+\alpha\right)},
\]
for $d$ sufficiently large.\end{lemma} \begin{proof} Since $\pi$
is stationary distribution, we have 
\[
\frac{d}{dt}\mathbb{E}_{\pi}\left[U_{\mu}\left(x\right)\right]=\int\left(\left(\triangle U_{\mu}\left(x\right)\right)-\left\langle \nabla U\left(x\right),\nabla U_{\mu}\left(x\right)\right\rangle \right)\pi\left(x\right)dx=0.
\]
So

\begin{align*}
\mathbb{E}_{\pi}\left\langle \nabla U\left(x\right),\nabla U_{\mu}\left(x\right)\right\rangle  & =\mathbb{E}_{\pi}\left(\triangle U_{\mu}\left(x\right)\right)\\
 & \stackrel{}{\leq}\sum_{i}\frac{L_{i}}{\mu^{1-\alpha_{i}}}d^{\frac{2}{p}},
\end{align*}
where the last step comes from Lemma \ref{2.1}that $\nabla U_{\mu}\left(x\right)$
is $\sum_{i}\frac{L_{i}}{\mu^{1-\alpha_{i}}}d^{\frac{2}{p}}$-Lipschitz,
$\nabla^{2}U_{\mu}\left(x\right)\preceq\left(\sum_{i}\frac{L_{i}}{\mu^{1-\alpha_{i}}}d^{\frac{2}{p}}\right)\,I$.
In addition,

\begin{align*}
\mathbb{E}_{\pi}\left\langle \nabla U\left(x\right),\nabla U_{\mu}\left(x\right)\right\rangle  & =\mathbb{E}_{\pi}\left[\left\Vert \nabla U(x)\right\Vert ^{2}\right]+\mathbb{E}_{\pi}\left\langle \nabla U\left(x\right),\nabla U_{\mu}\left(x\right)-\nabla U\left(x\right)\right\rangle \\
 & \stackrel{_{1}}{\geq}\mathbb{E}_{\pi}\left[\left\Vert \nabla U(x)\right\Vert ^{2}\right]-\mathbb{E}_{\pi}\left\Vert \nabla U\left(x\right)\right\Vert \left\Vert \nabla U_{\mu}\left(x\right)-\nabla U\left(x\right)\right\Vert \\
 & \stackrel{}{\geq}\mathbb{E}_{\pi}\left[\left\Vert \nabla U(x)\right\Vert ^{2}\right]-\sqrt{\mathbb{E}_{\pi}\left[\left\Vert \nabla U(x)\right\Vert ^{2}\right]}\sum_{i}L_{i}\mu^{\alpha_{i}}d^{\frac{3}{p}},
\end{align*}
where step $1$ follows from Young inequality and the last step comes
from Cauchy inequality and Lemma \ref{2.1} . From quadratic inequality

\[
\mathbb{E}_{\pi}\left[\left\Vert \nabla U(x)\right\Vert ^{2}\right]-\sqrt{\mathbb{E}_{\pi}\left[\left\Vert \nabla U(x)\right\Vert ^{2}\right]}\sum_{i}L_{i}\mu^{\alpha_{i}}d^{\frac{3}{p}}\leq\sum_{i}\frac{L_{i}}{\mu^{1-\alpha_{i}}}d^{\frac{2}{p}}
\]
and since $\sqrt{\mathbb{E}_{\pi}\left[\left\Vert \nabla U(x)\right\Vert ^{2}\right]}\geq0$
we obtain

\begin{align*}
\sqrt{\mathbb{E}_{\pi}\left[\left\Vert \nabla U(x)\right\Vert ^{2}\right]} & \leq\frac{1}{2}\left[\sqrt{\left(\sum_{i}L_{i}\mu^{\alpha_{i}}\right)^{2}d^{\frac{6}{p}}+4\sum_{i}\frac{L_{i}}{\mu^{1-\alpha_{i}}}d^{\frac{2}{p}}}+\sum_{i}L_{i}\mu^{\alpha_{i}}d^{\frac{3}{p}}\right].
\end{align*}

Simply choose $\mu=1,$ we get 
\begin{align*}
\mathbb{E}_{\pi}\left[\left\Vert \nabla U(x)\right\Vert ^{2}\right] & \leq\frac{1}{4}\left[\sqrt{\left(\sum_{i}L_{i}\right)^{2}d^{\frac{6}{p}}+4\left(\sum_{i}L_{i}\right)d^{\frac{2}{p}}}+\sum_{i}L_{i}d^{\frac{3}{p}}\right]^{2}\\
 & \leq2\left(\sum_{i}L_{i}\right)^{2}d^{\frac{3}{p}},
\end{align*}

for large enough $d.$ If we replace Assumption \ref{A0} by Assumption
\ref{A1}, we can choose $p=2$ and $\mu=\frac{1}{d^{\frac{2}{1+\alpha}}}$,
we deduce 
\begin{align*}
\mathbb{E}_{\pi}\left[\left\Vert \nabla U(x)\right\Vert ^{2}\right] & \leq\frac{1}{4}\left[\sqrt{L^{2}\mu^{2\alpha}d^{\frac{6}{p}}+4\frac{Ld^{\frac{2}{p}}}{\mu^{1-\alpha}}}+L\mu^{\alpha}d^{\frac{3}{p}}\right]^{2}\\
 & \leq L^{2}d^{\frac{3-\alpha}{1+\alpha}},
\end{align*}
for $d$ large enough as desired. Since $\alpha\leq1$, $x\rightarrow x^{\alpha}$
is concave function. By Jensen inequality 
\begin{align*}
\mathbb{E}_{\pi}\left[\left\Vert \nabla U(x)\right\Vert ^{2\alpha}\right] & \leq\left(\mathbb{E}_{\pi}\left[\left\Vert \nabla U(x)\right\Vert ^{2}\right]\right)^{\alpha}\\
 & \leq L^{2\alpha}d^{\frac{3-\alpha}{1+\alpha}\alpha}.
\end{align*}
Similarly, $\ell+\alpha\leq1,$by Jensen inequality we also have

\begin{align*}
\mathbb{E}_{\pi}\left[\left\Vert \nabla U(x)\right\Vert ^{2\ell+2\alpha}\right] & \leq\left(\mathbb{E}_{\pi}\left[\left\Vert \nabla U(x)\right\Vert ^{2}\right]\right)^{\ell+\alpha}\\
 & \leq L^{2\left(\ell+\alpha\right)}d^{\frac{3-\alpha}{1+\alpha}\left(\ell+\alpha\right)},
\end{align*}
as desired.

\end{proof}

\subsection{Proof of Lemma \ref{Lem:OneStep}\label{Proof-of-LemmaOneStep}}

\begin{lemma} Suppose $\pi$ is $\gamma-$log-Sobolev, $\alpha$-mixture
weakly smooth with $\max\left\{ L_{i}\right\} =L\geq1$. If $0<\eta\le\left(\frac{\gamma}{9N^{\frac{3}{2}}L^{3}}\right)^{\frac{1}{\alpha}}$
, then along each step of ULA~\eqref{Eq:ULA},

\begin{align}
H(p_{k+1}|\pi)\le e^{-\gamma\eta}H(p_{k}|\pi)+2\eta^{\alpha+1}D_{3},\label{Eq:Main1-2-1-1-1-1}
\end{align}
where $D_{3}=\sum_{i}10N^{3}L^{6}+16NL^{4}+8N^{2}L^{4}d^{\frac{3}{p}}+4NL^{2}d$.

In particular, if$\pi$ is $\gamma-$log-Sobolev, $\left(\alpha,\ell\right)-$weakly
smooth with $0<\alpha+\ell\leq1$. If $0<\eta\le\left(\frac{\gamma}{2L^{1+\alpha}}\right)^{\frac{1}{\alpha}}$,
then along each step of ULA~\eqref{Eq:ULA},

\begin{align}
H(p_{k+1}|\pi)\le e^{-\gamma\eta}H(p_{k}|\pi)+2\eta^{\alpha+1}D_{3}^{\prime},\label{Eq:Main1-2-1-1-1}
\end{align}
where $D_{3}^{\prime}=16L^{2+2\alpha+2\ell}+4L^{2+2\alpha}d^{\frac{3-\alpha}{1+\alpha}\left(\alpha+\ell\right)}+4L^{2}d^{\alpha+\ell}$.

\end{lemma}\begin{proof}We adapt the proof of \citep{vempala2019rapid}.
First, recall that the discretization of the LMC is 
\begin{center}
$x_{k,t}\stackrel{}{=}x_{k}-t\nabla U(x_{k})+\sqrt{2t}\,z_{k}$, 
\par\end{center}

where $z_{k}\sim N(0,I)$ is independent of $x_{k}$. Let $x_{k}\sim p_{k}$
and $x^{\ast}\sim\pi$ with an optimal coupling $(x_{k},x^{\ast})$
so that $\mathbb{E}[\|x_{k}-x^{\ast}\|^{2}]=W_{2}(p_{k},\pi)^{2}$.
Let $D_{1i}=8NL_{i}^{2+2\alpha_{i}}\left(\left(\sum_{j}L_{i}\right)^{2}+1\right)+16L_{i}^{2+2\alpha_{i}}+8L_{i}^{2}\left(\sum_{i}L_{i}\right)^{2}d^{\frac{3}{p}}+4L_{i}^{2}d^{\alpha_{i}}$,
we deduce

\begin{align}
 & L_{i}^{2}E_{p_{k}}\left[\left\Vert -t\nabla U(x_{k})+\sqrt{2t}z_{k}\right\Vert ^{2\alpha_{i}}\right]\nonumber \\
 & \stackrel{_{1}}{\leq}2L_{i}^{2}t^{2\alpha_{i}}\mathbb{E}_{p_{k}}\left[\left\Vert \nabla U(x_{k})\right\Vert ^{2\alpha_{i}}\right]+4L_{i}^{2}t^{\alpha_{i}}\mathbb{E}_{p_{k}}\left[\left\Vert z_{k}\right\Vert ^{2\alpha_{i}}\right]\nonumber \\
 & \stackrel{_{2}}{\leq}2L_{i}^{2}t^{2\alpha_{i}}\mathbb{E}_{p_{k}}\left[\left\Vert \nabla U(x_{k})\right\Vert ^{2\alpha_{i}}\right]+4L_{i}^{2}t^{\alpha_{i}}\mathbb{E}_{p_{k}}\left[\left\Vert z_{k}\right\Vert ^{2}\right]^{\alpha_{i}}\nonumber \\
 & \stackrel{_{3}}{\leq}4L_{i}^{2}t^{2\alpha_{i}}\mathbb{E}\left[\left\Vert \nabla U(x_{k})-\nabla U(x^{*})\right\Vert ^{2\alpha_{i}}+\left\Vert \nabla U(x^{*})\right\Vert ^{2\alpha_{i}}\right]+4L_{i}^{2}t^{\alpha_{i}}d^{\alpha_{i}}\nonumber \\
 & \stackrel{_{4}}{\leq}4L_{i}^{2}t^{2\alpha_{i}}\mathrm{\mathbb{E}}\left(\sum_{i}L_{i}\left\Vert x_{k}-x^{*}\right\Vert ^{\alpha_{i}}\right)^{2\alpha_{i}}+4L_{i}^{2}t^{2\alpha_{i}}\mathbb{E}\left\Vert \nabla U(x^{*})\right\Vert ^{2\alpha_{i}}+4L_{i}^{2}t^{\alpha_{i}}d^{\alpha_{i}}\nonumber \\
 & \leq8L_{i}^{2+2\alpha_{i}}t^{2\alpha_{i}}N\sum_{j}L_{i}^{2\alpha_{i}}\mathrm{\mathbb{E}}\left[\left\Vert x_{k}-x^{*}\right\Vert ^{2\alpha_{j}\alpha_{i}}\right]+4L_{i}^{2}t^{2\alpha}\mathbb{E}\left\Vert \nabla U(x^{*})\right\Vert ^{2}\nonumber \\
 & +4L_{i}^{2}t^{2\alpha}+4L_{i}^{2}t^{\alpha}d^{\alpha}\nonumber \\
 & \stackrel{_{5}}{\leq}8NL_{i}^{2+2\alpha_{i}}t^{2\alpha_{i}}\left(\left(\sum_{j}L_{i}\right)^{2}+1\right)\mathrm{\mathbb{E}}\left[1+\left\Vert x_{k}-x^{*}\right\Vert ^{2}\right]+4L_{i}^{2}t^{2\alpha}\mathbb{E}\left\Vert \nabla U(x^{*})\right\Vert ^{2}\nonumber \\
 & +4L_{i}^{2}t^{2\alpha}+4L_{i}^{2}t^{\alpha}d^{\alpha}\nonumber \\
 & \stackrel{}{\leq}8NL_{i}^{2+2\alpha_{i}}\eta^{2\alpha}\left(\left(\sum_{j}L_{i}\right)^{2}+1\right)\mathrm{\mathbb{E}}\left[\left\Vert x_{k}-x^{*}\right\Vert ^{2}\right]\nonumber \\
 & +\left(8NL_{i}^{2+2\alpha_{i}}\left(\left(\sum_{j}L_{i}\right)^{2}+1\right)+16L_{i}^{2+2\alpha_{i}}+8L_{i}^{2}\left(\sum_{i}L_{i}\right)^{2}d^{\frac{3}{p}}+4L_{i}^{2}d^{\alpha_{i}}\right)\eta^{\alpha_{i}}\nonumber \\
 & \leq\frac{16N}{\gamma}\left(\left(\sum_{j}L_{i}\right)^{2}+1\right)L^{2+2\alpha_{i}}\eta^{2\alpha_{i}}H(p_{k}|\pi)+D_{1i}\eta^{\alpha_{i}},\label{eq:Norm}
\end{align}
where step $1$ follows from Lemma \ref{L21} in Appendix F, step
$2$ is from $\alpha\leq1$ and Jensen's inequality, step $3$ comes
from normal distribution, and step $4$ follows our Assumption \ref{A1},
and in step $5$ we have used $\alpha_{i}\leq1$ and the last step
is due to Talagrand inequality which comes from log-Sobolev inequality
and Lemma \ref{L4} in Appendix F below. Similarly, we get

\begin{align}
 & \mathrm{\mathbb{E}}_{p_{kt}}\left\Vert \nabla U(x_{k})-\nabla U(x_{k,t})\right\Vert ^{2}\nonumber \\
 & \stackrel{_{1}}{\leq}\sum_{i}L_{i}^{2}\mathrm{\mathbb{E}}_{p_{kt}}\left\Vert \tilde{x}_{k,t}-x_{k}\right\Vert ^{2\alpha_{i}}\nonumber \\
 & =\sum_{i}L_{i}^{2}\mathrm{\mathbb{E}}_{p_{k}}\left\Vert -t\nabla U(x_{k})+\sqrt{2t}z_{k}\right\Vert ^{2\alpha_{i}}\nonumber \\
 & \stackrel{_{2}}{\leq}\sum_{i}\frac{16N}{\gamma}\left(\left(\sum_{j}L_{i}\right)^{2}+1\right)L^{2+2\alpha_{i}}\eta^{2\alpha_{i}}H(p_{k}|\pi)+\left(\sum_{i}D_{1i}\eta^{\alpha_{i}}\right)\nonumber \\
 & \stackrel{_{3}}{\leq}\frac{20N^{3}}{\gamma}L^{6}\eta^{2\alpha}H(p_{k}|\pi)+D_{3}\eta^{\alpha}\label{eq:Main}
\end{align}
where step $1$ follows from Assumption \ref{A1}, step $2$ comes
from similar reasoning as equation (\ref{eq:Norm}), and the last
step comes from $\eta\leq\frac{1}{L}$ and $\eta\leq1$ and definition
of $D_{3}$. Therefore, from \citep{vempala2019rapid} Lemma 3, the
time derivative of KL divergence along LMC is bounded by 
\begin{align*}
\frac{d}{dt}H\left(p_{k,t}|\pi\right) & \leq-\frac{3}{4}I\left(p_{k,t}|\pi\right)+\mathbb{E}_{p_{kt}}\left[\left\Vert \nabla U(x_{k,t})-\nabla U(x_{k})\right\Vert ^{2}\right]\\
 & \leq-\frac{3}{4}I(p_{k}|\pi)+\frac{20N^{3}}{\gamma}L^{6}\eta^{2\alpha}H(p_{k}|\pi)+D_{3}\eta^{\alpha}\\
 & \leq-\mathrm{\frac{3\gamma}{2}}H(p_{k,t}|\pi)+\frac{20N^{3}}{\gamma}L^{6}\eta^{2\alpha}H(p_{k}|\pi)+D_{3}\eta^{\alpha},
\end{align*}
where in the last inequality we have used the definition \ref{D1}
of LSI inequality. Multiplying both sides by $e^{\frac{3\gamma}{2}t}$,
and integrating both sides from $t=0$ to $t=\eta$ we obtain 
\begin{align}
 & e^{\frac{3\gamma}{2}\eta}H(p_{k+1}|\pi)-H(p_{k}|\pi)\nonumber \\
 & \leq2\left(\frac{e^{\frac{3\gamma}{2}\eta}-1}{3\gamma}\right)\left(\frac{20N^{3}}{\gamma}L^{6}\eta^{2\alpha}H(p_{k}|\pi)+D_{3}\eta^{\alpha}\right)\\
 & \leq2\eta\left(\frac{20N^{3}}{\gamma}L^{6}\eta^{2\alpha}H(p_{k}|\pi)+D_{3}\eta^{\alpha}\right)
\end{align}
where the last line holds by $e^{c}\leq1+2c$ for $0<c=\frac{3\gamma}{2}\eta<1$.
Rearranging the term of the above inequality and using the facts that
$1+\eta^{1+2\alpha}\frac{40N^{3}}{\gamma}L^{6}\leq1+\frac{\gamma\eta}{2}\leq e^{\frac{\gamma\eta}{2}}$
when $\eta\leq\left(\frac{\gamma}{9N^{\frac{3}{2}}L^{3}}\right)^{\frac{1}{\alpha}}$
and $e^{-\frac{3\gamma}{2}\eta}\leq1$ leads to 
\begin{align}
H(p_{k+1}|\pi) & \leq e^{-\frac{3\gamma}{2}\eta}\left(1+\eta^{1+2\alpha}\frac{40N^{3}}{\gamma}L^{6}\right)H(p_{k}|\pi)+2\eta^{\alpha+1}D_{3}\nonumber \\
 & \leq e^{-\gamma\eta}H(p_{k}|\pi)+2\eta^{\alpha+1}D_{3}.
\end{align}
as desired. \end{proof}

\subsection{Proof of Theorem \ref{T1}\label{Proof-of-Theorem1}}

\begin{thm} Suppose $\pi$ is $\gamma-$log-Sobolev, $\alpha$-mixture
weakly smooth with $\max\left\{ L_{i}\right\} =L\geq1$, and for any
$x_{0}\sim p_{0}$ with $H(p_{0}|\pi)=C_{0}<\infty$, the iterates
$x_{k}\sim p_{k}$ of ULA~ with step size 
\begin{equation}
\eta\le\min\left\{ 1,\frac{1}{4\gamma},\left(\frac{\gamma}{9N^{\frac{3}{2}}L^{3}}\right)^{\frac{1}{\alpha}}\right\} 
\end{equation}
satisfies 
\begin{align}
H(p_{k}|\pi)\le e^{-\frac{3\gamma}{2}\eta k}H(p_{0}|\pi)+2\eta^{\alpha+1}D_{3},\label{Eq:Main1a}
\end{align}
where $D_{3}=\sum_{i}10N^{3}L^{6}+16NL^{4}+8N^{2}L^{4}d^{\frac{3}{p}}+4NL^{2}d$.
Then, for any $\epsilon>0$, to achieve $H(p_{k}|\pi)<\epsilon$,
it suffices to run LMC with step size 
\begin{equation}
\eta\le\min\left\{ 1,\frac{1}{4\gamma},\left(\frac{\gamma}{9N^{\frac{3}{2}}L^{3}}\right)^{\frac{1}{\alpha}},\left(\frac{3\epsilon\gamma}{16D_{3}}\right)^{\frac{1}{\alpha}}\right\} 
\end{equation}
for $k\ge\frac{1}{\gamma\eta}\log\frac{2H\left(p_{0}|\pi\right)}{\epsilon}$
iterations.\end{thm} \begin{proof}Applying inequality \ref{Eq:Main1a}
recursively, and using the inequality $1-e^{-c}\ge\frac{3}{4}c$ for
$0<c=\gamma\eta\le\frac{1}{4}$ we obtain 
\begin{align}
H(p_{k}|\pi) & \le\,e^{-\gamma\eta k}H(p_{0}|\pi)+\frac{2\eta^{\alpha+1}D_{3}}{1-e^{-\gamma\eta}}\nonumber \\
 & \le\,e^{-\gamma\eta k}H(p_{0}|\pi)+\frac{2\eta^{\alpha+1}D_{3}}{\frac{3}{4}\gamma\eta}\nonumber \\
 & \le\,e^{-\gamma\eta k}H(p_{0}|\pi)+\frac{8\eta^{\alpha}D_{3}}{3\gamma}.
\end{align}
Note that last inequality holds if we choose $\eta$ such that it
satisfies 
\[
\eta\le\min\left\{ 1,\frac{1}{4\gamma},\left(\frac{\gamma}{9N^{\frac{3}{2}}L^{3}}\right)^{\frac{1}{\alpha}}\right\} .
\]
Given $\epsilon>0$, if we further assume $\eta\le\left(\frac{3\epsilon\gamma}{16D_{3}}\right)^{\frac{1}{\alpha}}$,
then the above implies $H(p_{k}|\pi)\le e^{-\gamma\eta k}H(p_{0}|\pi)+\frac{\epsilon}{2}.$
This means for $k\ge\frac{1}{\gamma\eta}\log\frac{2H\left(p_{0}|\pi\right)}{\epsilon},$
we have $H(p_{k}|\pi)\le\frac{\epsilon}{2}+\frac{\epsilon}{2}=\epsilon$,
as desired. \end{proof}

\setcounter{lemma}{0}

\section{Proof of sampling via smoothing potential\label{AppC}}

\subsection{Proof of Lemma~\ref{3.3.1}\label{Proof-of-Lemma3.3.1}}

\begin{lemma} For any $x_{k}\in\mathbb{R}^{d}$, then $g_{\mu}(x_{k},\zeta_{k})=\nabla U_{\mu}(x_{k})+\zeta_{k}$
is an unbiased estimator of $\nabla U_{\mu}$ such that 
\begin{align*}
\mathrm{Var}\left[g_{\mu}(x_{k},\zeta_{k})\right] & \leq4N^{2}L^{2}\mu^{2\alpha}d^{\frac{2\alpha}{p}}.
\end{align*}
\end{lemma}\begin{proof} Recall that by definition of $U_{\mu}$,
we have $\nabla U_{\mu}(x)=\mathrm{\mathrm{\mathbb{E}}}_{\zeta}[U(x+\mu\mathrm{\zeta})]$,
where $\mathrm{\zeta}\sim N_{p}(0,I_{d\times d})$, and is independent
of $\zeta_{1}$. Clearly, $\mathrm{E}_{\mathrm{\mathrm{\zeta_{1}}}}[g(x,\mathrm{\zeta_{1}})]=\nabla U_{\mu}(x)$.
We now proceed to bound the variance of $g(x,\zeta_{1})$. We have:
\begin{align*}
~ & \mathrm{\mathbb{E}}_{\mathrm{\zeta_{1}}}[\Vert\nabla U_{\mu}(x)-g(x,\zeta_{1})\Vert_{2}^{2}]\\
 & \leq\mathrm{\mathbb{E}}_{\zeta_{1}}[\Vert\mathrm{E}_{\zeta}[U(x+\mu\mathrm{\zeta})]-\nabla U(x+\mu\mathrm{\zeta_{1}})\Vert^{2}]\text{ }\\
~ & \leq\mathrm{\mathbb{E}}_{\zeta_{1},\mathrm{\zeta}}[\Vert\nabla U(x+\mu\mathrm{\zeta})-\nabla U(x+\mu\mathrm{\zeta_{1}})\Vert^{2}].\\
 & \leq N\sum_{i}L_{i}^{2}\mathrm{\mathbb{E}}_{\mathrm{\zeta_{1}},\mathrm{\zeta}}[\Vert\mu(\mathrm{\zeta}-\mathrm{\zeta_{1}})\Vert^{2\alpha_{i}}\\
 & \leq N\sum_{i}L_{i}^{2}\mu^{2\alpha_{i}}\mathrm{\mathbb{E}}_{\zeta_{1},\mathrm{\zeta}}[\Vert\mathrm{\zeta}-\mathrm{\zeta_{1}}\Vert^{2\alpha_{i}}]\\
 & \leq2N\sum_{i}L_{i}^{2}\mu^{2\alpha_{i}}\left(\mathrm{\mathbb{E}}\left[\Vert\mathrm{\zeta}\Vert^{2\alpha_{i}}\right]+\mathrm{\mathbb{E}}\left[\Vert\mathrm{\zeta_{1}}\Vert^{2\alpha_{i}}\right]\right)\\
 & \leq2N\sum_{i}L_{i}^{2}\mu^{2\alpha_{i}}\left(\left(\mathrm{\mathbb{E}}\left[\Vert\mathrm{\zeta}\Vert^{2}\right]\right)^{\alpha_{i}}+\left(\mathrm{\mathbb{E}}\left[\Vert\zeta_{1}\Vert^{2}\right]\right)^{\alpha_{i}}\right)\\
 & \leq4N\sum_{i}L_{i}^{2}\mu^{2\alpha_{i}}d^{\frac{2\alpha_{i}}{p}}\\
 & \leq4N^{2}L^{2}\mu^{2\alpha}d^{\frac{2\alpha}{p}},
\end{align*}
as claimed. \end{proof}

\subsection{Proof of Lemma~\ref{Theorem3.3.1}\label{Proof-of-Theorem3.3.1}}

Before proving Theorem \ref{Theorem3.3.1}, we need an additional
lemma. \begin{lemma}{[}\citep{vempala2019rapid} modified Lemma 3{]}
Suppose $x_{k,t}$ is the interpolation of the discretized process
\eqref{cont}. Let $p_{k,t}$, $p_{kt}$ and $p_{kt\zeta}$ denote
its distribution, the joint distribution of $x_{k,t}$ and $x_{k}$
and the joint distribution of $x_{k,t}$, $x_{k}$ and $\zeta$ respectively.
Here $g(x_{k},\zeta)$ is an estimate of $\nabla U(x_{k})$ with noise
$\zeta$ such that $E_{\zeta}g(x_{k},\zeta)=\nabla U(x_{k})$. Then
\begin{equation}
{\displaystyle \frac{d}{dt}H\left(p_{k,t}|\pi_{\mu}\right)\leq-\frac{3}{4}I\left(p_{k,t}|\pi_{\mu}\right)+\mathbb{E}_{p_{kt\zeta}}\left[\left\Vert \nabla U(x_{k,t})-g(x_{k},\zeta)\right\Vert ^{2}\right]}.
\end{equation}
\end{lemma} \begin{proof} The steps follow exactly as in Lemma 3
and we provide the proof here for completeness. For each $t>0$, let
$p_{k\zeta|t}(x_{k},\zeta)$ denote the distributions of $x_{k}$
and $\zeta$ conditioned on $x_{k,t}$ and $p_{t|k\zeta}(x_{k,t})$
denote the distributions of $x_{k,t}$ conditioned on $x_{k}$ and
$\zeta$. Following Fokker-Planck equation, we have 
\begin{equation}
\frac{\partial p_{t|k\zeta}(x_{k,t})}{\partial t}=\nabla\cdot\left(p_{t|k\zeta}(x_{k,t})g(x_{k},\zeta)\right)+\triangle p_{t|k\zeta}(x_{k,t}),
\end{equation}
which integrating with respect to $x_{k}$ and $\zeta$ achieves

\begin{align}
\frac{\partial p_{k,t}(x)}{\partial t} & =\int\int\frac{\partial p_{t|k\zeta}(x)}{\partial t}p_{k\zeta}(x_{k},\zeta)dx_{k}d\zeta\nonumber \\
 & =\int\int\left(\nabla\cdot\left(p_{t|k\zeta}(x_{k,t})g(x_{k},\zeta)\right)+\triangle p_{t|k\zeta}(x_{k,t})\right)dx_{k}d\zeta\nonumber \\
 & =\int\int\left(\nabla\cdot\left(p_{t|k\zeta}(x_{k,t})g(x_{k},\zeta)\right)\right)+\triangle p_{k,t}(x)\nonumber \\
 & =\nabla\cdot(p_{k,t}(x)\int\int p_{k\zeta|t}(x_{k})g(x_{k},\zeta)dx_{k}d\zeta)+\triangle p_{k,t}(x)\\
 & =\nabla\cdot\ (p_{k,t}(x)\mathrm{\mathbb{E}}_{p_{k\zeta|t}}[g(x_{k},\zeta)|x_{k,t}=x])+\triangle p_{k,t}(x).
\end{align}
Combining with $\int p_{t}\frac{\partial}{\partial t}\log\frac{p_{t}}{\pi_{\mu}}\,dx=\int\frac{\partial p_{t}}{\partial t}\,dx=\frac{d}{dt}\int p_{t}\,dx=0$,
we get the following inequality for time derivative of KL-divergence.

\begin{align}
\frac{d}{dt}H\left(p_{k,t}|\pi_{\mu}\right) & =\frac{d}{dt}\int p_{k,t}(x)\log\left(\frac{p_{k,t}(x)}{\pi_{\mu}(x)}\right)dx\nonumber \\
 & =\int\frac{\partial p_{k,t}}{\partial t}(x)\log\left(\frac{p_{k,t}(x)}{\pi_{\mu}(x)}\right)dx\nonumber \\
 & =\int\left[\nabla\cdot\left(p_{k,t}(x)\mathrm{\mathbb{E}}_{p_{k\zeta|t}}[g(x_{k},\zeta)|x_{k,t}=x]\right)\right]\log\left(\frac{p_{k,t}(x)}{\pi_{\mu}(x)}\right)dx\nonumber \\
 & +\int\left[\triangle p_{k,t}(x)\right]\log\left(\frac{p_{k,t}(x)}{\pi_{\mu}(x)}\right)dx\nonumber \\
 & \stackrel{\left(i\right)}{=}\int\left[\nabla\cdot\left(p_{k,t}(x)\mathrm{\mathbb{E}}_{p_{k\zeta|t}}[g(x_{k},\zeta)|x_{k,t}=x]\right)\right]\log\left(\frac{p_{k,t}(x)}{\pi_{\mu}(x)}\right)dx\nonumber \\
 & +\int\left[\nabla\cdot\left(\nabla\log\left(\frac{p_{k,t}(x)}{\pi_{\mu}(x)}\right)-\nabla U(x)\right)\right]\log\left(\frac{p_{k,t}(x)}{\pi_{\mu}(x)}\right)dx\nonumber \\
 & \stackrel{\left(ii\right)}{=}-\int p_{k,t}(x)\left\langle \mathrm{\mathbb{E}}_{p_{k\zeta|t}}[g(x_{k},\zeta)|x_{k,t}=x],\ \nabla\log\left(\frac{p_{k,t}(x)}{\pi_{\mu}(x)}\right)\right\rangle dx\nonumber \\
 & -\int p_{k,t}(x)\left\langle \nabla\log\left(\frac{p_{k,t}(x)}{\pi_{\mu}(x)}\right)-\nabla U(x),\ \nabla\log\left(\frac{p_{k,t}(x)}{\pi_{\mu}(x)}\right)\right\rangle dx\nonumber \\
 & =-I\left(p_{k,t}|\pi_{\mu}\right)\nonumber \\
 & +\int p_{k,t}(x)\left\langle \nabla U(x)-\mathrm{\mathbb{E}}_{p_{k\zeta|t}}[g(x_{k},\zeta)|x_{k,t}=x],\ {\displaystyle \nabla\log\left(\frac{p_{k,t}(x)}{\pi_{\mu}(x)}\right)}\right\rangle dx\nonumber \\
 & =-I\left(p_{k,t}|\pi_{\mu}\right)+\mathrm{\mathbb{E}}_{p_{kt\zeta}}\left\langle \nabla U(x_{k,t})-g(x_{k},\zeta),\ {\displaystyle \nabla\log\left(\frac{p_{k,t}(x)}{\pi_{\mu}(x)}\right)}\right\rangle \nonumber \\
 & \stackrel{\left(iii\right)}{\leq}-I\left(p_{k,t}|\pi_{\mu}\right)\nonumber \\
 & +\mathrm{E}_{p_{kt\zeta}}\left\Vert \nabla U(x_{k,t})-g(x_{k},\zeta)\right\Vert ^{2}+\frac{1}{4}\mathrm{\mathbb{E}}_{p_{k,t}}\left\Vert \nabla\log\left(\frac{p_{k,t}(x)}{\pi_{\mu}(x)}\right)\right\Vert ^{2}\nonumber \\
 & =-\frac{3}{4}I\left(p_{k,t}|\pi_{\mu}\right)+\mathrm{\mathbb{E}}_{p_{kt\zeta}}\left\Vert \nabla U(x_{k,t})-g(x_{k},\zeta)\right\Vert ^{2}
\end{align}
in which equality $\left(i\right)$ is follows from $\triangle p_{k,t}=\nabla\cdot(\nabla p_{k,t})$,
equality $\left(ii\right)$ follows from the divergence theorem, inequality
$\left(iii\right)$ follows from $\left\langle u,\ v\right\rangle {\displaystyle \leq\Vert u\Vert^{2}+\frac{1}{4}\Vert v\Vert^{2}}$,
and in the last step, the expectation is taken with respect to both
$x_{k}$ ,$x_{k,t}$ and $\zeta.$ \end{proof} We now ready to state
and prove Theorem \ref{Theorem3.3.1}. \begin{thm} Suppose $\pi_{\mu}$
is $\gamma_{1}-$log-Sobolev, $\alpha$-mixture weakly smooth, $L=1\vee\max\left\{ L_{i}\right\} $,
and for any $x_{0}\sim p_{0}$ with $H(p_{0}|\pi)=C_{0}<\infty$,
the iterates $x_{k}\sim p_{k}$ of ULA~ with step size 
\begin{equation}
\eta\le\min\left\{ 1,\frac{1}{4\gamma},\left(\frac{\gamma_{1}}{13N^{\frac{3}{2}}L^{3}}\right)^{\frac{1}{\alpha}}\right\} 
\end{equation}
satisfies 
\begin{align}
H(p_{k}|\pi_{\mu})\le e^{-\frac{3\gamma_{1}}{2}\eta k}H(p_{0}|\pi_{\mu})+2\eta^{\alpha+1}D_{4},\label{Eq:Main1a-1}
\end{align}
where $D_{4}=\sum_{i}10N^{3}L^{6}+16NL^{4}+8N^{2}L^{4}d^{\frac{3}{p}}+4NL^{2}d+8N^{2}L^{2}d^{\frac{2\alpha}{p}}$.
Then, for any $\epsilon>0$, to achieve $H(p_{k}|\pi)<\epsilon$,
it suffices to run LMC with step size 
\begin{equation}
\eta\le\min\left\{ 1,\frac{1}{4\gamma_{1}},\left(\frac{\gamma_{1}}{13N^{\frac{3}{2}}L^{3}}\right)^{\frac{1}{\alpha}},\left(\frac{3\epsilon\gamma_{1}}{16D_{4}}\right)^{\frac{1}{\alpha}}\right\} 
\end{equation}
for $k\ge\frac{2}{\gamma_{1}\eta}\log\frac{3H\left(p_{0}|\pi_{\mu}\right)}{\epsilon}$
iterations.

\end{thm}\begin{proof}We adapt the proof of \citep{vempala2019rapid}.
First, recall that the discretization of the ULA is 
\begin{center}
$x_{k,t}\stackrel{}{=}x_{k}-\eta g(x_{k},\zeta)+\sqrt{2\eta}\,z_{k}$, 
\par\end{center}

where $z_{k}\sim N(0,I)$ is independent of $x_{k}$. Let $x_{k}\sim p_{k}$
and $x^{\ast}\sim\pi$ with an optimal coupling $(x_{k},x^{\ast})$
so that $\mathbb{E}[\|x_{k}-x^{\ast}\|^{2}]=W_{2}(p_{\mu,k},\pi_{\mu})^{2}$.
Choosing $\mu=\sqrt{\eta}$, we have

\begin{align*}
 & \mathrm{\mathbb{E}}_{p_{kt\zeta}}\left\Vert \nabla U(x_{k,t})-g(x_{k},\zeta)\right\Vert ^{2}\\
 & \stackrel{_{1}}{\leq}2\left[\mathrm{\mathbb{E}}_{p_{kt\zeta}}\left\Vert \nabla U(x_{k,t})-\nabla U(x_{k})\right\Vert ^{2}+\left\Vert \nabla U(x_{k})-g(x_{k},\zeta)\right\Vert ^{2}\right]\\
 & \stackrel{_{2}}{\leq}\frac{40N^{3}}{\gamma_{1}}L^{6}\eta^{2\alpha}H(p_{k}|\pi_{\mu})+D_{3}\eta^{\alpha}+8N^{2}L^{2}\mu^{2\alpha}d^{\frac{2\alpha}{p}}\\
 & \stackrel{}{\leq}\frac{40N^{3}}{\gamma_{1}}L^{6}\eta^{2\alpha}H(p_{k}|\pi_{\mu})+D_{4}\eta^{\alpha},
\end{align*}
where step 1 follows from Young inequality and Assumption 2, step
$2$ comes from equation (\ref{eq:Main}) , and the last step comes
from $\eta\leq\frac{1}{L}$ and $\eta\leq1$ and the definition of
$D_{4}$. Therefore, from Lemma \ref{Theorem3.3.1}, the time derivative
of KL divergence along LMC is bounded by 
\begin{align}
\frac{d}{dt}H\left(p_{k,t}|\pi_{\mu}\right) & \leq-\frac{3}{4}I(p_{k,t}|\pi_{\mu})+\frac{40N^{3}}{\gamma_{1}}L^{6}\eta^{2\alpha}H(p_{k}|\pi_{\mu})+D_{4}\eta^{\alpha}\nonumber \\
 & \leq-\mathrm{\frac{3\gamma_{1}}{2}}H(p_{k,t}|\pi_{\mu})+\frac{40N^{3}}{\gamma_{1}}L^{6}\eta^{2\alpha}H(p_{k}|\pi_{\mu})+D_{4}\eta^{\alpha},
\end{align}
where in the last inequality we have used the definition \ref{D1}
of LSI inequality. Multiplying both sides by $e^{\frac{3\gamma_{1}}{2}t}$,
and integrating both sides from $t=0$ to $t=\eta$ we obtain 
\begin{align}
e^{\frac{3\gamma}{2}\eta}H(p_{k+1}|\pi_{\mu})-H(p_{k}|\pi_{\mu}) & \leq2\left(\frac{e^{\frac{3\gamma_{1}}{2}\eta}-1}{3\gamma_{1}}\right)\left(\frac{40N^{3}}{\gamma_{1}}L^{6}\eta^{2\alpha}H(p_{k}|\pi_{\mu})+D_{4}\eta^{\alpha}\right)\nonumber \\
 & \leq2\eta\left(\frac{40N^{3}}{\gamma_{1}}L^{6}\eta^{2\alpha}H(p_{k}|\pi)+D_{4}\eta^{\alpha}\right)
\end{align}
where the last line holds by $e^{c}\leq1+2c$ for $0<c=\frac{3\gamma_{1}}{2}\eta<1$.
Rearranging the term of the above inequality and using the facts that
$1+\eta^{1+2\alpha}\frac{80N^{3}}{\gamma_{1}}L^{6}\leq1+\frac{\gamma_{1}\eta}{2}\leq e^{\frac{\gamma_{1}\eta}{2}}$
when $\eta\leq\left(\frac{\gamma_{1}}{13N^{\frac{3}{2}}L^{3}}\right)^{\frac{1}{\alpha}}$
and $e^{-\frac{3\gamma_{1}}{2}\eta}\leq1$ leads to 
\begin{align}
H(p_{k+1}|\pi_{\mu}) & \leq e^{-\frac{3\gamma_{1}}{2}\eta}\left(1+\eta^{1+2\alpha}\frac{80N^{3}}{\gamma_{1}}L^{6}\right)H(p_{k}|\pi_{\mu})+2\eta^{\alpha+1}D_{4}\nonumber \\
 & \leq e^{-\gamma_{1}\eta}H(p_{k}|\pi_{\mu})+2\eta^{\alpha+1}D_{3}.
\end{align}
Applying this inequality recursively, and using the inequality $1-e^{-c}\ge\frac{3}{4}c$
for $0<c=\gamma_{1}\eta\le\frac{1}{4}$ we obtain 
\begin{align}
H(p_{k}|\pi_{\mu}) & \le\,e^{-\gamma_{1}\eta k}H(p_{0}|\pi_{\mu})+\frac{2\eta^{\alpha+1}D_{4}}{1-e^{-\gamma_{1}\eta}}\nonumber \\
 & \le\,e^{-\gamma_{1}\eta k}H(p_{0}|\pi_{\mu})+\frac{2\eta^{\alpha+1}D_{4}}{\frac{3}{4}\gamma_{1}\eta}\nonumber \\
 & \le\,e^{-\gamma_{1}\eta k}H(p_{0}|\pi_{\mu})+\frac{8\eta^{\alpha}D_{4}}{3\gamma_{1}}.
\end{align}
Note that last inequality holds if we choose $\eta$ such that it
satisfies 
\[
\eta\le\min\left\{ 1,\frac{1}{4\gamma_{1}},\left(\frac{\gamma_{1}}{13N^{\frac{3}{2}}L^{3}}\right)^{\frac{1}{\alpha}}\right\} .
\]

From Lemma \ref{3.3.2}, by choosing $\mu=\sqrt{\eta}$ small enough
so that $W_{2}(\pi,\ \pi_{\mu})\leq3\sqrt{NLE_{2}}\eta^{\frac{\alpha}{2}}d^{\frac{1}{p}}$.
Since $\pi$ satisfies log-Sobolev inequality, by triangle inequality
we also get

\begin{align*}
W_{2}(p_{\mu k},\ \pi) & \leq W_{2}(p_{\mu k},\ \pi_{\mu})+W_{2}(\pi,\ \pi_{\mu})\\
 & \leq\sqrt{\frac{2}{\gamma}H(p_{\mu k},\pi_{\mu})}+W_{2}(\pi,\ \pi_{\mu})\\
 & \leq\frac{1}{\sqrt{\gamma_{1}}}e^{-\frac{\gamma_{1}}{2}\eta k}\sqrt{H(p_{0}|\pi_{\mu})}+\frac{2}{\gamma_{1}}\eta^{\frac{\alpha}{2}}\sqrt{D_{4}}+3\sqrt{NLE_{2}}\eta^{\frac{\alpha}{2}}d^{\frac{1}{p}}.
\end{align*}
Given $\epsilon>0$, if we further assume $\eta\le\left(\frac{\epsilon\gamma_{1}}{6\sqrt{D_{4}}}\right)^{\frac{2}{\alpha}}\wedge\left(\frac{\epsilon}{9\sqrt{NLE_{2}}d^{\frac{1}{p}}}\right)^{\frac{2}{\alpha}}$,
then the above inequality implies $H(p_{k}|\pi_{\mu})\le\frac{1}{\sqrt{\gamma_{1}}}e^{-\frac{\gamma_{1}}{2}\eta k}\sqrt{H(p_{0}|\pi_{\mu})}+\frac{2\epsilon}{3}.$
This means for $k\ge\frac{2}{\gamma_{1}\eta}\log\frac{3\sqrt{H\left(p_{0}|\pi_{\mu}\right)\gamma_{1}}}{\epsilon},$
we have $H(p_{k}|\pi)\le\frac{\epsilon}{3}+\frac{2\epsilon}{3}=\epsilon$,
as desired. \end{proof}

\subsection{Proof of Lemma~\ref{3.3.2}\label{Proof-of-Lemma3.3.2}}

\begin{lemma} Assume that $\pi\propto\exp(-\pi)$ and $\pi_{\mu}\propto\exp(-U_{\mu})$
and $\pi$ has a bounded second moment, that is $\int\left\Vert x\right\Vert ^{2}\pi(x)dx=E_{2}<\infty$.
We deduce the following bounds 
\[
W_{2}^{2}(\pi,\ \pi_{\mu})\leq8.24NL\mu^{1+\alpha}d^{\frac{2}{p}}E_{2}.
\]
for any $\mu\leq0.05$. \end{lemma}\begin{proof} This proof adapts
the technique of the proof of \citep{dalalyan2019bounding}'s Proposition
1. Without loss of generality we may assume that ${\displaystyle \int_{\mathbb{R}^{p}}\exp(-U(x))dx=1}$.
We first give upper and lower bounds to the normalizing constant of
$\pi_{\mu}$, that is 
\begin{align*}
c_{\mu} & \stackrel{_{\triangle}}{=}\int_{\mathbb{R}^{d}}\pi(x)e^{-\left(U_{\mu}(x)-U(x)\right)}dx.\\
 & =\mathbb{E}_{\pi}\left(e^{-\left(U_{\mu}(x)-U(x)\right)}\right)
\end{align*}
The constant $c_{\mu}$ is an expectation of $e^{-\left(U_{\mu}(x)-U(x)\right)}$
with respect to the density $\pi$ so it can be trivially upper bounded
by $e^{M}$ and lower bounded by $e^{-M}$ where $\left|U_{\mu}(x)-U(x)\right|\leq\sum_{i}L_{i}\mu^{1+\alpha_{i}}d^{\frac{2}{p}}=M$.
Now we control the distance between densities $\pi$ and $\pi_{\mu}$
at any fixed $x\in\mathbb{R}^{d}$: 
\begin{align*}
\left|\pi(x)-\pi_{\mu}(x)\right| & =\pi(x)\left|1-\frac{e^{-\left(U_{\mu}(x)-U(x)\right)}}{c_{\mu}}\right|\\
~ & \leq\pi(x)\left\{ \left(1-\frac{e^{-\left(U_{\mu}(x)-U(x)\right)}}{e^{M}}\right)+e^{-\left(U_{\mu}(x)-U(x)\right)}\left(\frac{1}{c_{\mu}}-\frac{1}{e^{M}}\right)\right\} \\
 & \leq\pi(x)\left(1-e^{-2M}+e^{2M}-1\right)\\
 & \leq\pi(x)\left(2M+e^{2M}-1\right).
\end{align*}
The first inequality is from triangle inequality of absolute value,
second inequality is trivial while the last inequality follows from
$1-e^{-x}\leq x$ for any $x\geq0$. To bound $W_{2}$, we use an
inequality from \cite{villani2008optimal}(Theorem 6.15, page 115):
\[
W_{2}^{2}(\pi,\ \pi_{\mu})\leq2\int_{\mathbb{R}^{d}}\Vert x\Vert_{2}^{2}\left|\pi(x)-\pi_{\mu}(x)\right|dx.
\]
Combining this with the bound on $\left|\pi(x)-\pi_{\mu}(x)\right|$
shown above, we have 
\begin{align*}
W_{2}^{2}(\pi,\ \pi_{\mu}) & \leq2\int_{\mathbb{R}^{d}}\Vert x\Vert_{2}^{2}\pi(x)\left(2M+e^{2M}-1\right)dx\\
 & \leq2\left(2M+e^{2M}-1\right)E_{\pi}\left[\Vert x\Vert^{2}\right]\\
 & \leq2\left(2M+e^{2M}-1\right)E_{2}\\
 & \leq8.24\sum_{i}L_{i}\mu^{1+\alpha_{i}}d^{\frac{2}{p}}E_{2}\\
 & \leq8.24NL\mu^{1+\alpha}d^{\frac{2}{p}}E_{2},
\end{align*}
where in the last inequality $M<0.05$ ensures that $e^{2M}-1\leq2.12M$.
This gives the desired result.

\end{proof}

\setcounter{lemma}{3}

\section{Convexification of non-convex domain}

\label{sec:latent}

\subsection{Proof of Lemma~\ref{Lem4.0.1}\label{Proof-of-Lemma4.0.1}}

\begin{lemma} For function $V$ defined as

\begin{equation}
V(\ x)=\inf_{\substack{\{\ x_{i}\}\subset\Omega,\\
\left\{ \lambda_{i}\big|\sum_{i}\lambda_{i}=1\right\} \\
\text{s.t.},\sum_{i}\lambda_{i}\ x_{i}=\ x
}
}\left\{ \sum_{i=1}^{l}\lambda_{i}U(\ x_{i})\right\} ,
\end{equation}
$\forall\ x\in\mathbb{B}(0,R)$, $\inf_{\left\Vert x\right\Vert =R}U(x)\leq V(\ x)\leq\sup_{\left\Vert x\right\Vert =R}U(x)$.
\label{lemma:V} \end{lemma}

\begin{proof}

First, by definition of $V$ inside $\mathbb{B}(0,R)$, we show that
for any linear combination of the form $\sum_{i}\lambda_{i}U(\ x_{i})$
where$\sum_{i}\lambda_{i}=1,$ we can find another representation
$\sum_{j}\lambda_{j}U(\ x_{j})$ where $\sum_{j}\lambda_{j}=1$ and
$\left\Vert x_{j}\right\Vert =R$ such that $\sum_{j}\lambda_{j}U(\ x_{j})\leq\sum_{i}\lambda_{i}U(\ x_{i})$.
This follows straightforwardly as follows.

For any $\ x_{j}\in\{\ x_{i}\}$, such that $\left\Vert \bar{x}_{j}\right\Vert >R$,
there exists a new convex combination $\{\ x_{i}\}\bigcup\{\bar{x}_{j}\}\setminus\{\ x_{j}\}$
with $\left\Vert \bar{x}_{j}\right\Vert =R$, such that $\sum_{i}\lambda_{i}U(\ x_{i})\geq\tilde{\lambda}_{j}U(\bar{x}_{j})+\sum_{i\neq j}\tilde{\lambda}_{i}U(\ x_{i})$.
In this case, we choose $\bar{x}_{j}$ where $\left\Vert \bar{x}_{j}\right\Vert =R$,
such that: 
\begin{align}
\bar{x}_{j} & =\dfrac{1-\bar{\lambda}_{j}}{1-\lambda_{j}}\ x+\dfrac{\bar{\lambda}_{j}-\lambda_{j}}{1-\lambda_{j}}\ x_{j},\:\lambda_{j}<\bar{\lambda}_{j}<1,\nonumber \\
 & =\bar{\lambda}_{j}\ x_{j}+\left(\dfrac{1-\bar{\lambda}_{j}}{1-\lambda_{j}}\right)\left(\sum_{i\neq j}\lambda_{i}\ x_{i}\right).
\end{align}
Since $U$ is convex on $\Omega$, 
\begin{equation}
U(\bar{x}_{j})\leq\bar{\lambda}_{j}U(\ x_{j})+\left(\dfrac{1-\bar{\lambda}_{j}}{1-\lambda_{j}}\right)\left(\sum_{i\neq j}\lambda_{i}U(\ x_{i})\right).
\end{equation}
On the other hand,$x$ can be represented as a convex combination
of $\{\ x_{i}\}\bigcup\{\bar{x}_{j}\}\setminus\{\ x_{j}\}$: 
\begin{equation}
\ x=\dfrac{\lambda_{j}}{\bar{\lambda}_{j}}\bar{x}_{j}+\left(1-\dfrac{\lambda_{j}}{\bar{\lambda}_{j}}\dfrac{1-\bar{\lambda}_{j}}{1-\lambda_{j}}\right)\left(\sum_{i\neq j}\lambda_{i}\ x_{i}\right)=\tilde{\lambda}_{j}\bar{x}_{j}+\sum_{i\neq j}\tilde{\lambda}_{i}\ x_{i},
\end{equation}
and that 
\begin{align}
\sum_{i}\lambda_{i}U(\ x_{i}) & \geq\dfrac{\lambda_{j}}{\bar{\lambda}_{j}}U(\bar{x}_{j})+\left(1-\dfrac{\lambda_{j}}{\bar{\lambda}_{j}}\dfrac{1-\bar{\lambda}_{j}}{1-\lambda_{j}}\right)\left(\sum_{i\neq j}\lambda_{i}U(\ x_{i})\right)\nonumber \\
 & =\tilde{\lambda}_{j}U(\bar{x}_{j})+\sum_{i\neq j}\tilde{\lambda}_{i}U(\ x_{i}).
\end{align}

As a result, $V(\ x)$ can be represented as

\begin{equation}
V(\ x)=\inf_{\substack{\{\ x_{j}\}\subset\Omega,\\
\left\{ \lambda_{j}\big|\sum_{j}\lambda_{j}=1\right\} \\
\text{s.t.},\sum_{j}\lambda_{j}\ x_{j}=\ x,\,\left\Vert x_{i}\right\Vert =R
}
}\left\{ \sum_{j}\lambda_{j}U(\ x_{j})\right\} .
\end{equation}
By the representation of $V$ inside $\mathbb{B}(0,R)$, we obtain
$\inf_{\left\Vert \bar{x}\right\Vert =R}U(\bar{x})\leq V(\ x)\leq\sup_{\left\Vert \bar{x}\right\Vert =R}U(\bar{x}).$
\end{proof}

\subsection{Proof of Lemma~\ref{Lem4.0.2}\label{Proof-of-Lemma4.0.2}}

\begin{lemma} For $U$ satisfying $\alpha$-mixture weakly smooth
and $\left(\mu,\theta\right)$-degenerated convex outside the ball
radius $R$, there exists $\hat{U}\in C^{1}(\mathbb{R}^{d})$ with
a Hessian that exists everywhere on $\mathbb{R}^{d}$, and $\hat{U}$
is $\left(\left(1-\theta\right)\frac{\mu}{2},\theta\right)$-degenerated
convex on $\mathbb{R}^{d}$ (that is $\nabla^{2}\hat{U}(x)\succeq\left(1-\theta\right)\frac{\mu}{2}\left(1+\left\Vert x\right\Vert ^{2}\right)^{-\frac{\theta}{2}}I_{d}$),
such that 
\begin{align}
\sup\left(\hat{U}(\ x)-U(\ x)\right) & -\inf\left(\hat{U}(\ x)-U(\ x)\right)\leq\sum_{i}L_{i}R^{1+\alpha_{i}}+\frac{4\mu}{\left(2-\theta\right)}\ R^{2-\theta}.
\end{align}
\label{L2} \end{lemma} \begin{proof} Following closely to \citep{ma2019sampling}'s
approach, let $g(\ x)=\frac{\mu}{2\left(2-\theta\right)}\ \left(1+\left\Vert x\right\Vert ^{2}\right)^{1-\frac{\theta}{2}}$
for $0\leq\theta<1$. The gradient of $g\left(x\right)$ is $\nabla g(\ x)=\frac{\mu}{2}\left(1+\left\Vert x\right\Vert ^{2}\right)^{-\frac{\theta}{2}}x$
and the Hessian of $g\left(x\right)$ is 
\begin{align}
\nabla^{2}g(\ x) & =\frac{\mu}{2}\left[\left(1+\left\Vert x\right\Vert ^{2}\right)^{-\frac{\theta}{2}}I_{d}-\theta\left(1+\left\Vert x\right\Vert ^{2}\right)^{-\frac{\theta}{2}-1}xx^{T}\right]\nonumber \\
 & \preceq\frac{\mu}{2}\left(1+\left\Vert x\right\Vert ^{2}\right)^{-\frac{\theta}{2}}I_{d}.
\end{align}
On the other hand, we also have 
\begin{align}
\nabla^{2}g(\ x) & =\frac{\mu}{2}\left[\left(1+\left\Vert x\right\Vert ^{2}\right)^{-\frac{\theta}{2}}I_{d}-\theta\left(1+\left\Vert x\right\Vert ^{2}\right)^{-\frac{\theta}{2}-1}xx^{T}\right]\nonumber \\
 & =\frac{\mu}{2}\left(1+\left\Vert x\right\Vert ^{2}\right)^{-\frac{\theta}{2}-1}\left[I_{d}+I_{d}\left\Vert x\right\Vert ^{2}-\theta\left\Vert x\right\Vert ^{2}\frac{xx^{T}}{\left\Vert x\right\Vert ^{2}}\right]\nonumber \\
 & =\frac{\mu}{2}\left(1+\left\Vert x\right\Vert ^{2}\right)^{-\frac{\theta}{2}-1}\left[I_{d}+I_{d}\left(1-\theta\right)\left\Vert x\right\Vert ^{2}+\theta\left\Vert x\right\Vert ^{2}\left(I_{d}-\frac{xx^{T}}{\left\Vert x\right\Vert ^{2}}\right)\right]\nonumber \\
 & \succeq\frac{\mu}{2}\left(1+\left\Vert x\right\Vert ^{2}\right)^{-\frac{\theta}{2}-1}\left(\left(1-\theta\right)\left\Vert x\right\Vert ^{2}+1\right)I_{d}\nonumber \\
 & \succeq\frac{\mu}{2}\left(1+\left\Vert x\right\Vert ^{2}\right)^{-\frac{\theta}{2}-1}\left(\left(1-\theta\right)\left(\left\Vert x\right\Vert ^{2}+1\right)\right)I_{d}\nonumber \\
 & \succeq\left(1-\theta\right)\frac{\mu}{2}\left(1+\left\Vert x\right\Vert ^{2}\right)^{-\frac{\theta}{2}}I_{d}.
\end{align}
We adapt \citep{yan2012extension} by denoting $\tilde{U}\left(x\right)=U\left(x\right)-g\left(x\right).$
Since $U\left(x\right)$ is $\left(\mu,\theta\right)$-degenerated
convex outside the ball, we deduce for every $\left\Vert x\right\Vert \geq R,$
\begin{align}
\nabla^{2}\tilde{U}\left(x\right) & =\nabla^{2}U\left(x\right)-\nabla^{2}g\left(x\right)\nonumber \\
 & \succeq\mu\left(1+\left\Vert x\right\Vert {}^{2}\right)^{-\frac{\theta}{2}}I_{d}-\frac{\mu}{2}\left(1+\left\Vert x\right\Vert ^{2}\right)^{-\frac{\theta}{2}}I_{d}\nonumber \\
 & \succeq\frac{\mu}{2}\left(1+\left\Vert x\right\Vert ^{2}\right)^{-\frac{\theta}{2}}I_{d},
\end{align}
which implies that $\tilde{U}\left(x\right)$ is $\left(\frac{\mu}{2},\theta\right)$-degenerated
convex outside the ball. Now, we construct $\hat{U}(\ x)$ so that
it is twice differentiable, degenerated convex on all $\mathbb{R}^{d}$
and differs from $U(\ x)$ less than $4LR^{1+\alpha}+4LR^{1+\ell+\alpha}+\frac{4\mu}{\left(2-\theta\right)}\ R^{2-\theta}$.
Based on the same construction of \citep{ma2019sampling}, we first
define the function $V$ as the convex extension \citep{yan2012extension}
of $\tilde{U}$ from domain $\Omega=R^{d}\setminus\mathbb{B}\left(0,R\right)$
to its convex hull $\Omega^{co}$, $V\left(x\right)=\inf\left\{ \sum_{i}\lambda_{i}\tilde{U}(\ x_{i})\right\} $
for every $x\in\mathbb{R}^{d}.$ Since $\tilde{U}(\ x)$ is convex
in $\Omega$, $V(\ x)=\tilde{U}(\ x)$ for $\ x\in\Omega$. By Lemma~\ref{Lem4.0.1},
$V(\ x)$ is convex on the entire domain $R^{d}$ and $V(\ x)$ can
be represented as 
\begin{equation}
V(\ x)=\inf_{\substack{\{\ x_{j}\}\subset\Omega,\\
\left\{ \lambda_{j}\big|\sum_{j}\lambda_{j}=1\right\} \\
\text{s.t.},\sum_{j}\lambda_{j}\ x_{j}=\ x,\mathrm{and}\left\Vert x_{i}\right\Vert =R
}
}\left\{ \sum_{j}\lambda_{j}\tilde{U}(\ x_{j})\right\} .
\end{equation}
Therefore, $\forall\ x\in\mathbb{B}(0,R)$, $\inf_{\left\Vert \bar{x}\right\Vert =R}\tilde{U}(\bar{x})\leq V(\ x)\leq\sup_{\left\Vert \bar{x}\right\Vert =R}\tilde{U}(\bar{x})$.
Next we construct $\tilde{V}(\ x)$ to be a smoothing of $V$ on $\mathbb{B}\left(0,R+\epsilon\right)$.
Consider the function $\varphi{\displaystyle (x)}$ of a variable
$x$ in $\mathbb{R}^{d}$ defined by 
\begin{equation}
{\displaystyle \varphi(x)=\begin{cases}
Ce^{-1/(1-\left\Vert x\right\Vert ^{2})} & \text{ if }\left\Vert x\right\Vert <1\\
0 & \text{ if }\left\Vert x\right\Vert \geq1
\end{cases}}
\end{equation}
where the numerical constant $C$ ensures normalization. Let ${\displaystyle \varphi_{\delta}(x)=\delta^{-d}\varphi(\delta^{-1}x)}$
be a smooth function supported on the ball $\mathbb{B}(0,\delta)$.
Define 
\begin{align}
\tilde{V}(\ x) & =\int V(\ y)\varphi_{\delta}(\ x-y)dy\nonumber \\
 & =\int V(\ x-y)\varphi_{\delta}(y)dy\nonumber \\
 & =E_{y}\left[V(x-y)\right].
\end{align}
The third equality implies that for any $x$ and $z\in\mathbb{R}^{d}$,
\begin{align}
\left\langle \nabla\tilde{V}(\ x)-\nabla\tilde{V}(\ z),x-z\right\rangle  & =\left\langle \nabla E_{y}\left[V(x-y)\right]-\nabla E_{y}\left[V(z-y)\right],x-z\right\rangle \nonumber \\
 & \stackrel{_{1}}{=}\left\langle E_{y}\left[\nabla V(x-y)\right]-E_{y}\left[\nabla V(z-y)\right],x-z\right\rangle \nonumber \\
 & =\left\langle E_{y}\left[\nabla V(x-y)-\nabla V(z-y)\right],x-z\right\rangle \nonumber \\
 & =E_{y}\left\langle \nabla V(x-y)-\nabla V(z-y),x-z\right\rangle \nonumber \\
 & \geq0,
\end{align}
where step $1$ follows from exchangeability of gradient and integral
and the last line is because of convexity of $V$, which indicates
$\tilde{V}$ is a smooth and convex function on $R^{d}.$ Also, note
that the definition of $\tilde{V}$ implies that $\forall\left\Vert x\right\Vert <R+\epsilon$,
\begin{equation}
\inf_{\left\Vert \bar{x}\right\Vert <R+\epsilon+\delta}V(\bar{x})\leq\tilde{V}(\ x)\leq\sup_{\left\Vert \bar{x}\right\Vert <R+\epsilon+\delta}V(\bar{x}).
\end{equation}
And by Lemma~\ref{Lem4.0.1}, for $\quad\forall\left\Vert \bar{x}\right\Vert <R+\epsilon$
\begin{align}
\inf_{\bar{x}\in\mathbb{B}\left(0,R+\epsilon+\delta\right)\setminus\mathbb{B}(0,R)}\tilde{U}(\bar{x})\leq\tilde{V}(\ x)\leq\sup_{\bar{x}\in\mathbb{B}\left(0,R+\epsilon+\delta\right)\setminus\mathbb{B}(0,R)}\tilde{U}(\bar{x}).\label{eq:BoundT_V}
\end{align}
Finally, we construct the auxiliary function: 
\begin{align}
\hat{U}(\ x)-g\left(x\right)=\left\{ \begin{array}{l}
\tilde{U}(\ x),\ \left\Vert x\right\Vert \geq R+2\epsilon\\
\alpha(\ x)\tilde{U}(\ x)+(1-\alpha(\ x))\tilde{V}(\ x),\ R+\epsilon<\left\Vert x\right\Vert <R+2\epsilon\\
\tilde{V}(\ x),\ \left\Vert x\right\Vert \leq R+\epsilon
\end{array}\right.
\end{align}
where $\alpha(\ x)=\dfrac{1}{2}\cos\left(\pi\dfrac{\left\Vert x\right\Vert ^{2}}{\epsilon\left(2R+3\epsilon\right)^{2}}-\frac{\left(R+\epsilon\right)^{2}}{\epsilon\left(2R+3\epsilon\right)^{2}}\pi\right)+\dfrac{1}{2}$.
Here we know that $\tilde{U}(\ x)$ is convex and smooth in $\mathbb{R}^{d}\setminus\mathbb{B}\left(0,R\right)$;
$\tilde{V}(\ x)$ is also convex and smooth in $\mathbb{R}^{d}\setminus\mathbb{B}\left(0,R+\epsilon\right)$.
Hence for $R+\epsilon<\left\Vert x\right\Vert <R+2\epsilon$, 
\begin{align}
\nabla^{2}\left(\hat{U}(\ x)-g\left(x\right)\right) & =\nabla^{2}\tilde{U}(\ x)+\nabla^{2}\left((1-\alpha(\ x))(\tilde{V}(\ x)-\tilde{U}(\ x))\right)\nonumber \\
 & =\alpha(\ x)\nabla^{2}\tilde{U}(\ x)+(1-\alpha(\ x))\nabla^{2}\tilde{V}(\ x)\nonumber \\
 & -\nabla^{2}\alpha(\ x)\left(\tilde{V}(\ x)-\tilde{U}(\ x)\right)-2\nabla\alpha(\ x)\left(\nabla\tilde{V}(\ x)-\nabla\tilde{U}(\ x)\right)^{T}\nonumber \\
 & \succeq-\nabla^{2}\alpha(\ x)\left(\tilde{V}(\ x)-\tilde{U}(\ x)\right)-2\nabla\alpha(\ x)\left(\nabla\tilde{V}(\ x)-\nabla\tilde{U}(\ x)\right)^{T}.
\end{align}
Note that for $R+\epsilon<\left\Vert x\right\Vert <R+2\epsilon$,
we have 
\begin{align}
 & \left\Vert \nabla g(\ x)-\nabla g(\ x-y)\right\Vert \nonumber \\
 & =\left\Vert \frac{\mu}{2}\left(1+\left\Vert x\right\Vert ^{2}\right)^{-\frac{\theta}{2}}x-\frac{\mu}{2}\left(1+\left\Vert x-y\right\Vert ^{2}\right)^{-\frac{\theta}{2}}\left(x-y\right)\right\Vert \\
 & \leq\left\Vert \frac{\mu}{2}\left(1+\left\Vert x\right\Vert ^{2}\right)^{-\frac{\theta}{2}}x-\frac{\mu}{2}\left(1+\left\Vert x\right\Vert ^{2}\right)^{-\frac{\theta}{2}}\left(x-y\right)\right\Vert \nonumber \\
 & +\left\Vert \frac{\mu}{2}\left(1+\left\Vert x\right\Vert ^{2}\right)^{-\frac{\theta}{2}}\left(x-y\right)-\frac{\mu}{2}\left(1+\left\Vert x-y\right\Vert ^{2}\right)^{-\frac{\theta}{2}}\left(x-y\right)\right\Vert \nonumber \\
 & \leq\frac{\mu}{2}\left(1+\left\Vert x\right\Vert ^{2}\right)^{-\frac{\theta}{2}}\left\Vert y\right\Vert +\frac{\mu}{2}\left|\left(1+\left\Vert x\right\Vert ^{2}\right)^{-\frac{\theta}{2}}-\left(1+\left\Vert x-y\right\Vert ^{2}\right)^{-\frac{\theta}{2}}\right|\left\Vert x-y\right\Vert \nonumber \\
 & \leq\frac{\mu}{2}\left(1+\left(R+\epsilon\right){}^{2}\right)^{-\frac{\theta}{2}}\delta+\frac{\mu}{2}\frac{\left|\left(1+\left\Vert x\right\Vert ^{2}\right)^{\frac{\theta}{2}}-\left(1+\left\Vert x-y\right\Vert ^{2}\right)^{\frac{\theta}{2}}\right|}{\left(1+\left\Vert x\right\Vert ^{2}\right)^{\frac{\theta}{2}}\left(1+\left\Vert x-y\right\Vert ^{2}\right)^{\frac{\theta}{2}}}\left\Vert \left(x-y\right)\right\Vert \nonumber \\
 & \stackrel{_{1}}{\leq}\frac{\mu}{2}\left(1+\left(R+\epsilon\right){}^{2}\right)^{-\frac{\theta}{2}}\delta+\frac{\mu}{2}\frac{\left|\left(1+\left\Vert x\right\Vert ^{2}\right)-\left(1+\left\Vert x-y\right\Vert ^{2}\right)\right|}{\left(1+\left\Vert x\right\Vert ^{2}\right)^{\frac{\theta}{2}}\left(1+\left\Vert x-y\right\Vert ^{2}\right)^{\frac{\theta}{2}}}\left\Vert \left(x-y\right)\right\Vert \nonumber \\
 & \leq\frac{\mu}{2}\left(1+\left(R+\epsilon\right){}^{2}\right)^{-\frac{\theta}{2}}\delta+\frac{\mu}{2}\frac{\left|\left(\left\Vert x\right\Vert -\left\Vert x-y\right\Vert \right)\left(\left\Vert x\right\Vert +\left\Vert x-y\right\Vert \right)\right|}{\left(1+\left\Vert x\right\Vert ^{2}\right)^{\frac{\theta}{2}}\left(1+\left\Vert x-y\right\Vert ^{2}\right)^{\frac{\theta}{2}}}\left\Vert \left(x-y\right)\right\Vert \nonumber \\
 & \stackrel{_{2}}{\leq}\frac{\mu}{2}\left(1+\left(R+\epsilon\right){}^{2}\right)^{-\frac{\theta}{2}}\delta+\frac{\mu}{2}\frac{\left\Vert y\right\Vert \left(\left\Vert x\right\Vert +\left\Vert x-y\right\Vert \right)}{\left(1+\left\Vert x\right\Vert ^{2}\right)^{\frac{\theta}{2}}\left(1+\left\Vert x-y\right\Vert ^{2}\right)^{\frac{\theta}{2}}}\left\Vert \left(x-y\right)\right\Vert \nonumber \\
 & \leq\frac{\mu}{2}\left(1+\left(R+\epsilon\right){}^{2}\right)^{-\frac{\theta}{2}}\delta+\frac{\mu}{2}\frac{2\left(R+2\epsilon+\delta\right)^{2}\delta}{\left(1+\left(R+\epsilon\right)^{2}\right)^{\frac{\theta}{2}}\left(1+\left(R+\epsilon-\delta\right)^{2}\right)^{\frac{\theta}{2}}},
\end{align}
where $1$ follows from Lemma \ref{L24}, while $2$ is due to triangle
inequality. As a result, we get 
\begin{align}
\left\Vert \nabla\tilde{V}(\ x)-\nabla\tilde{U}(\ x)\right\Vert  & =\int\left\Vert \nabla\tilde{U}(\ x-\ y)-\nabla\tilde{U}(\ x)\right\Vert \varphi_{\delta}(\ y)dy\nonumber \\
 & \leq\sum_{i}L_{i}\delta^{\alpha_{i}}+\left\Vert \nabla g(\ x)-\nabla g(\ x-y)\right\Vert \nonumber \\
 & \leq NL\delta^{\alpha}+\frac{\mu}{2}\left(1+\left(R+\epsilon\right){}^{2}\right)^{-\frac{\theta}{2}}\delta\\
 & +\frac{\mu}{2}\frac{2\left(R+\epsilon-\delta\right)^{2}\delta}{\left(1+\left(R+\epsilon\right)^{2}\right)^{\frac{\theta}{2}}\left(1+\left(R+\epsilon-\delta\right)^{2}\right)^{\frac{\theta}{2}}}.
\end{align}
On the other hand, we also acquire 
\begin{align}
 & |\tilde{U}(\mathrm{x})-\tilde{U}(x-\mathrm{y})|\nonumber \\
 & \leq\mathrm{\max}\left\{ \left\langle \nabla U(\mathrm{x-y}),\mathrm{y}\right\rangle ,\left\langle \nabla U(\mathrm{x}),\mathrm{-y}\right\rangle \right\} \nonumber \\
 & +\sum_{i}\frac{L}{1+\alpha_{i}}\Vert\mathrm{y}\Vert^{\alpha_{i}+1}+\left|g\left(x\right)-g\left(x-y\right)\right|\nonumber \\
 & \mathrm{\leq\max}\left\{ \left\langle \nabla U(\mathrm{x-y}),\mathrm{y}\right\rangle ,\left\langle \nabla U(\mathrm{x}),\mathrm{-y}\right\rangle \right\} +\sum_{i}\frac{L}{1+\alpha_{i}}\Vert\mathrm{y}\Vert^{\alpha_{i}+1}\nonumber \\
 & +\left|\frac{\mu}{2\left(2-\theta\right)}\ \left(1+\left\Vert x\right\Vert ^{2}\right)^{1-\frac{\theta}{2}}-\frac{\mu}{2\left(2-\theta\right)}\ \left(1+\left\Vert x-y\right\Vert ^{2}\right)^{1-\frac{\theta}{2}}\right|\nonumber \\
 & \stackrel{_{1}}{\leq}\mathrm{\max}\left\{ \sum_{i}L_{i}\left\Vert \mathrm{x-y}\right\Vert ^{\alpha_{i}}\left\Vert y\right\Vert ,\sum_{i}L_{i}\left\Vert \mathrm{x}\right\Vert ^{\alpha_{i}}\left\Vert y\right\Vert \right\} \nonumber \\
 & +\sum_{i}\frac{L}{1+\alpha_{i}}\Vert\mathrm{y}\Vert^{\alpha_{i}+1}+\frac{\mu}{2\left(2-\theta\right)}\left|\left(1+\left\Vert x\right\Vert ^{2}\right)-\left(1+\left\Vert x-y\right\Vert ^{2}\right)\right|\\
 & \leq L\left\Vert y\right\Vert \mathrm{\max}\left\{ \sum_{i}L_{i}\left\Vert \mathrm{x-y}\right\Vert ^{\alpha_{i}},\sum_{i}L_{i}\left\Vert \mathrm{x}\right\Vert ^{\alpha_{i}}\right\} \nonumber \\
 & +\sum_{i}\frac{L}{1+\alpha_{i}}\Vert\mathrm{y}\Vert^{\alpha_{i}+1}+\frac{\mu}{2\left(2-\theta\right)}\left|\left(\left\Vert x\right\Vert -\left\Vert x-y\right\Vert \right)\left(\left\Vert x\right\Vert +\left\Vert x-y\right\Vert \right)\right|\\
 & \leq L\left\Vert y\right\Vert \mathrm{\max}\left\{ \sum_{i}L_{i}\left\Vert \mathrm{x-y}\right\Vert ^{\alpha_{i}},\sum_{i}L_{i}\left\Vert \mathrm{x}\right\Vert ^{\alpha_{i}}\right\} +\\
 & +\sum_{i}\frac{L}{1+\alpha_{i}}\Vert\mathrm{y}\Vert^{\alpha_{i}+1}+\frac{\mu}{2\left(2-\theta\right)}\left(\left\Vert x\right\Vert +\left\Vert x-y\right\Vert \right)\left\Vert y\right\Vert ,
\end{align}
where $1$ follows again from Lemma \ref{L24} and the last inequality
is because of triangle inequality. Hence for $R+\epsilon<\left\Vert x\right\Vert <R+2\epsilon$,
$\left\Vert y\right\Vert \leq\delta$, 
\begin{align*}
\tilde{V}(\ x)-\tilde{U}(\ x) & =\int\left(\tilde{U}(\ x-\ y)-\tilde{U}(\ x)\right)\varphi_{\delta}(\ y)d\ y\\
 & \leq L\left\Vert y\right\Vert \mathrm{\max}\left\{ \sum_{i}L_{i}\left\Vert \mathrm{x-y}\right\Vert ^{\alpha_{i}},\sum_{i}L_{i}\left\Vert \mathrm{x}\right\Vert ^{\alpha_{i}}\right\} +\\
 & +\sum_{i}\frac{L}{1+\alpha_{i}}\Vert\mathrm{y}\Vert^{\alpha_{i}+1}+\frac{\mu}{2\left(2-\theta\right)}\left(\left\Vert x\right\Vert +\left\Vert x-y\right\Vert \right)\left\Vert y\right\Vert \\
 & \leq L\delta\left[\sum_{i}L_{i}\left(R+2\epsilon+\delta\right)^{\alpha_{i}}\right]+\\
 & +\sum_{i}\frac{L}{1+\alpha_{i}}\delta^{\alpha_{i}+1}+\frac{\mu}{\left(2-\theta\right)}\left(R+2\epsilon+\delta\right)\delta
\end{align*}
Therefore, when $R+\epsilon<\left\Vert x\right\Vert <R+2\epsilon$,
\begin{align}
\nabla^{2}\left(\hat{U}(\ x)-g\left(x\right)\right) & \succeq-\frac{\left(R+\epsilon\right)^{2}\pi\left(L\delta\left[\sum_{i}L_{i}\left(R+2\epsilon+\delta\right)^{\alpha_{i}}\right]\right)}{\epsilon\left(2R+3\epsilon\right)}I_{d}\nonumber \\
 & -\frac{\left(R+\epsilon\right)^{2}\pi\left(+\sum_{i}\frac{L}{1+\alpha_{i}}\delta^{\alpha_{i}+1}+\frac{\mu}{\left(2-\theta\right)}\left(R+2\epsilon+\delta\right)\delta\right)}{\epsilon\left(2R+3\epsilon\right)}I_{d}\nonumber \\
 & -\frac{\left(R+\epsilon\right)^{4}\pi^{2}\left(NL\delta^{\alpha}+\frac{\mu}{2}\left(1+\left(R+\epsilon\right){}^{2}\right)^{-\frac{\theta}{2}}\delta\right)}{\epsilon^{2}\left(2R+3\epsilon\right)}I_{d}\nonumber \\
 & -\frac{\left(R+\epsilon\right)^{4}\pi^{2}\left(\frac{\mu}{2}\frac{2\left(R+\epsilon-\delta\right)^{2}\delta}{\left(1+\left(R+\epsilon\right)^{2}\right)^{\frac{\theta}{2}}\left(1+\left(R+\epsilon-\delta\right)^{2}\right)^{\frac{\theta}{2}}}\right)}{\epsilon^{2}\left(2R+3\epsilon\right)}I_{d}.
\end{align}
Taking the limit when $\delta\rightarrow0^{+}$, we obtain that for
$R+\epsilon<\left\Vert x\right\Vert <R+2\epsilon$, $\nabla^{2}\left(\hat{U}(\ x)-g\left(x\right)\right)$
is positive semi-definite; hence it is positive semi-definite on the
entire $R^{d}$, or $\hat{U}(\ x)-g\left(x\right)$ is convex on $\mathbb{R}^{d}$.
From \eqref{eq:BoundT_V}, we know that for $R+\epsilon<\left\Vert x\right\Vert <R+2\epsilon$,
\begin{align}
\inf_{\bar{x}\in\mathbb{B}\left(0,R+2\epsilon\right)\setminus\mathbb{B}(0,R)}\tilde{U}(\bar{x}) & \leq\hat{U}(\ x)-g\left(x\right)\leq\sup_{\bar{x}\in\mathbb{B}\left(0,R+2\epsilon\right)\setminus\mathbb{B}(0,R)}\tilde{U}(\bar{x}).
\end{align}
Therefore, 
\begin{align}
 & \sup\left(\hat{U}(\ x)-U(\ x)\right)-\inf\left(\hat{U}(\ x)-U(\ x)\right)\nonumber \\
 & =\sup\left(\hat{U}(\ x)-g\left(x\right)-\tilde{U}(\ x)\right)-\inf\left(\hat{U}(\ x)-g\left(x\right)-\tilde{U}(\ x)\right)\\
 & \leq2\left(\sup_{\bar{x}\in\mathbb{B}\left(0,R+2\epsilon\right)\setminus\mathbb{B}(0,R)}\tilde{U}(\bar{x})-\inf_{\bar{x}\in\mathbb{B}\left(0,R+2\epsilon\right)\setminus\mathbb{B}(0,R)}\tilde{U}(\bar{x})\right)\nonumber \\
 & \leq2\left(\sup_{\bar{x}\in\mathbb{B}\left(0,R+2\epsilon\right)}\tilde{U}(\bar{x})-\inf_{\bar{x}\in\mathbb{B}\left(0,R+2\epsilon\right)}\tilde{U}(\bar{x})\right).
\end{align}
Since $U$ is $\left(\alpha,\ell\right)$-weakly smooth and $\nabla U(0)=0$,
we deduce 
\begin{align}
\left|U(\ x)-U(0)\right| & =\left|U(\ x)-U(0)-\ \left\langle x,\nabla U(0)\right\rangle \right|\nonumber \\
 & \leq\sum_{i}\frac{L_{i}}{1+\alpha_{i}}\left\Vert x\right\Vert ^{1+\alpha_{i}}\nonumber \\
 & \leq\sum_{i}\frac{L_{i}}{1+\alpha_{i}}\left(R+2\epsilon\right)^{1+\alpha_{i}}\nonumber \\
 & \leq\sum_{i}L_{i}R^{1+\alpha_{i}}
\end{align}
and 
\begin{align}
\left|g(\ x)\right| & =\left|\frac{\mu}{2\left(2-\theta\right)}\ \left(1+\left\Vert x\right\Vert ^{2}\right)^{1-\frac{\theta}{2}}\right|\nonumber \\
 & \leq\frac{\mu}{2\left(2-\theta\right)}\ \left(1+\left(R+2\epsilon\right)^{2}\right)^{1-\frac{\theta}{2}}\nonumber \\
 & \leq\frac{\mu}{\left(2-\theta\right)}\ R^{2-\theta}.
\end{align}
So for $\forall\left\Vert x\right\Vert \leq\left(R+2\epsilon\right)$,
$\epsilon$ is sufficiently small, 
\begin{align*}
\sup_{\bar{x}\in\mathbb{B}\left(R+2\epsilon\right)}\tilde{U}(\bar{x})-\inf_{\bar{x}\in\mathbb{B}\left(R+2\epsilon\right)} & \tilde{U}(\bar{x})\leq\sum_{i}L_{i}R^{1+\alpha_{i}}+\frac{2\mu}{\left(2-\theta\right)}\ R^{2-\theta}.
\end{align*}
As a result, we get 
\begin{align*}
\sup\left(\hat{U}(\ x)-U(\ x)\right)-\inf & \left(\hat{U}(\ x)-U(\ x)\right)\leq2\sum_{i}L_{i}R^{1+\alpha_{i}}+\frac{4\mu}{\left(2-\theta\right)}\ R^{2-\theta}.
\end{align*}
\end{proof} \begin{remark} When $\theta=0,$ the $\left(\mu,\theta\right)$-degenerated
convex outside the ball is equivalent to the $\mu$-strongly convex
outside the ball, we achieve a result for strongly convex outside
the ball similar to \citep{ma2019sampling} but for $\left(\alpha,\ell\right)$-weakly
smooth instead of smooth. The constant could be improved by a factor
of $2$ if we take $\epsilon$ to be arbitrarily small. \end{remark}

\subsection{Proof of lemma \ref{Lem4.1.1}\label{Proof-of-lemma4.1.1}}

\begin{lemma} For $U$ satisfying $\gamma-$Poincar\'{e}, $\alpha$-mixture
weakly smooth with $\alpha_{N}=1$ and $2-$dissipative, there exists
$\breve{U}\in C^{1}(\mathbb{R}^{d})$ with a Hessian that exists everywhere
on $R^{d}$, and $\breve{U}$ is log-Sobolev on $\mathbb{R}^{d}$
such that 
\begin{equation}
\sup\left(\breve{U}(\ x)-U(\ x)\right)-\inf\left(\breve{U}(\ x)-U(\ x)\right)\leq2\sum_{i}L_{i}R^{1+\alpha_{i}}+4L_{N}R^{2}+4LR^{1+\alpha}.
\end{equation}
\label{lemma:hat_U-2-2} \end{lemma} \begin{proof} First, given
$R>0,$ let $\overline{U}(\mathrm{x}):=U(\mathrm{x})+\frac{L_{N}+\lambda_{0}}{2}\left\Vert x\right\Vert ^{2}$
for $\lambda_{0}=\frac{2L}{R^{1-\alpha}}$, we obtain the following
property 
\begin{align}
 & \left\langle \nabla\overline{U}(\mathrm{x})-\nabla\overline{U}(\mathrm{y}),x-y\right\rangle \nonumber \\
 & =\left\langle \nabla\left(U(\mathrm{x})+\frac{L_{N}+\lambda_{0}}{2}\left\Vert x\right\Vert ^{2}\right)-\nabla\left(U(\mathrm{y})+\frac{L_{N}+\lambda_{0}}{2}\left\Vert y\right\Vert ^{2}\right),x-y\right\rangle \\
 & =\left\langle \nabla U(\mathrm{x})-\nabla U(\mathrm{y})+(L_{N}+\lambda_{0})\left(x-y\right),x-y\right\rangle \nonumber \\
 & \stackrel{i}{\geq}-\sum_{i<N}L_{i}\left\Vert x-y\right\Vert ^{1+\alpha}+\lambda_{0}\left\Vert x-y\right\Vert ^{2}\nonumber \\
 & \geq\frac{\lambda_{0}}{2}\left\Vert x-y\right\Vert ^{2}\,for\,\left\Vert x-y\right\Vert \geq\left(\frac{NL}{\lambda_{0}}\right)^{\frac{1}{1-\alpha_{1}}}=R,
\end{align}
where $(i)$ follows from Assumption \ref{A1}. This implies that
$\overline{U}(\mathrm{x})$ is $\lambda_{0}-$ strongly convex outside
the ball $B_{R}=\left\{ x:\left\Vert x\right\Vert \leq R\right\} $.
Though $\overline{U}(\mathrm{x})$ behaves differently than Lemma
\ref{Lem4.0.2} assumptions, with some additional verifications, we
still can apply Lemma \ref{Lem4.0.2} to derive the result. We sketch
the proof as follows. There exists $\hat{U}\in C^{1}(\mathbb{R}^{d})$
with a Hessian that exists everywhere on $R^{d}$, 
\begin{align}
\hat{U}(\ x)-\frac{\lambda_{0}}{4}\left\Vert x\right\Vert ^{2}=\left\{ \begin{array}{l}
\tilde{\overline{U}}(\ x),\ \left\Vert x\right\Vert \geq R+2\epsilon\\
\alpha(\ x)\tilde{\overline{U}}(\ x)+(1-\alpha(\ x))\tilde{V}(\ x),\ R+\epsilon<\left\Vert x\right\Vert <R+2\epsilon\\
\tilde{V}(\ x),\ \left\Vert x\right\Vert \leq R+\epsilon
\end{array}\right.
\end{align}
where $\alpha(\ x)$ is defined as before. Both $\tilde{\overline{U}}(\ x)$
and $\tilde{V}(\ x)$ are convex and smooth in $\mathbb{R}^{d}\setminus\mathbb{B}\left(0,R\right)$
and for $R+\epsilon<\left\Vert x\right\Vert <R+2\epsilon$, $\left\Vert y\right\Vert \leq\delta$,
\begin{align}
\nabla^{2}\left(\hat{U}(\ x)-\frac{\lambda_{0}}{4}\left\Vert x\right\Vert ^{2}\right) & \succeq-\nabla^{2}\alpha(\ x)\left(\tilde{V}(\ x)-\tilde{\overline{U}}(\ x)\right)-2\nabla\alpha(\ x)\left(\nabla\tilde{V}(\ x)-\nabla\tilde{\overline{U}}(\ x)\right)^{T}.
\end{align}
In this case, we have 
\begin{align}
\left\Vert \nabla\tilde{V}(\ x)-\nabla\tilde{\overline{U}}(\ x)\right\Vert  & =\left\Vert \nabla\int\left(\overline{U}(\ x-\ y)-\overline{U}(\ x)\right)\varphi_{\delta}(\ y)dy\right\Vert \nonumber \\
 & \stackrel{_{1}}{\leq}\left\Vert \nabla\int\left(U(\ x-\ y)-U(\ x)\right)\varphi_{\delta}(\ y)dy\right\Vert \nonumber \\
 & +\lambda_{0}\int\left\Vert y\right\Vert \varphi_{\delta}(\ y)dy\nonumber \\
 & \leq\left\Vert \int\left(\nabla U(\ x-\ y)-\nabla U(\ x)\right)\varphi_{\delta}(\ y)dy\right\Vert +\lambda_{0}\delta\nonumber \\
 & \leq\sum_{i}L_{i}\delta^{\alpha_{1}}+\lambda_{0}\delta,
\end{align}
where $1$ holds by triangle inequality and the last line is because
of $\left(\alpha,\ell\right)-$weakly smooth assumption, while 
\begin{align}
 & \left|\tilde{\overline{U}}(\mathrm{x})-\tilde{\overline{U}}(x-\mathrm{y})\right|\nonumber \\
 & \stackrel{_{1}}{\leq}\left|\overline{U}(\mathrm{x})-\overline{U}(x-\mathrm{y})\right|+\left|\frac{L+\lambda_{0}}{2}\left\Vert x\right\Vert ^{2}-\frac{L+\lambda_{0}}{2}\left\Vert x-y\right\Vert ^{2}\right|\nonumber \\
 & \stackrel{_{2}}{\leq}\left\{ \left\langle \nabla U(\mathrm{x-y}),\mathrm{y}\right\rangle \vee\left\langle \nabla U(\mathrm{x}),\mathrm{-y}\right\rangle \right\} +\sum_{i}\frac{L_{i}}{1+\alpha_{i}}\left\Vert y\right\Vert ^{\alpha_{i}+1}\nonumber \\
 & +\frac{L_{N}+\lambda_{0}}{2}\left|\left\Vert x\right\Vert ^{2}-\left\Vert x-y\right\Vert ^{2}\right|\\
 & \mathrm{\leq}\left\{ \left\langle \nabla U(\mathrm{x-y}),\mathrm{y}\right\rangle \vee\left\langle \nabla U(\mathrm{x}),\mathrm{-y}\right\rangle \right\} +\sum_{i}\frac{L_{i}}{1+\alpha_{i}}\left\Vert y\right\Vert ^{\alpha_{i}+1}\nonumber \\
 & +\frac{L_{N}+\lambda_{0}}{2}\left(\left\Vert x\right\Vert -\left\Vert x-y\right\Vert \right)\left(\left\Vert x\right\Vert +\left\Vert x-y\right\Vert \right)\\
 & \leq\left\{ \left(\sum_{i}L_{i}\left\Vert \mathrm{x-y}\right\Vert ^{\alpha_{i}}\right)\left\Vert y\right\Vert \vee\left(\sum_{i}L_{i}\left\Vert \mathrm{x}\right\Vert ^{\alpha_{i}}\right)\left\Vert y\right\Vert \right\} \nonumber \\
 & +\sum_{i}\frac{L_{i}}{1+\alpha_{i}}\left\Vert y\right\Vert ^{\alpha_{i}+1}+\frac{L_{N}+\lambda_{0}}{2}\left\Vert y\right\Vert \mathrm{\max}\left\{ \left\Vert \mathrm{x-y}\right\Vert ,\left\Vert \mathrm{x}\right\Vert \right\} \\
 & \leq\sum_{i}L_{i}\left(R+2\epsilon+\delta\right)^{\alpha_{i}}\delta\\
 & +\sum_{i}\frac{L_{i}}{1+\alpha_{i}}\delta^{\alpha_{i}+1}+\frac{L_{N}+\lambda_{0}}{2}\left(R+2\epsilon+\delta\right)\delta,
\end{align}
where $1$ is due to triangle inequality, $2$ follows from Assumption
1, and the last line holds by plugging in all the limits. Taking the
limit when $\delta\rightarrow0^{+},$ and for sufficiently small $\epsilon$,
we obtain $\hat{U}(\ x)-\frac{\lambda_{0}}{4}\left\Vert x\right\Vert ^{2}$
is convex on all $\mathbb{R}^{d}$ or $\hat{U}(\ x)$ is $\frac{\lambda_{0}}{2}$-
strongly convex. By definition of $\overline{U}$, for $R+\epsilon<\left\Vert x\right\Vert <R+2\epsilon$
we obtain 
\begin{align}
\left|\overline{U}(\ x)-\overline{U}(0)\right| & \leq\left|U(\ x)-U(0)-\ \left\langle x,\nabla U(0)\right\rangle \right|+\frac{L_{N}+\lambda_{0}}{2}\left\Vert x\right\Vert ^{2}\nonumber \\
 & \leq+\sum_{i}\frac{L_{i}}{1+\alpha_{i}}\left\Vert x\right\Vert ^{\alpha_{i}+1}+\frac{L_{N}+\lambda_{0}}{2}\left\Vert x\right\Vert ^{2}\nonumber \\
 & \leq+\sum_{i}\frac{L_{i}}{1+\alpha_{i}}\left(R+2\epsilon+\delta\right)^{\alpha_{i}+1}+\frac{L_{N}+\lambda_{0}}{2}\left(R+2\epsilon+\delta\right)^{2}\nonumber \\
 & \leq\sum_{i}L_{i}R^{1+\alpha_{i}}+\left(L_{N}+\lambda_{0}\right)R^{2}.
\end{align}
As a result, from Lemma \ref{Lem4.0.2} we deduce 
\begin{align}
\sup\left(\hat{U}(\ x)-\overline{U}(\ x)\right) & -\inf\left(\hat{U}(\ x)-\overline{U}(\ x)\right)\leq2\sum_{i}L_{i}R^{1+\alpha_{i}}+2\left(L_{N}+\lambda_{0}\right)R^{2}.
\end{align}
Let $\breve{U}\left(x\right)=\hat{U}\left(x\right)-\left(\frac{L_{N}}{2}+\frac{\lambda_{0}}{4}\right)\left\Vert x\right\Vert ^{2}$
then for $\left\Vert x\right\Vert >R+2\epsilon+\delta$, $\hat{U}\left(x\right)=\overline{U}\left(x\right)$
so $\breve{U}\left(x\right)=U\left(x\right)$. For $\left\Vert x\right\Vert \leq R+2\epsilon+\delta$,
we have 
\begin{align}
 & \sup\left(\breve{U}(x)-U(x)\right)-\inf\left(\breve{U}(\ x)-U(x)\right)\nonumber \\
 & \leq\sup\left(\hat{U}(x)+\frac{L_{N}+\lambda_{0}}{2}\left\Vert x\right\Vert ^{2}-\overline{U}(x)\right)-\inf\left(\hat{U}(x)+\frac{L_{N}+\lambda_{0}}{2}\left\Vert x\right\Vert ^{2}-\overline{U}(x)\right)\nonumber \\
 & \leq\sup\left(\hat{U}(x)-\overline{U}(x)\right)-\inf\left(\hat{U}(x)-\overline{U}(x)\right)+\left(L_{N}+\lambda_{0}\right)\left(R+2\epsilon+\delta\right)^{2}\nonumber \\
 & \leq2\sum_{i}L_{i}R^{1+\alpha_{i}}+2\left(L_{N}+\lambda_{0}\right)R^{2}+2\left(L_{N}+\lambda_{0}\right)R^{2}.\nonumber \\
 & \leq2\sum_{i}L_{i}R^{1+\alpha_{i}}+4L_{N}R^{2}+4LR^{1+\alpha}.
\end{align}
So for every $x\in\mathbb{R}^{d},$ 
\[
\sup\left(\breve{U}(x)-U(x)\right)-\inf\left(\breve{U}(\ x)-U(x)\right)\leq2\sum_{i}L_{i}R^{1+\alpha_{i}}+4L_{N}R^{2}+4LR^{1+\alpha}.
\]
Since $U(x)$ is $PI(\gamma)$, and using \citep{ledoux2001logarithmic}'s
Lemma 1.2 we have, $\breve{U}(\ x)$ is Poincar\'{e} with constant 
\[
\gamma_{1}=\gamma e^{-4\left(2\sum_{i}L_{i}R^{1+\alpha_{i}}+4L_{N}R^{2}+4LR^{1+\alpha}\right)}.
\]
On the other hand, we know that $\nabla^{2}\breve{U}\left(x\right)=\nabla^{2}\hat{U}\left(x\right)-\left(L_{N}+\frac{\lambda_{0}}{2}\right)I\succeq-LI$
for since $\hat{U}\left(x\right)$ is $\frac{\lambda_{0}}{2}-$strongly
convex, which implies that $\nabla^{2}\breve{U}\left(x\right)$ is
lower bounded by $-LI$. In addition, for $\left\Vert x\right\Vert >R+2\epsilon+\delta$
from $2-$dissipative assumption, we have for some $a,$ $b>0,\left\langle \nabla\breve{U}(\mathrm{x}),x\right\rangle \geq a\left\Vert x\right\Vert ^{2}-b$,
while for $\left\Vert x\right\Vert \leq R+2\epsilon+\delta$ 
\begin{align*}
\left\langle \nabla\breve{U}\left(x\right),x\right\rangle  & \geq\left\langle -\nabla\left(\left(\frac{L_{N}}{2}+\frac{\lambda_{0}}{4}\right)\left\Vert x\right\Vert ^{2}\right),x\right\rangle \\
 & \geq-\left(L_{N}+\frac{\lambda_{0}}{2}\right)\left\Vert x\right\Vert ^{2}\\
 & \geq-\left(L_{N}+\frac{\lambda_{0}}{2}\right)R^{2}.\\
 & \geq a\left\Vert x\right\Vert ^{2}-\left(L_{N}+\frac{\lambda_{0}}{2}\right)R^{2}-aR^{2}.
\end{align*}
so for every $x\in\mathbb{R}^{d},$

\[
\left\langle \nabla\breve{U}(\mathrm{x}),x\right\rangle \geq a\left\Vert x\right\Vert ^{2}-\left(b+\left(L_{N}+\frac{\lambda_{0}}{2}\right)R^{2}+aR^{2}\right).
\]
We choose $W=e^{a_{1}\left\Vert x\right\Vert ^{2}}$ and $V=a_{1}\left\Vert x\right\Vert ^{2}$
with $0<a_{1}=\frac{a}{4}$. One sees that $W$ satisfies Lyapunov
inequality 
\begin{align}
\mathcal{L}W & =\left(2a_{1}d+4a_{1}^{2}\left\Vert x\right\Vert ^{2}-2a_{1}\left\langle \nabla U(\mathrm{x}),x\right\rangle \right)W\nonumber \\
 & \leq\left(2a_{1}d+4a_{1}^{2}\left\Vert x\right\Vert ^{2}-2a_{1}a\left\Vert x\right\Vert ^{2}+2a_{1}\left(b+\left(L_{N}+\frac{\lambda_{0}}{2}\right)R^{2}+aR^{2}\right)\right)W\nonumber \\
 & \leq\left(-\frac{a^{2}}{2}\left\Vert x\right\Vert ^{2}+\frac{a}{2}\left(b+\left(L_{N}+\frac{\lambda_{0}}{2}\right)R^{2}+aR^{2}+d\right)\right)W.
\end{align}
By \citep{cattiaux2010note}'s Theorem 1.9, $\breve{U}\left(x\right)$
satisfies a defective log Sobolev. In addition, by Rothaus' lemma,
a defective log-Sobolev inequality together with the $PI(\gamma_{1})$
implies the log-Sobolev inequality with the log Sobolev constant is
$\gamma_{2}=\frac{2}{[A+(B+2)\frac{1}{\gamma_{_{1}}})]}$ where 
\begin{align}
A & =(1-\frac{L}{2})\frac{8}{a^{2}}+\zeta,\\
B & =2\left[\frac{2\left(\left(b+4\left(L_{N}+\frac{\lambda_{0}}{4}\right)R^{2}+aR^{2}\right)+d\right)}{a}+M_{2}\right](1-\frac{L}{2}+\frac{1}{\zeta}).
\end{align}
where $M_{2}=\int\left\Vert x\right\Vert ^{2}e^{-\breve{U}(x)}dx$.
But it is well known from Lemma 10 that $M_{2}=O(d)$, so the log-Sobolev
constant is just $O(d)$. This concludes the proof. \end{proof}

\subsection{Proof of lemma \ref{Prop4.1.1}\label{Proof-of-Prop4.1.1}}

\begin{thm} \label{T1-2-2-2} Suppose $\pi$ is $\gamma-$Poincar\'{e},
$\alpha$-mixture weakly smooth with $\alpha_{N}=1$ and $2-$dissipativity
(i.e.$\left\langle \nabla U(x),x\right\rangle \geq a\left\Vert x\right\Vert ^{2}-b$)
for some $a,b>0$, and for any $x_{0}\sim p_{0}$ with $H(p_{0}|\pi)=C_{0}<\infty$,
the iterates $x_{k}\sim p_{k}$ of LMC~ with step size $\eta\leq1\wedge\frac{1}{4\gamma_{3}}\wedge\left(\frac{\gamma_{3}}{16L^{1+\alpha}}\right)^{\frac{1}{\alpha}}$satisfies
\begin{align}
H(p_{k}|\pi)\le e^{-\gamma_{3}\epsilon k}H(p_{0}|\pi)+\frac{8\eta^{\alpha}D_{3}}{3\gamma_{3}},\label{Eq:Main1-2-1-4-1-2}
\end{align}
where $D_{3}$ is defined as in equation (\ref{eq:D3}) and 
\begin{align}
M_{2} & =\int\left\Vert x\right\Vert ^{2}e^{-\breve{U}(x)}dx=O(d)\\
\zeta & =\sqrt{2\left[\frac{2\left(b+\left(L+\frac{\lambda_{0}}{2}\right)R^{2}+aR^{2}+d\right)}{a}+M_{2}\right]\frac{e^{4\left(2\sum_{i}L_{i}R^{1+\alpha_{i}}+4L_{N}R^{2}+4LR^{1+\alpha}\right)}}{\gamma}}\\
A & =(1-\frac{L}{2})\frac{8}{a^{2}}+\zeta,\\
B & =2\left[\frac{2\left(\left(b+4\left(L+\frac{\lambda_{0}}{4}\right)R^{2}+aR^{2}\right)+d\right)}{a}+M_{2}\right](1-\frac{L}{2}+\frac{1}{\zeta}),\\
\gamma_{3} & =\frac{2\gamma e^{-\left(2\sum_{i}L_{i}R^{1+\alpha_{i}}+4L_{N}R^{2}+4LR^{1+\alpha}\right)}}{A\gamma+(B+2)e^{4\left(2\sum_{i}L_{i}R^{1+\alpha_{i}}+4L_{N}R^{2}+4LR^{1+\alpha}\right)}}.\nonumber 
\end{align}
Then, for any $\epsilon>0$, to achieve $H(p_{k}|\pi)<\epsilon$,
it suffices to run ULA with step size $\eta\le1\wedge\frac{1}{4\gamma_{3}}\wedge\left(\frac{\gamma_{3}}{16L^{1+\alpha}}\right)^{\frac{1}{\alpha}}\wedge\left(\frac{3\epsilon\gamma_{3}}{16D_{3}}\right)^{\frac{1}{\alpha}}$for
$k\ge\frac{1}{\gamma_{3}\eta}\log\frac{2H\left(p_{0}|\pi\right)}{\epsilon}$
iterations.\end{thm}\begin{proof} From Lemma \ref{lemma:hat_U-2-2}
, we can optimize over $\zeta$ and get 
\[
\zeta=\sqrt{2\left[\frac{2\left(b+\left(L+\frac{\lambda_{0}}{2}\right)R^{2}+aR^{2}+d\right)}{a}+M_{2}\right]\frac{1}{\gamma_{1}}}.
\]
By using Holley Stroock perturbation theorem \citep{holley1986logarithmic},
we have $U(x)$ is log-Sobolev on $\mathbb{R}^{d}$ with constant
\[
\gamma_{3}=\frac{2\gamma e^{-\left(2\sum_{i}L_{i}R^{1+\alpha_{i}}+4L_{N}R^{2}+4LR^{1+\alpha}\right)}}{[A\gamma+(B+2)e^{4\left(2\sum_{i}L_{i}R^{1+\alpha_{i}}+4L_{N}R^{2}+4LR^{1+\alpha}\right)})]}.
\]
Applying theorem \ref{T1}, we get the desired result.\end{proof}

\subsection{Proof of lemma \ref{Prop4.1.1}\label{Proof-of-Prop4.1.1-1}}

\begin{lemma} If $U$ satisfies Assumption \ref{A3}, then 
\begin{center}
\begin{equation}
U(x)\geq\frac{a}{2\beta}\Vert x\Vert^{\beta}+U(0)-\sum_{i}\frac{L_{i}}{\alpha_{i}+1}R^{\alpha_{i}+1}-b.
\end{equation}
\par\end{center}

\label{L13} \end{lemma} \begin{proof} Using the technique of \citep{erdogdu2020convergence},
let $R=\left(\frac{2b}{a}\right)^{\frac{1}{\beta}}$, we lower bound
$U\left(x\right)$ when $\left\Vert x\right\Vert \leq R$,

\begin{align}
U(x) & =U(0)+\int_{0}^{1}\left\langle \nabla U(tx),\ x\right\rangle dt\nonumber \\
 & \geq U(0)-\int_{0}^{1}\left\Vert \nabla U(tx)\right\Vert \left\Vert x\right\Vert dt\nonumber \\
 & \geq U(0)-\sum_{i}L_{i}\left\Vert x\right\Vert ^{\alpha_{i}+1}\int_{0}^{1}t^{\alpha_{i}}dt\nonumber \\
 & \geq U(0)-\sum_{i}\frac{L_{i}}{\alpha_{i}+1}\left\Vert x\right\Vert ^{\alpha_{i}+1}\nonumber \\
 & \geq U(0)-\sum_{i}\frac{L_{i}}{\alpha_{i}+1}R^{\alpha_{i}+1}
\end{align}
For $\left\Vert x\right\Vert >R$, we can lower bound $U$ as follows.
\begin{align}
U(x) & =U(0)+\int_{0}^{\frac{R}{\Vert x\Vert}}\left\langle \nabla U(tx),\ x\right\rangle dt+\int_{\frac{R}{\Vert x\Vert}}^{1}\left\langle \nabla U(tx),\ x\right\rangle dt\nonumber \\
 & \geq U(0)-\int_{0}^{\frac{R}{\left\Vert x\right\Vert }}\left\Vert \nabla U(tx)\right\Vert \left\Vert x\right\Vert dt+\int_{\frac{R}{\left\Vert x\right\Vert }}^{1}\frac{1}{t}\left\langle \nabla U(tx),\ tx\right\rangle dt\nonumber \\
 & \geq U(0)-\left\Vert x\right\Vert \int_{0}^{\frac{R}{\left\Vert x\right\Vert }}\sum_{i}L_{i}\left\Vert tx\right\Vert ^{\alpha_{i}}dt+\int_{\frac{R}{\left\Vert x\right\Vert }}^{1}\frac{1}{t}\left(a\left\Vert tx\right\Vert ^{\beta}-b\right)dt\nonumber \\
 & \stackrel{_{1}}{\geq}U(0)-\sum_{i}L_{i}\left\Vert x\right\Vert ^{\alpha_{i}+1}\int_{0}^{\frac{R}{\left\Vert x\right\Vert }}t^{\alpha_{i}}dt\ +\frac{1}{2}\int_{\frac{R}{\left\Vert x\right\Vert }}^{1}\frac{1}{t}a\left\Vert tx\right\Vert ^{\beta}dt\nonumber \\
 & \stackrel{_{2}}{\geq}U(0)-\sum_{i}\frac{L_{i}}{\alpha_{i}+1}\left\Vert x\right\Vert ^{\alpha_{i}+1}\frac{R^{\alpha_{i}+1}}{\left\Vert x\right\Vert ^{\alpha_{i}+1}}+\frac{a}{2}\left\Vert x\right\Vert ^{\beta}\int_{\frac{R}{\left\Vert x\right\Vert }}^{1}t^{\beta-1}dt\nonumber \\
 & \geq U(0)-\sum_{i}\frac{L_{i}}{\alpha_{i}+1}R^{\alpha_{i}+1}+\frac{a}{2\beta}\left\Vert x\right\Vert ^{\beta}\left(1-\frac{R^{\beta}}{\left\Vert x\right\Vert ^{\beta}}\right)\nonumber \\
 & \geq\frac{a}{2\beta}\left\Vert x\right\Vert ^{\beta}+U(0)-\sum_{i}\frac{L_{i}}{\alpha_{i}+1}R^{\alpha_{i}+1}-b,
\end{align}
where $1$ follows from Assumption \ref{A3} and $2$ uses the fact
that if $t{\displaystyle \geq\frac{R}{\left\Vert x\right\Vert }}$
then ${\displaystyle a\left\Vert tx\right\Vert ^{\beta}-b\geq\frac{a}{2}\left\Vert tx\right\Vert ^{\beta}}.$
Now, since for $\left\Vert x\right\Vert \leq R$, $\frac{a}{2\beta}\left\Vert x\right\Vert ^{\beta}\leq b$,
we combine the inequality for $\left\Vert x\right\Vert \leq R$ and
get 
\begin{equation}
U(x)\geq\frac{a}{2\beta}\left\Vert x\right\Vert ^{\beta}+U(0)-\sum_{i}\frac{L_{i}}{\alpha_{i}+1}R^{\alpha_{i}+1}-b.
\end{equation}
\end{proof}

\subsection{Proof of Lemma 5}

\begin{lemma} Assume that $U$ satisfies Assumption \ref{A3}, then
for $\pi=e^{-U}$ and any distribution $\rho$, we have 
\begin{center}
\begin{equation}
\frac{4\beta}{a}\left[\mathrm{H}(\rho|\pi)+\tilde{d}+\tilde{\mu}\right]\geq\mathrm{E}_{\rho}\left[\left\Vert x\right\Vert {}^{\beta}\right],
\end{equation}
\par\end{center}

where 
\begin{align}
\tilde{\mu} & =\frac{1}{2}\log(\frac{2}{\beta})+\sum_{i}\frac{L_{i}}{\alpha_{i}+1}\left(\frac{2b}{a}\right)^{\frac{\alpha_{i}+1}{\beta}}+b+|U(0)|,\\
\tilde{d} & =\frac{d}{\beta}\left[\frac{\beta}{2}log\left(\pi\right)+\log\left(\frac{4\beta}{a}\right)+(1-\frac{\beta}{2})\log(\frac{d}{2e})\right].
\end{align}
\end{lemma} \begin{proof} Let $q(x)=e^{\frac{a}{4\beta}\left\Vert x\right\Vert {}^{\beta}-U(x)}$
and $C_{q}=\int e^{\frac{a}{4\beta}\left\Vert x\right\Vert {}^{\beta}-U(x)}dx$.
First, we need to bound $\log C_{q}$. Using Lemma \ref{L13}, we
have 
\begin{align}
U(x) & \geq\frac{a}{2\beta}\left\Vert x\right\Vert ^{\beta}+U(0)-\sum_{i}\frac{L_{i}}{\alpha_{i}+1}\left(\frac{2b}{a}\right)^{\frac{\alpha_{i}+1}{\beta}}-b.
\end{align}
Regrouping the terms and integrating both sides gives 
\begin{align}
 & \int e^{\frac{a}{4\beta}\left\Vert x\right\Vert {}^{\beta}-U(x)}dx\leq e^{-U(0)+\sum_{i}\frac{L_{i}}{\alpha_{i}+1}\left(\frac{2b}{a}\right)^{\frac{\alpha_{i}+1}{\beta}}+b}\int e^{-\frac{a}{4\beta}\left\Vert x\right\Vert {}^{\beta}}dx\nonumber \\
 & =\frac{2\pi^{d/2}}{\beta}\left(\frac{4\beta}{a}\right)^{\frac{d}{\beta}}e^{-U(0)+\sum_{i}\frac{L_{i}}{\alpha_{i}+1}\left(\frac{2b}{a}\right)^{\frac{\alpha_{i}+1}{\beta}}+b}\frac{\Gamma\left(\frac{d}{\beta}\right)}{\Gamma\left(\frac{d}{2}\right)}\nonumber \\
 & \leq\frac{2\pi^{d/2}}{\beta}\left(\frac{4\beta}{a}\right)^{\frac{d}{\beta}}\frac{\left(\frac{d}{\beta}\right)^{\frac{d}{\beta}-\frac{1}{2}}}{\left(\frac{d}{2}\right)^{\frac{d}{2}-\frac{1}{2}}}e^{\frac{d}{2}-\frac{d}{\beta}}e^{-U(0)+\sum_{i}\frac{L_{i}}{\alpha_{i}+1}\left(\frac{2b}{a}\right)^{\frac{\alpha_{i}+1}{\beta}}+b},
\end{align}
where the equality on the second line comes from using polar coordinates
and the third line follows from an inequality for the ratio of Gamma
functions \citep{kevckic1971some}. Plugging this back into the previous
inequality and taking logs, we deduce 
\begin{align}
{\displaystyle \log(C_{q})} & ={\displaystyle \log(\int e^{\frac{a}{4\beta}\left\Vert x\right\Vert {}^{\beta}-U(x)}dx)}\nonumber \\
 & \leq\frac{d}{2}\log(\pi)+\frac{d}{\beta}\log\left(\frac{4\beta}{a}\right)+(\frac{d}{\beta}-\frac{d}{2})\log(\frac{d}{2e})\nonumber \\
 & +(\frac{d}{\beta}+\frac{1}{2})\log(\frac{2}{\beta})+\sum_{i}\frac{L_{i}}{\alpha_{i}+1}\left(\frac{2b}{a}\right)^{\frac{\alpha_{i}+1}{\beta}}+b+|U(0)|\nonumber \\
 & \leq\frac{d}{\beta}\left[\frac{\beta}{2}log\left(\pi\right)+\log\left(\frac{4\beta}{a}\right)+(1-\frac{\beta}{2})\log(\frac{d}{2e})\right]\nonumber \\
 & +\frac{1}{2}\log(\frac{2}{\beta})+\sum_{i}\frac{L_{i}}{\alpha_{i}+1}\left(\frac{2b}{a}\right)^{\frac{\alpha_{i}+1}{\beta}}+b+|U(0)|\nonumber \\
 & \leq\tilde{d}+\tilde{\mu,}
\end{align}
as definitions of $\tilde{d}$ and $\tilde{\mu}$. Using this bound
on $\log C_{q}$ we get 
\begin{align}
\mathrm{H}(\rho|\pi) & =\int\rho\log\frac{\rho}{q/C_{q}}+\int\rho\log\frac{q/C_{q}}{\pi}\nonumber \\
 & =\mathrm{H}(\rho|q/C_{q})+\mathrm{E}_{\rho}\left[\log\frac{q/C_{q}}{e^{-U}}\right]\nonumber \\
 & \stackrel{_{\left(1\right)}}{\geq}\frac{a}{4\beta}\mathrm{E}_{\rho}\left[\left\Vert x\right\Vert {}^{\beta}\right]-\log\left(C_{q}\right)\\
 & \geq\frac{a}{4\beta}\mathrm{E}_{\rho}\left[\left\Vert x\right\Vert {}^{\beta}\right]-\tilde{d}-\tilde{\mu,}
\end{align}
where $\left(1\right)$ follows from definition of $C_{q}$ and the
fact that relative information is always non-negative. Rearranging
the terms completes the proof. \end{proof}

\begin{thm} \label{Prop4.1.1-1-1} Suppose $\pi$ is non-strongly
convex outside the ball $\mathbb{B}(0,R)$, $\alpha$-mixture weakly
smooth with $\alpha_{N}=1$ and $2-$dissipativity (i.e.$\left\langle \nabla U(x),x\right\rangle \geq a\left\Vert x\right\Vert ^{2}-b$)
for some $a,b>0$, and for any $x_{0}\sim p_{0}$ with $H(p_{0}|\pi)=C_{0}<\infty$,
the iterates $x_{k}\sim p_{k}$ of LMC~ with step size $\eta\leq1\wedge\frac{1}{4\gamma_{3}}\wedge\left(\frac{\gamma_{3}}{16L^{1+\alpha}}\right)^{\frac{1}{\alpha}}$satisfies
\begin{align}
H(p_{k}|\pi)\le e^{-\gamma_{3}\epsilon k}H(p_{0}|\pi)+\frac{8\eta^{\alpha}D_{3}}{3\gamma_{3}},\label{Eq:Main1-2-1-4-1-3-1}
\end{align}
where $D_{3}$ is defined as in equation (\ref{eq:D3}) and for some
universal constant $K$, 
\begin{align}
M_{2} & =\int\left\Vert x\right\Vert ^{2}e^{-\breve{U}(x)}dx=O(d)\\
\zeta & =K\sqrt{64d\left[\frac{2\left(b+\left(L+\frac{\lambda_{0}}{2}\right)R^{2}+aR^{2}+d\right)}{a}+M_{2}\right]\left(\frac{a+b+2aR^{2}+3}{ae^{-4\left(4L_{N}R^{2}+4LR^{1+\alpha}\right)}}\right)}\\
A & =(1-\frac{L}{2})\frac{8}{a^{2}}+\zeta,\\
B & =2\left[\frac{2\left(\left(b+4\left(L+\frac{\lambda_{0}}{4}\right)R^{2}+aR^{2}\right)+d\right)}{a}+M_{2}\right](1-\frac{L}{2}+\frac{1}{\zeta}),\\
\gamma_{3} & =\frac{2e^{-\left(2\sum_{i}L_{i}R^{1+\alpha_{i}}+4L_{N}R^{2}+4LR^{1+\alpha}\right)}}{A+(B+2)32K^{2}d\left(\frac{a+b+2aR^{2}+3}{a}\right)e^{4\left(4L_{N}R^{2}+4LR^{1+\alpha}\right)}}=\frac{1}{O(d)}.\nonumber 
\end{align}
Then, for any $\epsilon>0$, to achieve $H(p_{k}|\pi)<\epsilon$,
it suffices to run ULA with step size $\eta\le1\wedge\frac{1}{4\gamma_{3}}\wedge\left(\frac{\gamma_{3}}{16L^{1+\alpha}}\right)^{\frac{1}{\alpha}}\wedge\left(\frac{3\epsilon\gamma_{3}}{16D_{3}}\right)^{\frac{1}{\alpha}}$for
$k\ge\frac{1}{\gamma_{3}\eta}\log\frac{2H\left(p_{0}|\pi\right)}{\epsilon}$
iterations.\end{thm}\begin{proof} Using Lemma 2, there exists $\breve{U}\left(x\right)\in C^{1}(R^{d})$
with its Hessian exists everywhere on $R^{d}$, and $\breve{U}$ is
convex on $R^{d}$ such that 
\begin{equation}
\sup\left(\breve{U}(\ x)-U(\ x)\right)-\inf\left(\breve{U}(\ x)-U(\ x)\right)\leq2\sum_{i}L_{i}R^{1+\alpha_{i}}.
\end{equation}
We can prove by two different approaches.

First approach: Since $\breve{U}$ is convex, by Theorem 1.2 of \citep{bobkov1999isoperimetric},
$\breve{U}$ satisfies Poincar\'{e} inequality with constant 
\begin{align*}
\gamma & \geq\frac{1}{4K^{2}\int\left\Vert x-E_{\pi}(x)\right\Vert ^{2}\pi\left(x\right)dx}\\
 & \stackrel{_{1}}{\geq}\frac{1}{8K^{2}\left(E_{\pi}\left(\left\Vert x\right\Vert ^{2}\right)+\left\Vert E_{\pi}(x)\right\Vert ^{2}\right)}\\
 & \stackrel{}{\geq}\frac{1}{16K^{2}E_{\pi}\left(\left\Vert x\right\Vert ^{2}\right)},
\end{align*}
where $K$ is a universal constant, step $1$ follows from Young inequality
and the last line is due to Jensen inequality. In addition, for $\left\Vert x\right\Vert >R+2\epsilon+\delta$
from $2-$dissipative assumption, we have for some $a,$ $b>0,\left\langle \nabla\breve{U}(x),x\right\rangle =\left\langle \nabla U(x),x\right\rangle \geq a\left\Vert x\right\Vert ^{2}-b$,
while for $\left\Vert x\right\Vert \leq R+2\epsilon+\delta$ by convexity
of $\breve{U}$ 
\begin{align*}
\left\langle \nabla\breve{U}(x),x\right\rangle  & \geq0\\
 & \geq a\left\Vert x\right\Vert ^{2}-a\left(R+2\epsilon+\delta\right)^{2}\\
 & \geq a\left\Vert x\right\Vert ^{2}-2aR^{2}.
\end{align*}
so for every $x\in\mathbb{R}^{d},$

\[
\left\langle \nabla\breve{U}(x),x\right\rangle \geq a\left\Vert x\right\Vert ^{2}-\left(b+2aR^{2}\right).
\]

Therefore, $\breve{U}(\mathrm{x})$ also satisfies $2-$dissipative,
which implies 
\[
E_{\breve{\pi}}\left(\left\Vert x\right\Vert ^{2}\right)\leq2d\left(\frac{a+b+2aR^{2}+3}{a}\right),
\]
so the Poincar\'{e} constant satisfies

\[
\gamma\stackrel{}{\geq}\frac{1}{32K^{2}d\left(\frac{a+b+2aR^{2}+3}{a}\right)}.
\]
From \citep{ledoux2001logarithmic}'s Lemma 1.2, we have $U$ satisfies
Poincar\'{e} inequality with constant 
\[
\gamma\geq\frac{1}{32K^{2}d\left(\frac{a+b+2aR^{2}+3}{a}\right)}e^{-4\left(2\sum_{i}L_{i}R^{1+\alpha_{i}}\right)}.
\]
Now, applying previous section result, we derive the desired result.

Second approach. By employing Lemma \ref{L4}, combined with $2-$dissipative
assumption, we get: 
\begin{equation}
\int e^{\frac{a}{8}\left\Vert x\right\Vert {}^{2}-U(x)}dx\leq e^{\left(\tilde{d}+\tilde{\mu}\right)},
\end{equation}
which in turn implies 
\begin{equation}
\int e^{\frac{a}{8}\left\Vert x\right\Vert {}^{2}-\breve{U}(x)}dx\leq e^{\left(\tilde{d}+\tilde{\mu}\right)+2\sum_{i}L_{i}R^{1+\alpha_{i}}}.
\end{equation}
Let $\mu_{1}=\frac{e^{\frac{-a}{16p}\left\Vert x\right\Vert {}^{2}-\breve{U}(x)}}{\int e^{\frac{-a}{16p}\left\Vert x\right\Vert {}^{2}-\breve{U}(x)}dx}$
and assume that $\mu_{2}=\frac{\mu_{1}e^{\frac{a}{16p}\left\Vert x\right\Vert {}^{2}}}{\int e^{\frac{a}{16p}\left\Vert x\right\Vert {}^{2}}d\mu_{1}}$.
We have $\mu_{1}$ is $\frac{a}{8p}$ strongly convex or log Sobolev
with constant $\frac{a}{8p}$ and by Cauchy Schwarz inequality, we
have 
\begin{align}
\left\Vert \frac{d\mu_{2}}{d\mu_{1}}\right\Vert _{L^{p}\left(\mu_{1}\right)}^{p} & =\frac{\int e^{\frac{a}{16}\left\Vert x\right\Vert {}^{2}}d\mu_{1}}{\left(\int e^{\frac{a}{16p}\left\Vert x\right\Vert {}^{2}}d\mu_{1}\right)^{p}}\nonumber \\
 & \leq\left(\int e^{\frac{a}{8}\left\Vert x\right\Vert {}^{2}}d\mu_{1}\right)^{\frac{1}{2}}\left(\int e^{\frac{-a}{16p}\left\Vert x\right\Vert {}^{2}}d\mu_{1}\right)^{p}\nonumber \\
 & =\left(\frac{\int e^{\frac{a\left(2p-1\right)}{16p}\left\Vert x\right\Vert {}^{2}-\breve{U}(x)}dx}{\int e^{\frac{-a}{16p}\left\Vert x\right\Vert {}^{2}-\breve{U}(x)}dx}\right)^{\frac{1}{2}}\left(\frac{\int e^{\frac{-a}{8p}\left\Vert x\right\Vert {}^{2}-\breve{U}(x)}dx}{\int e^{\frac{-a}{16p}\left\Vert x\right\Vert {}^{2}-\breve{U}(x)}dx}\right)^{p}
\end{align}
Since 
\begin{align}
\left|U(\ x)-U(0)\right| & =\left|U(\ x)-U(0)-\ \left\langle x,\nabla U(0)\right\rangle \right|\nonumber \\
 & \leq\sum_{i<N}\frac{L_{i}}{1+\alpha_{i}}\left\Vert x\right\Vert ^{1+\alpha_{i}}+\frac{L_{N}}{2}\left\Vert x\right\Vert ^{2}
\end{align}
this implies $U(\ x)\leq\left|U(0)\right|+\sum_{i<N}\frac{L_{i}}{1+\alpha_{i}}\left\Vert x\right\Vert ^{1+\alpha_{i}}+\frac{L_{N}}{2}\left\Vert x\right\Vert ^{2}$
which in turn indicates 
\begin{align}
\int e^{\frac{-a}{16p}\left\Vert x\right\Vert {}^{2}-\breve{U}(x)}dx & \geq\int e^{\frac{-a}{16p}\left\Vert x\right\Vert {}^{2}-\left|U(0)\right|-\sum_{i<N}\frac{L_{i}}{1+\alpha_{i}}\left\Vert x\right\Vert ^{1+\alpha_{i}}-\frac{L_{N}}{2}\left\Vert x\right\Vert ^{2}-2\sum_{i}L_{i}R^{1+\alpha_{i}}}dx\nonumber \\
 & \geq e^{-\left|U(0)\right|-2\sum_{i}L_{i}R^{1+\alpha_{i}}}\int e^{\frac{-a}{16p}\left\Vert x\right\Vert {}^{2}-\sum_{i<N}\frac{L_{i}}{1+\alpha_{i}}\left\Vert x\right\Vert ^{1+\alpha_{i}}-\frac{L_{N}}{2}\left\Vert x\right\Vert ^{2}}dx\nonumber \\
 & \geq e^{-\left|U(0)\right|-2\sum_{i}L_{i}R^{1+\alpha_{i}}-\sum_{i<N}\frac{L_{i}}{1+\alpha_{i}}}\int e^{-\left(\frac{a}{16p}+\sum_{i<N}\frac{L_{i}}{1+\alpha_{i}}+\frac{L_{N}}{2}\right)\left\Vert x\right\Vert {}^{2}}dx\nonumber \\
 & \geq\frac{\pi^{\frac{d}{2}}}{\left(\frac{a}{16p}+\sum_{i<N}\frac{L_{i}}{1+\alpha_{i}}+\frac{L_{N}}{2}\right)^{\frac{d}{2}}}e^{-\left|U(0)\right|-2\sum_{i}L_{i}R^{1+\alpha_{i}}-\frac{L}{1+\alpha}}.
\end{align}
On the other hand, 
\begin{align}
\int e^{\frac{-a}{8p}\left\Vert x\right\Vert {}^{2}-\breve{U}(x)}dx & \leq\int e^{\frac{a\left(2p-1\right)}{16p}\left\Vert x\right\Vert {}^{2}-\breve{U}(x)}dx\nonumber \\
 & \leq\int e^{\frac{a}{8p}\left\Vert x\right\Vert {}^{2}-\breve{U}(x)}dx\nonumber \\
 & \leq e^{\left(\tilde{d}+\tilde{\mu}\right)+2\sum_{i}L_{i}R^{1+\alpha_{i}}}.
\end{align}
Combining this with previous inequality, we obtain 
\begin{align}
\left\Vert \frac{d\mu_{2}}{d\mu_{1}}\right\Vert _{L^{p}\left(\mu_{1}\right)}^{p} & \leq\left(\frac{e^{\left(\left(\tilde{d}+\tilde{\mu}\right)+2\sum_{i}L_{i}R^{1+\alpha_{i}}\right)}}{\frac{\pi^{\frac{d}{2}}}{\left(\frac{a}{16p}+\sum_{i<N}\frac{L_{i}}{1+\alpha_{i}}+\frac{L_{N}}{2}\right)^{\frac{d}{2}}}e^{-\left|U(0)\right|-2\sum_{i}L_{i}R^{1+\alpha_{i}}-\sum_{i<N}\frac{L_{i}}{1+\alpha_{i}}}}\right)^{p+\frac{1}{2}}\nonumber \\
 & =\Lambda^{p}.
\end{align}
Taking logarithm of $\Lambda$ we get 
\begin{align}
\log\Lambda & =\frac{\left(p+\frac{1}{2}\right)}{p}\log\left(\frac{e^{\left(\left(\tilde{d}+\tilde{\mu}\right)+2\sum_{i}L_{i}R^{1+\alpha_{i}}\right)}}{\frac{\pi^{\frac{d}{2}}}{\left(\frac{a}{16p}+\sum_{i<N}\frac{L_{i}}{1+\alpha_{i}}+\frac{L_{N}}{2}\right)^{\frac{d}{2}}}e^{-\left|U(0)\right|-2\sum_{i}L_{i}R^{1+\alpha_{i}}-\sum_{i<N}\frac{L_{i}}{1+\alpha_{i}}}}\right)\nonumber \\
 & =\frac{\left(p+\frac{1}{2}\right)}{p}\left(\tilde{d}+\frac{d}{2}\log\left(\frac{a}{8p}+\frac{a}{16p}+\sum_{i<N}\frac{L_{i}}{1+\alpha_{i}}+\frac{L_{N}}{2}\right)-\frac{d}{2}\log\left(\pi\right)\right)\nonumber \\
 & +\frac{\left(p+\frac{1}{2}\right)}{p}\left(\tilde{\mu}+2\sum_{i}L_{i}R^{1+\alpha_{i}}+\left|U(0)\right|+\sum_{i<N}\frac{L_{i}}{1+\alpha_{i}}\right)\nonumber \\
 & =\tilde{O}\left(d\right).
\end{align}
Since $\mu_{2}$ is log concave, from Lemma 9, we have for some universal
constant $C$ (not depending on $d$), it is log Sobolev with constant
\begin{align}
C({\displaystyle \Lambda,p)} & =\frac{1}{C}\frac{a}{8p}\frac{p-1}{p}\frac{1}{1+\log\Lambda}\nonumber \\
 & =\frac{1}{C}\frac{a}{8p}\frac{p-1}{p}\frac{1}{1+\tilde{O}\left(d\right)}\nonumber \\
 & =\frac{1}{\tilde{O}\left(d\right)}.
\end{align}
From this, by using Holley-Stroock perturbation theorem, we obtain
$U(\ x)$ is log Sobolev on $R^{d}$ with constant $\frac{1}{\tilde{O}\left(d\right)}e^{-2\sum_{i}L_{i}R^{1+\alpha_{i}}}.$
Now, applying theorem \ref{T1}, we derive the desired result.

\end{proof}

\section{Proof of additional lemmas}

\setcounter{lemma}{12} \begin{lemma} \label{L21} For any $0\leq\varpi\leq k\in N^{+}$,
we have 
\begin{equation}
\Vert x+y\Vert^{\varpi}\leq2^{k-1}(\Vert x\Vert^{\varpi}+\Vert y\Vert^{\varpi}).
\end{equation}
\end{lemma} \begin{proof} Let's consider functions $f_{k}(u)=2^{k-1}(u^{\varpi}+1)-(1+u)^{\varpi}$.
We prove $f_{k}(u)\geq0$ for every $u\geq0$ by induction. For $k=1$,
since $0\leq\varpi\leq1,$ we have $f_{1}^{\prime}(u)=\varpi u^{\varpi-1}-\varpi(1+u)^{\varpi-1}\geq0$.
This implies $f_{1}(u)$ increases on $\left[0,\infty\right]$ and
since $f(0)=0,$ which in turn indicates $f(u)\geq0.$ Therefore,
the statement is true for $k=1.$\\
 Assume that it is true for $k=n$, we will show that it is also true
for $k=n+1.$ If we differentiate $f_{n+1}(u)$ we get 
\begin{align}
f_{n+1}^{\prime}(u) & =2^{n}\varpi u^{\varpi-1}-\varpi(1+u)^{\varpi-1}\nonumber \\
 & =\varpi\left(2^{n}u^{\varpi-1}-(1+u)^{\varpi-1}\right)\nonumber \\
 & \geq0,
\end{align}
for $1\leq\varpi\leq n+1$ by induction assumption while for $0\leq\varpi\leq1,$
$u^{\varpi-1}-(1+u)^{\varpi-1}\geq u^{\varpi-1}-(1+u)^{\varpi-1}\geq0.$
Hence, $f$ increases on $\left[0,\infty\right]$ and since $f(0)=2^{k-1}-1\geq0,$
this implies $f\geq0$.\\
 Applying to our case for $0\leq\varpi\leq k$, 
\begin{align}
2^{k-1}(\Vert x\Vert^{\varpi}+\Vert y\Vert^{\varpi}) & =\Vert x\Vert^{\varpi}2^{k-1}\left(1+\left(\frac{\left\Vert y\right\Vert }{\left\Vert x\right\Vert }\right)^{\omega}\right)\nonumber \\
 & \geq\Vert x\Vert^{\varpi}\left(1+\left(\frac{\left\Vert y\right\Vert }{\left\Vert x\right\Vert }\right)\right)^{\varpi}\nonumber \\
 & =\left(\left\Vert x\right\Vert +\left\Vert y\right\Vert \right)^{\varpi}\nonumber \\
 & \geq\left(\left\Vert x+y\right\Vert \right)^{\varpi},
\end{align}
which conclude our proof. \end{proof}

\begin{lemma} \label{L22} For $\theta>0,$ $f\left(r\right)=m\left(r\right)r^{2}=\mu\left(1+r^{2}\right)^{-\frac{\theta}{2}}r{}^{2}\geq\frac{\mu}{2}r{}^{2-\theta}-\frac{\mu}{2}2{}^{\frac{2-\theta}{\theta}},$
and for $\theta=0,$ $f\left(r\right)=\mu r{}^{2}.$ \end{lemma}
\begin{proof} For $\theta=0,$ it is straightforward. For $\theta>0,$
from Lemma 2 above, for $r\geq2^{\frac{1}{\theta}}$,

\begin{align}
f\left(r\right) & =\mu\left(1+r^{2}\right)^{-\frac{\theta}{2}}r{}^{2}\nonumber \\
 & \geq\mu\left(1+r^{\theta}\right)^{-1}r{}^{2}\nonumber \\
 & =\mu\left(r^{2\theta}-1\right)^{-1}r{}^{2}\left(r^{\theta}-1\right)\nonumber \\
 & \geq\mu r{}^{2-2\theta}\left(r^{\theta}-1\right)\nonumber \\
 & \geq\frac{\mu}{2}r{}^{2-\theta}.
\end{align}

For $r<2^{\frac{1}{\theta}}$, $f\left(r\right)\geq0\geq\frac{\mu}{2}r{}^{2-\theta}-\frac{\mu}{2}2{}^{\frac{2-\theta}{\theta}}$
which concludes statement. \end{proof} \begin{lemma} \label{L24}
$f\left(\theta\right)=\left|\left(1+\left\Vert x\right\Vert ^{2}\right)^{\frac{\theta}{2}}-\left(1+\left\Vert x-y\right\Vert ^{2}\right)^{\frac{\theta}{2}}\right|$is
increasing function. \end{lemma} \begin{proof} If $\left\Vert x\right\Vert \geq\left\Vert x-y\right\Vert ,$
we have $f\left(\theta\right)=\left(1+\left\Vert x\right\Vert ^{2}\right)^{\frac{\theta}{2}}-\left(1+\left\Vert x-y\right\Vert ^{2}\right)^{\frac{\theta}{2}}$.
Differentiate $f$ with respect to $\theta$ gives 
\begin{align}
f^{\prime}\left(\theta\right) & =\frac{1}{2}\ln\left(1+\left\Vert x\right\Vert ^{2}\right)\left(1+\left\Vert x\right\Vert ^{2}\right)^{\frac{\theta}{2}}\nonumber \\
 & -\frac{1}{2}\ln\left(1+\left\Vert x-y\right\Vert ^{2}\right)\left(1+\left\Vert x-y\right\Vert ^{2}\right)^{\frac{\theta}{2}}\nonumber \\
 & \geq0
\end{align}
Similarly, if $\left\Vert x\right\Vert \leq\left\Vert x-y\right\Vert $
we also obtain $f^{\prime}\left(\theta\right)\geq0,$ which implies
that $f$ increases as desired. \end{proof}

\begin{lemma}\label{L4} If $\xi\sim N_{p}\left(0,I_{d}\right)$
then $d^{\left\lfloor \frac{n}{p}\right\rfloor }\leq E(\left\Vert \xi\right\Vert _{p}^{n})\leq\left[d+\frac{n}{2}\right]^{\frac{n}{p}}$where$\left\lfloor x\right\rfloor $
denotes the largest integer less than or equal to $x.$ If $n=kp,$
then $E(\left\Vert \xi\right\Vert _{p}^{n})=d..(d+k-1)$. \end{lemma}
\begin{proof} From \citep{richter2007generalized}, we have $E(\left\Vert \xi\right\Vert _{p}^{n})=p^{\frac{n}{p}}\frac{\Gamma\left(\frac{d+n}{p}\right)}{\Gamma\left(\frac{d}{p}\right)}.$

Since $\Gamma$ is an inscreasing function, 
\[
p^{\frac{n}{p}}\frac{\Gamma\left(\frac{d+n}{p}\right)}{\Gamma\left(\frac{d}{p}\right)}\geq p^{\frac{n}{p}}\frac{\Gamma\left(\frac{d}{p}+\left\lfloor \frac{n}{p}\right\rfloor \right)}{\Gamma\left(\frac{d}{p}\right)}=p^{\frac{n}{p}}\frac{d}{p}\ldots\left(\frac{d}{p}+k-1\right)\geq d^{\left\lfloor \frac{n}{p}\right\rfloor }.
\]

If $n=kp$ for $k\in N$ then $E(\left\Vert \xi\right\Vert _{p}^{n})=p^{\frac{n}{p}}\frac{d}{p}\ldots\left(\frac{d}{p}+k-1\right).$
If $n\neq kp$, let $\left\lfloor \frac{n}{p}\right\rfloor =k$. Since
$\Gamma$ is log-convex, by Jensen's inequality for any $p\geq1$,
we acquire 
\begin{align*}
 & \left(1-\frac{n}{p\left\lfloor \frac{n}{p}\right\rfloor +p}\right)\log\Gamma\left(\frac{d}{p}\right)+\frac{n}{p\left\lfloor \frac{n}{p}\right\rfloor +p}\log\Gamma\left(\frac{d}{p}+\left\lfloor \frac{n}{p}\right\rfloor +1\right)\\
 & \geq\log\Gamma\left(\left(1-\frac{n}{p\left\lfloor \frac{n}{p}\right\rfloor +p}\right)\frac{d}{p}+\frac{n}{p\left\lfloor \frac{n}{p}\right\rfloor +p}\left(\frac{d}{p}+\left\lfloor \frac{n}{p}\right\rfloor +1\right)\right)\\
 & \geq\log\Gamma\left(\frac{d+n}{p}\right)>0.
\end{align*}
Raising $e$ to the power of both sides, we get 
\[
\Gamma\left(\frac{d}{p}\right)^{\left(1-\frac{n}{p\left\lfloor \frac{n}{p}\right\rfloor +p}\right)}\Gamma\left(\frac{d}{p}+\left\lfloor \frac{n}{p}\right\rfloor +1\right)^{\frac{n}{p\left\lfloor \frac{n}{p}\right\rfloor +p}}\geq\Gamma\left(\frac{d+n}{p}\right),
\]
which implies that 
\[
\begin{array}{cc}
\left[\frac{\Gamma\left(\frac{d}{p}+\left\lfloor \frac{n}{p}\right\rfloor +1\right)}{\Gamma\left(\frac{d}{p}\right)}\right]^{\frac{n}{p\left\lfloor \frac{n}{p}\right\rfloor +p}} & \geq\frac{\Gamma\left(\frac{d+n}{p}\right)}{\Gamma\left(\frac{d}{p}\right)}\\
\left[\frac{d}{p}\ldots\left(\frac{d}{p}+\left\lfloor \frac{n}{p}\right\rfloor \right)\right]^{\frac{n}{p\left\lfloor \frac{n}{p}\right\rfloor +p}} & \geq\frac{\Gamma\left(\frac{d+n}{p}\right)}{\Gamma\left(\frac{d}{p}\right)}.
\end{array}
\]
Combining with $E(\left\Vert \xi\right\Vert _{p}^{n})=p^{\frac{n}{p}}\frac{\Gamma\left(\frac{d+n}{p}\right)}{\Gamma\left(\frac{d}{p}\right)}$
gives the conclusion. \end{proof}

\section*{Acknowledgements}
This research was funded in part by University of Mississippi summer grant.

\bibliographystyle{imsart-number} 

\end{document}